\pdfoutput=1
\RequirePackage{ifpdf}
\ifpdf 
\documentclass[pdftex]{sigma}
\else
\documentclass{sigma}
\fi

\numberwithin{equation}{section}

\newtheorem{Theorem}{Theorem}[section]
\newtheorem{Corollary}[Theorem]{Corollary}

{ \theoremstyle{definition}
\newtheorem{Definition}[Theorem]{Definition}
\newtheorem{Example}[Theorem]{Example}
\newtheorem{Remark}[Theorem]{Remark} }

\usepackage{tikz}
\newcommand{\R}{\mathbb{R}}
\newcommand{\abs}[1]{\left\lvert #1 \right\rvert}
\newcommand{\order}[1]{\mathcal{O}\left( #1 \right)}
\DeclareMathOperator{\sgn}{sgn}
\DeclareMathOperator{\artanh}{artanh}
\DeclareMathOperator{\Ei}{Ei}

\DeclareMathOperator{\imag}{Im}
\begin{document}
\allowdisplaybreaks

\newcommand{\arXivNumber}{1807.01910}

\renewcommand{\PaperNumber}{017}

\FirstPageHeading

\ShortArticleName{Ghostpeakons and Characteristic Curves for the CH, DP and Novikov Equations}

\ArticleName{Ghostpeakons and Characteristic Curves\\ for the Camassa--Holm, Degasperis--Procesi\\ and Novikov Equations}

\Author{Hans LUNDMARK and Budor SHUAIB}
\AuthorNameForHeading{H.~Lundmark and B.~Shuaib}
\Address{Department of Mathematics, Link\"oping University, SE-581 83 Link\"oping, Sweden}
\Email{\href{mailto:hans.lundmark@liu.se}{hans.lundmark@liu.se}, \href{mailto:budor.shuaib@liu.se}{budor.shuaib@liu.se}, \href{mailto:budor.m@gmail.com>}{budor.m@gmail.com}}

\ArticleDates{Received July 06, 2018, in final form February 19, 2019; Published online March 06, 2019}

\Abstract{We derive explicit formulas for the characteristic curves associated with the multipeakon solutions of the Camassa--Holm, Degasperis--Procesi and Novikov equations. Such a curve traces the path of a fluid particle whose instantaneous velocity equals the elevation of the wave at that point (or the square of the elevation, in the Novikov case). The peakons themselves follow characteristic curves, and the remaining characteristic curves can be viewed as paths of ``ghostpeakons'' with zero amplitude; hence, they can be obtained as solutions of the ODEs governing the dynamics of multi\-peakon solutions. The previously known solution formulas for multipeakons only cover the case when all amplitudes are nonzero, since they are based upon inverse spectral methods unable to detect the ghostpeakons. We show how to overcome this problem by taking a suitable limit in terms of spectral data, in order to force a selected peakon amplitude to become zero. Moreover, we use direct integration to compute the characteristic curves for the solution of the Degasperis--Procesi equation where a shockpeakon forms at a peakon--antipeakon collision. In addition to the theoretical interest in knowing the characteristic curves, they are also useful for plotting multipeakon solutions, as we illustrate in several examples.}

\Keywords{peakons; characteristic curves; Camassa--Holm equation; Degasperis--Procesi equation; Novikov equation}

\Classification{35C05; 35C08; 70H06; 37J35; 35A30}

\section{Introduction}\label{sec:intro}

The Camassa--Holm (CH) equation
\begin{gather} \label{eq:CH-kappa}
 m_t + 2 \kappa u_x + m_x u + 2 m u_x = 0 ,\qquad m = u - u_{xx} ,
\end{gather}
was proposed as a model for shallow water waves by Camassa and Holm~\cite{camassa-holm:1993:CH-orginal-paper}, with the sought function $u(x,t)$ being the horizontal fluid velocity component and $\kappa$ a positive parameter.
For more information about the role of this equation in water wave theory we refer to other works
\cite{constantin-lannes:2009:hydrodynamical-CH-DP,dullin-gottwald-holm:2001:integrable-shallow-water-linear-nonlinear-dispersion,
 dullin-gottwald-holm:2003:CH-KdV5-other-asymptotically-equivalent-shallow-water, dullin-gottwald-holm:2004:asymptotically-equivalent-shallow-water, johnson:2002:CH-KdV-water-waves, johnson:2003:classical-water-waves}.
In this article we will be concerned exclusively with its mathematical properties as a PDE and as an integrable system, and moreover we will only consider the limiting case with $\kappa=0$,
\begin{gather} \label{eq:CH}
 m_t + m_x u + 2 m u_x = 0 ,\qquad m = u - u_{xx} ,
\end{gather}
which in a weak sense admits peaked soliton solutions, \emph{peakons} for short. The $N$-peakon solution takes the simple form
\begin{gather} \label{eq:multipeakons}
 u(x,t) = \sum_{i=1}^N m_i(t) {\rm e}^{-|x-x_i(t)|} ,
\end{gather}
where the positions~$x_i(t)$ and the amplitudes~$m_i(t)$ are determined by the system of ODEs~\eqref{eq:CH-peakon-ODEs} below.

The main aim of this article is to give explicit formulas for the \emph{characteristic curves} $x = \xi(t)$ associated with the Camassa--Holm multipeakon solutions~\eqref{eq:multipeakons}. We obtain these formulas as solutions of the peakon ODEs~\eqref{eq:CH-peakon-ODEs} where some amplitudes $m_k(t)$ are identically zero; these ``ghostpeakon'' solutions, as we call them, are in turn obtained from the known explicit formulas for
ordinary peakon solutions via a relatively simple limiting procedure, thereby avoiding the problem of directly trying to integrate the ODE $\dot \xi(t) = u( \xi(t), t)$ for the characteristics. We also solve the corresponding problem for two mathematical relatives of the Camassa--Holm equation, described in more detail below, namely the Degasperis--Procesi equation~\eqref{eq:DP} and the Novikov equation~\eqref{eq:Novikov}.

Having formulas for the characteristic curves makes it easy to produce good three-dimensional plots of multipeakon solutions, and a secondary aim of the article, which partly serves an expository and pedagogical purpose,
is to provide illustrations of phenomena such as peakon--antipeakon collisions and different continuations of the solution past a collision (conservative vs. dissipative).

\subsection*{Outline of the article}

Section~\ref{sec:background} contains some background material about the Camassa--Holm equation, weak solutions in general, peakons in particular, and the role of characteristic curves when considering the problem of non-uniqueness of weak solutions, and similarly for the Degasperis--Procesi and Novikov equations. Remark~\ref{rem:plotting} describes how knowledge of the characteristic curves is useful for plotting multipeakon solutions. Remark~\ref{rem:GX} describes yet another motivation for the current study, namely to develop techniques that can be used to obtain the most general peakon solution of the two-component Geng--Xue equation~\eqref{eq:GX}.

In Section~\ref{sec:solution-formulas} we recall the known explicit solution formulas for the $N$-peakon ODEs with all amplitudes nonzero. This includes explaining the notation needed for the rest of the article. Remarks~\ref{rem:CH-threepeakon-parameters} and~\ref{rem:DP-threepeakon-parameters} give symmetric ways of writing the Camassa--Holm and Degasperis--Procesi three-peakon solutions, which are new as far as we know.

The formulas for Camassa--Holm ghostpeakons (or, equivalently, the characteristic curves of the $N$-peakon solutions) are then derived and exemplified in Section~\ref{sec:CH-ghost}. Theorem~\ref{thm:CH-ghost} and Corollary~\ref{cor:CH-add-a-ghostpeakon} contain the main results. Remarks~\ref{rem:Matsuno} and~\ref{rem:Matsuno-continued} outline alternative proofs based on direct integration. Example~\ref{ex:CH-3p-ghost} illustrates a pure three-peakon solution and its characteristic curves, Example~\ref{ex:CH-2p-1ap-ghost} does the same for a conservative solution with two peakons and one antipeakon, and Example~\ref{ex:CH-dissipative-2p-1ap} treats the dissipative case where peakons merge at collisions; here we show how to obtain explicit formulas for the whole solution (and its characteristic curves) by gluing a three-peakon solution to a two-peakon solution, and then the two-peakon to a one-peakon solution, recalculating the spectral parameters at each collision.

The corresponding results for the Degasperis--Procesi equation are given in Section~\ref{sec:DP-ghost}, which also contains a calculation (by direct integration) of the characteristics for the Degasperis--Procesi one-shockpeakon solution which forms at a peakon--antipeakon collision; the symmetric case is treated in Example~\ref{ex:DP-p-ap-symm} and the asymmetric case in Example~\ref{ex:DP-p-ap-asymm}.

Finally, the formulas for Novikov ghostpeakons are derived in Section~\ref{sec:Novikov-ghost}. The plotting technique described in Remark~\ref{rem:plotting} does not quite work for Novikov's equation, since the relation $\dot \xi = u^2$ only gives the absolute value $\abs{u}$ and not the sign of~$u$. Instead, Theorem~\ref{thm:Novikov-u-and-ughost} provides a formula for computing~$u(\xi(t,\theta),t)$ directly, which allows us to also plot peakon--antipeakon solutions of Novikov's equation; see Fig.~\ref{fig:Novikov-1p-4cluster-wave} in Example~\ref{ex:Novikov-1p-4cluster}, which illustrates a solution with $N=5$, where a cluster of four peakons and antipeakons interacts with a single peakon.

\section{Background}\label{sec:background}

Functions $u(x,t)$ of the form~\eqref{eq:multipeakons} are not classically differentiable where $x = x_i(t)$, but if the derivatives are taken in the sense of distributions, then the quantity
\begin{gather*}
 m(x,t) = u(x,t) - u_{xx}(x,t) = 2 \sum_{i=1}^N m_i(t) \delta\bigl( x - x_i(t) \bigr)
\end{gather*}
is a linear combination of Dirac delta distributions
(and this also motivates the notation~$m_i$ for the amplitudes).
Inserting this directly into the Camassa--Holm equation~\eqref{eq:CH} leads to problems
with the multiplication in the term~$m u_x$,
since $u_x$ is undefined (with a jump singularity)
at exactly the points where the Dirac deltas in~$m$ are located,
and there are similar issues with the interpretation of the term~$m_x u$.
In order to rigorously define weak solutions,
one can reformulate~\eqref{eq:CH} by writing it as
\begin{gather*}
 \big(1-\partial_x^2\big)\big[u_t+\big(\tfrac12 u^2\big)_x\big] + \big(u^2 + \tfrac12 u_x^2\big)_x = 0
\end{gather*}
and inverting the operator $1-\partial_x^2$. In the context of peakon solutions, which tend to zero as $\abs{x} \to \infty$, this is achieved through convolution with the function $\tfrac12 {\rm e}^{-\abs{x}}$, so that the equation becomes
\begin{gather} \label{eq:CH-eulerian}
 u_t + \partial_x \bigl( \tfrac12 u^2 + \tfrac12 {\rm e}^{-\abs{x}} * \bigl( u^2 + \tfrac12 u_x^2 \bigr) \bigr) = 0 .
\end{gather}
Weak solutions can then be defined by integrating this against a test function. Weak solutions of the initial-value problem for~\eqref{eq:CH-eulerian} are not unique, and some additional criterion must be imposed in order to single out the solution one is interested in, for example conservation or dissipation of the $H^1$ energy $E(t) = \int_{\R} \big(u^2 + u_x^2\big) {\rm d}x$; see Remark~\ref{rem:characteristics-review} below, where we give some references and indicate why characteristic curves are important in this context.

The ansatz~\eqref{eq:multipeakons} turns out to satisfy \eqref{eq:CH-eulerian} in the weak sense if and only if the positions~$x_k(t)$ and amplitudes~$m_k(t)$ of the peakons satisfy the canonical equations generated by the Hamiltonian
\begin{gather*}
 H(x_1,\dots,x_N,m_1,\dots,m_N) = \frac12 \sum_{i,j=1}^N m_i m_j {\rm e}^{-|x_i-x_j|} ,
\end{gather*}
namely
\begin{gather}
 \dot x_k = \phantom{-} \frac{\partial H}{\partial m_k} = \sum_{i=1}^N m_i {\rm e}^{-|x_k-x_i|} ,\nonumber\\
 \dot m_k = -\frac{\partial H}{\partial x_k} = \sum_{i=1}^N m_k m_i \sgn(x_k-x_i) {\rm e}^{-|x_k-x_i|} . \label{eq:CH-peakon-ODEs}
\end{gather}
Here we use the convention $\sgn(0)=0$, and it is also assumed that all~$x_k$ are distinct, to avoid points where $H$ is not differentiable. In fact, in this paper we will always assume that the positions are numbered in
increasing order, $x_1(t) < \dots < x_N(t)$. It is known that \emph{pure peakon solutions}, meaning solutions where all amplitudes $m_k(t)$ are positive, are defined for all~$t \in \R$ and automatically preserve this ordering condition, and likewise for \emph{pure antipeakon solutions} where all amplitudes are negative. For \emph{mixed peakon--antipeakon solutions}, with some amplitudes positive and others negative, there will however in general be singularities in the form of \emph{peakon--antipeakon collisions}, where
\begin{gather*}
 x_{k+1}(t) - x_k(t) \to 0 ,\qquad m_k(t) \to +\infty ,\qquad m_{k+1}(t) \to -\infty ,\qquad \text{as $t \nearrow t_0$} ,
\end{gather*}
for some time~$t_0$ and some index~$k$. The question of continuation beyond the singularity is quite subtle, as already the simplest example shows.

\begin{Example} The symmetric peakon--antipeakon solution with $N=2$ was given already in the original Camassa--Holm paper~\cite{camassa-holm:1993:CH-orginal-paper}:
 \begin{gather} \label{eq:CH-symmetric-conservative}
 u(x,t) = \begin{cases}
 \dfrac{c}{\tanh c t} \big({-}{\rm e}^{-\abs{x + \ln \cosh c t}} + {\rm e}^{-\abs{x - \ln \cosh c t}} \big) ,& t < 0, \\
 0 ,& t = 0, \\
 \dfrac{c}{\tanh c t} \big({-}{\rm e}^{-\abs{x + \ln \cosh c t}} + {\rm e}^{-\abs{x - \ln \cosh c t}} \big) ,& t > 0 ,
 \end{cases}
 \end{gather}
 where $c>0$. This is a so-called \emph{conservative} weak solution, where the peakon and antipeakon cancel out completely at the instant of collision ($t=0$), but then immediately reappar. However, another function which also satisfies~\eqref{eq:CH-eulerian} in the weak sense is the following \emph{dissipative} weak solution,
 where $u$ stays identically zero after the collision:
 \begin{gather}
 \label{eq:CH-symmetric-dissipative}
 u(x,t) = \begin{cases}
 \dfrac{c}{\tanh c t} \big({-}{\rm e}^{-\abs{x + \ln \cosh c t}} + {\rm e}^{-\abs{x - \ln \cosh c t}} \big) ,& t < 0, \\
 0 ,& t \ge 0.
 \end{cases}
 \end{gather}
 And yet another (rather unphysical) weak solution has a peakon--antipeakon pair spontaneously appearing out of nowhere:
 \begin{gather}
 \label{eq:CH-symmetric-nonphysical}
 u(x,t) = \begin{cases}
 0 ,& t \le 0, \\
 \dfrac{c}{\tanh c t} \big( {-}{\rm e}^{-\abs{x + \ln \cosh c t}} + {\rm e}^{-\abs{x - \ln \cosh c t}} \big) ,& t > 0 .
 \end{cases}
 \end{gather}
Clearly we must rule out this kind of behaviour, where peakon--antipeakon pairs can be created anywhere at any moment, if we want the solution to be unique. The solution~\eqref{eq:CH-symmetric-nonphysical} also has the undesirable property that the energy $E(t) = \int_{\R} \big(u^2 + u_x^2\big) {\rm d}x$ increases; it jumps from being zero for $t \le 0$ to being positive for $t > 0$. The condition which singles out the conservative solution~\eqref{eq:CH-symmetric-conservative} is that $E(t)$ is required to be constant for almost all~$t$ (in this case, all $t \neq 0$), while the condition that gives the dissipative solution~\eqref{eq:CH-symmetric-dissipative} is that $E(t)$ is required to be a non-increasing function of~$t$. Both these criteria give the same unique two-peakon solution up until the time of collision, but pick out different continuations after the collision.
\end{Example}

The $2N$ coupled nonlinear ODEs~\eqref{eq:CH-peakon-ODEs} form a completely integrable finite-dimensional Hamiltonian system, and explicit formulas for the solution $\{ x_k(t), m_k(t) \}_{k=1}^N$ were derived by Beals, Sattinger and Szmigielski~\cite{beals-sattinger-szmigielski:1999:stieltjes,beals-sattinger-szmigielski:2000:moment} using inverse spectral techniques based on the Lax pair for the Camassa--Holm equation. These formulas will be recalled in Section~\ref{sec:CH-N}. In the derivation it is assumed that all the amplitudes~$m_k(t)$ are nonzero. This is natural, since if some $m_k(t)$ is zero for some~$t$, then the ODEs~\eqref{eq:CH-peakon-ODEs}
imply that $m_k(t)$ stays zero for all~$t$, and therefore the term $m_k {\rm e}^{-|x-x_k|}$ does not contribute to the function $u(x,t)$ given by~\eqref{eq:multipeakons}, and can be disregarded. This means that the Beals--Sattinger--Szmigielski formulas indeed provide the general multipeakon solution~\eqref{eq:multipeakons} of the PDE, at least locally, away from peakon--antipeakon collisions.

However, if we view the system of ODEs~\eqref{eq:CH-peakon-ODEs} as an integrable system in its own right, the formulas do \emph{not} give the most general solution. For example, if one~$m_k$ vanishes identically, and the other amplitudes~$m_i$ are nonzero, then the ODEs for $\{x_i, m_i \}_{i \neq k}$ reduce to the $(N-1)$-peakon ODEs with nonzero amplitudes, for which we know the solution, but feeding this solution into the remaining equation for $x_k$ gives a non-autonomous ODE $\dot x_k = f(x_k,t)$, and it is definitely not obvious how to integrate this equation to obtain~$x_k(t)$ explicitly. The variable~$x_k$ is not directly accessible to the inverse spectral technique, since a peakon with amplitude zero leaves no trace in the spectral data. If several amplitudes~$m_k$ are identically zero, then in the same way we are left with a non-autonomous ODE for each corresponding~$x_k$,
but they are all of the same form and not coupled to each other, so if we know how to solve a typical one, we can solve them all.

One of the purposes of this article is to demonstrate how these exceptional solutions of the peakon ODEs~\eqref{eq:CH-peakon-ODEs}, with one or several~$m_k$ vanishing identically, can be obtained from the Beals--Sattinger--Szmigielski formulas via a relatively simple limiting procedure. This will remedy the somewhat peculiar situation of having an integrable system that we know how to integrate in the generic case, but not in the seemingly simpler case when some of the variables are identically zero.

We find the following terminology convenient:

\begin{Definition} \label{def:ghostpeakon} A solution $\{ x_i(t), m_i(t) \}_{i=1}^N$ of the Camassa--Holm $N$-peakon ODEs~\eqref{eq:CH-peakon-ODEs} is said to have a \textit{ghostpeakon} at site~$k$ if $m_k(t)=0$ for all~$t$.
 The corresponding function $x_k(t)$ will be referred to as \textit{the position of the ghostpeakon}, and \textit{the trajectory of the ghostpeakon} is the curve $x=x_k(t)$ in the $(x,t)$ plane.
\end{Definition}

\begin{Remark} \label{rem:brasklapp} To avoid possible misunderstandings, we emphasize that this (by definition) is nothing but terminology relating to the system of ODEs~\eqref{eq:CH-peakon-ODEs}, viewed as a~finite-dimensional dynamical system in its own right. The actual wave~\eqref{eq:multipeakons}, i.e., the solution $u(x,t)$ of the Camassa--Holm PDE~\eqref{eq:CH}, is made up entirely of ordinary peakons with nonzero amplitudes, so ghostpeakons are not a physical phenomenon, and they do not influence the ordinary peakons in any way. (It is the other way around: the ordinary peakons determine the dynamics of the ghostpeakons.) Despite this, ghostpeakons actually do have some relevance to the understanding of the PDE solution~$u(x,t)$, since the ghostpeakon trajectories are characteristic curves for the multipeakon solution $u(x,t)$ formed by the ordinary peakons, as we will explain below.
\end{Remark}

\begin{Definition} \label{def:char-curve-CH} The \textit{characteristic curves} for a given solution $u(x,t)$ of the Camassa--Holm equation~\eqref{eq:CH} are the solutions curves $x = \xi(t)$ of the ODE
 \begin{gather*}
 \dot \xi(t) = u ( \xi(t), t ) .
 \end{gather*}
\end{Definition}

\begin{Remark} \label{rem:characteristics-review} These characteristic curves (or \emph{characteristics} for short) play a central role in the study of the Camassa--Holm equation and similar PDEs. Consider for example the initial-value problem for the Camassa--Holm equation~\eqref{eq:CH} on the real line. There are various works \cite[and many others]{constantin-escher:1998:CH-global-existence-blowup,li-olver:2000:CH-wellposedness-blowup,rodriguezblanco:2001:CH-cauchy-problem,molinet:2004:CH-wellposedness-survey} which prove existence and uniqueness of a solution in a suitable function space, at least on some time interval $0 \le t < T$. The limitation $t<T$ is unavoidable in general, since for certain initial data $u_0(x)=u(x,0)$ it may happen that the solution leaves the function space in question after some finite time~$T$; typically the solution~$u$ remains continuous but its derivative~$u_x$ blows up. This idea of ``wave breaking in finite time'' is present already in the original Camassa--Holm paper~\cite{camassa-holm:1993:CH-orginal-paper}, where they consider the slope~$u_x$ at an inflection point of~$u$, to the right of the maximum of~$u$, and sketch a proof showing that this slope must tend to $-\infty$ in finite time. A fleshed-out argument was given in a later paper~\cite{camassa-holm-hyman:1994:CH-new-integrable}. Another type of condition implying finite-time blowup involves assuming a sign change from positive to negative in~$m_0(x)=m(x,0)$, where $m=u-u_{xx}$. In this case, one follows a characteristic curve emanating from a point where $m_0$ changes sign, and aims to show that $u_x \to -\infty$ along that curve after finite time. Many arguments of this kind follow the approach described in detail by Constantin~\cite{constantin:2000:CH-geometric-approach}. See also (for example) McKean~\cite{mckean:1998:CH-breakdown,mckean:2004:CH-breakdown2}, Jiang, Ni and Zhou~\cite{jiang-ni-zhou:2012:CH-wave-breaking,zhou:2004:CH-wave-breaking} and Brandolese~\cite{brandolese:2014:CH-local-in-space-blowup-criteria}. If the solution blows up after finite time, the question arises whether it is possible to continue it past the singularity. It turns out that the answer is yes, but the continuation (like weak solutions in general) is not unique unless some suitable additional condition is imposed. Global weak solutions of the Camassa--Holm equation were first studied by Constantin and Escher~\cite{constantin-escher:1998:CH-global-weak}, and have since been investigated in great detail \cite{bressan-chen-zhang:2015:CH-uniqueness-conservative-via-characteristics,bressan-constantin:2007:global-conservative-CH, bressan-constantin:2007:global-dissipative-CH, constantin-molinet:2000:CH-global-weak, grunert-holden:2016:CH-peakon-antipeakon-alpha-dissipative, holden-raynaud:2007:CH-global-conservative-multipeakon, holden-raynaud:2007:CH-global-conservative-Lagrangian, holden-raynaud:2008:CH-global-dissipative-multipeakon, holden-raynaud:2009:CH-dissipative-solutions, xin-zhang:2000:CH-weak-solutions-existence, xin-zhang:2002:CH-weak-solutions-uniqueness-large-time-behaviour}. The idea is to remove the singularity by changing to new variables, whose very definition involves characteristic curves. This means that even the ``simple'' task of verifying that a proposed explicit solution $u(x,t)$, such as \eqref{eq:CH-symmetric-conservative}, \eqref{eq:CH-symmetric-dissipative} or~\eqref{eq:CH-symmetric-nonphysical}, really satisfies the reformulated equation will require some knowledge of the characteristic curves for that solution. To appreciate how complicated this might be already for $N$-peakon solutions with small~$N$, see in particular the works by Grunert, Holden and Raynaud concerning peakon--antipeakon solutions~\cite{grunert-holden:2016:CH-peakon-antipeakon-alpha-dissipative, holden-raynaud:2007:CH-global-conservative-multipeakon,holden-raynaud:2008:CH-global-dissipative-multipeakon}. In the theory of global weak solutions of the Camassa--Holm equation, a fundamental distinction is that between \emph{conservative solutions} and \emph{dissipative solutions}. Conservative solutions preserve the $H^1$ energy $E(t) = \int_{\R} \big(u^2 + u_x^2\big) {\rm d}x$ for almost all~$t$; for example, at a~peakon--antipeakon collision, $E(t)$ momentarily drops to a~lower value as the peakon and antipeakon merge into a~single peakon (or cancel out completely), but then it immediately returns to its previous value again as the peakon and antipeakon reappear. Dissipative solutions are characterized by the condition that $E(t)$ is nonincreasing, so once the energy drops to a~lower value it cannot go back up, which means that the merged peakons must stay together, with some energy having been lost to dissipation at the collision. There is also the concept of \emph{$\alpha$-dissipative solutions}, introduced by Grunert, Holden and Raynaud~\cite{grunert-holden:2016:CH-peakon-antipeakon-alpha-dissipative, grunert-holden-raynaud:2015:alpha-dissipative-2CH}; these solutions have the property that the fraction~$\alpha \in (0,1)$ of the energy concentrated at the collision is lost, so they constitute an intermediate case between conservative ($\alpha=0$) and fully dissipative ($\alpha=1$). Finally, we mention that characteristic curves also play a prominent
role in many numerical schemes for solving Camassa--Holm-type equations \cite{camassa-huang-lee:2005:CH-completely-integrable-numerical-scheme, camassa-huang-lee:2006:CH-integral-and-integrable-algorithms, chertock-liu-pendleton:2012:CH-particle-method-convergence-analysis, chertock-liu-pendleton:2012:CH-DP-particle-method-global-weak, chertock-liu-pendleton:2015:CH-elastic-peakon-collisions, holden-raynaud:2006:CH-convergent-scheme-based-on-multipeakons, holden-raynaud:2008:CH-conservative-numerical-scheme-based-on-multipeakons}.
\end{Remark}

To explain the connection between ghostpeakons (Definition~\ref{def:ghostpeakon}) and characteristic curves (Definition~\ref{def:char-curve-CH}), note first that the equation for $x_k$ in the peakon ODEs~\eqref{eq:CH-peakon-ODEs} reads
\begin{gather*}
 \dot x_k(t) = \sum_{i=1}^N m_i(t) {\rm e}^{-|x_k(t) - x_i(t)|} = u( x_k(t), t) .\end{gather*}
In other words, the peakon trajectory $x = x_k(t)$ must be a characteristic curve for the multipeakon solution~\eqref{eq:multipeakons} itself. But this is true also if the amplitude~$m_k$ is zero. That is, if we have a solution $\{ x_i(t), m_i(t) \}_{i=1}^N$ of the peakon ODEs~\eqref{eq:CH-peakon-ODEs} with $m_k(t)=0$ and the other $m_i(t) \neq 0$, i.e., ``a solution with a ghostpeakon at site~$k$ (only)'', then the ghostpeakon trajectory $x=x_k(t)$ is a characteristic curve for the $(N-1)$-peakon solution
\begin{gather} \label{eq:u-with-ghostpeakon-at-site-k}
 u(x,t) = \sum_{i=1}^N m_i(t) {\rm e}^{-|x-x_i(t)|} = \sum_{\substack{1 \le i \le N \\[0.5ex] i \neq k}} m_i(t) {\rm e}^{-|x-x_i(t)|},
\end{gather}
which contains only the contributions from the ordinary (non-ghost) peakons. So the ghostpeakon trajectory is a characteristic curve which is not an ordinary peakon trajectory, but instead lies between two peakons (or in the region outside the peakons, if $k=1$ or $k=N$).

Turning this around, we can think of the situation as follows: if we start with a solution containing only ordinary peakons and are interested in finding a characteristic curve which is not a peakon trajectory, say $x = \xi(t)$ with $\xi(0)=\xi_0$ distinct from all~$x_k(0)$, we can imagine a ghostpeakon being added to the system with position~$\xi_0$ at time $t=0$, and obtain the characteristic curve as the trajectory of that ghostpeakon.
Thus, finding exact solution formulas for ghostpeakons will tell us explicitly what the characteristic curves $x = \xi(t)$ are in the case of multipeakon solutions. As mentioned in Remark~\ref{rem:characteristics-review},
this is of interest in the study of peakon--antipeakon collisions, for example.

\begin{Remark} \label{rem:plotting} An additional bonus of knowing the formulas for the characteristic curves is that it makes it much easier to produce high-quality three-dimensional plots of the graph of the function $u(x,t)$ for multipeakon solutions. With knowledge of $\{ x_i(t), m_i(t) \}_{i=1}^N$ from the Beals--Sattinger--Szmigielski solution formulas~\cite{beals-sattinger-szmigielski:1999:stieltjes,beals-sattinger-szmigielski:2000:moment}, we can of course compute $u(x,t) = \sum\limits_{i=1}^N m_i(t) {\rm e}^{-|x-x_i(t)|}$ for any $x$ and~$t$, but plotting this surface using an ordinary rectangular mesh, the ``mountain ridges'' along the peakon trajectories are likely to come out jagged and ugly. There is also the problem of numerical cancellation in the summation when plotting a solution involving peakon--antipeakon collisions where $m_k(t) \to +\infty$ and $m_{k+1}(t) \to -\infty$, and moreover $|u_x|$ becomes very large in a~small region near the collision, which is hard to display correctly even with a very fine rectangular mesh. If we instead use the explicit formulas for the characteristic curves $x = \xi(t, \theta)$ and plot the graph as the parametric surface
\begin{gather*}
 (\theta,t) \mapsto (x,t,u) = ( \xi(t,\theta), t, \dot\xi(t,\theta) ) ,
 \end{gather*}
where $\dot\xi$ can easily be computed symbolically (or by automatic differentiation), then we avoid all of these problems; we even get a mesh which is automatically finer at the points where $\abs{u_x}$ is large, since the characteristics between the colliding peakon and antipeakon converge at the point of collision. As a first example, Fig.~\ref{fig:CH-p-ap-wave} shows a conservative asymmetric peakon--antipeakon solution of the Camassa--Holm equation plotted using this technique, and we will provide several other illustrations of multipeakon solutions in later sections.
\end{Remark}

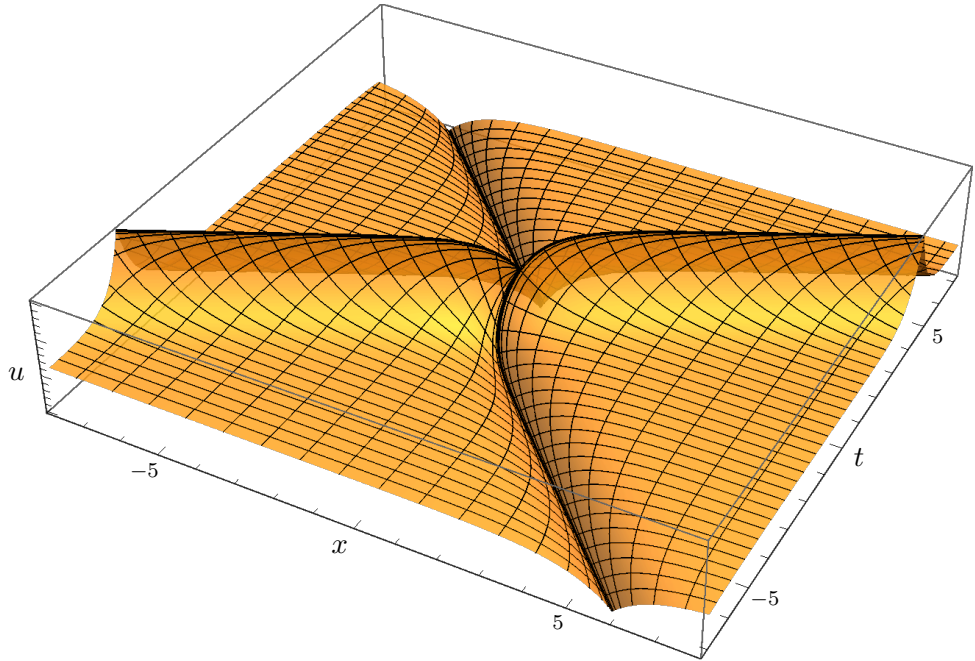
\begin{figure}[t] \centering
\begin{tikzpicture}
 \node[anchor=south west,inner sep=0] at (0,0)
 {\includegraphics[width=127mm]{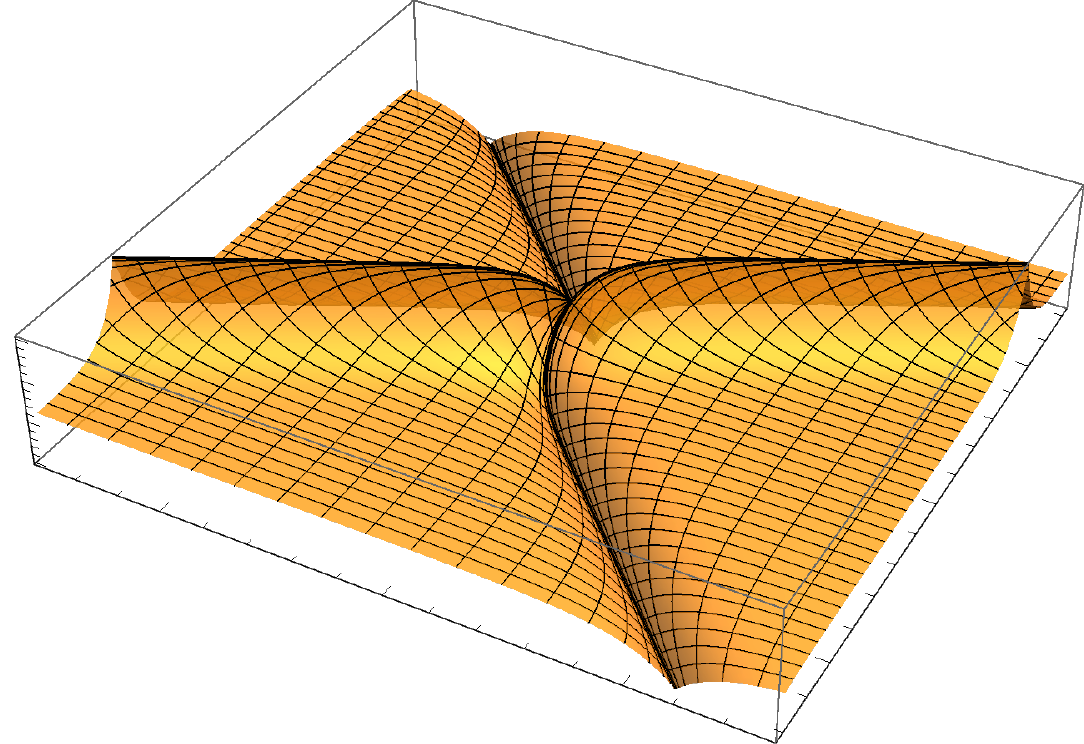}};

 \draw (4.3,1.5) node {$x$};
 \draw (1.7,2.5) node {\scriptsize $-5$};
 \draw (7.2,0.5) node {\scriptsize $5$};

 \draw (11.2,2.7) node {$t$};
 \draw (9.95,0.8) node {\scriptsize $-5$};
 \draw (12.25,4.3) node {\scriptsize $5$};

 \draw (0,3.8) node {$u$};
 \end{tikzpicture}

\caption{Graph of a conservative asymmetric peakon--antipeakon solution $u(x,t) = m_1(t) {\rm e}^{-\abs{x-x_1(t)}} + m_2(t) {\rm e}^{-\abs{x-x_2(t)}}$ of the Camassa--Holm equation, plotted as explained in Remark~\ref{rem:plotting}: the mesh consists of lines $t={\rm const}$, which are lifted to the surface so that they illustrate the peakon wave profile $u(x)=\sum m_k {\rm e}^{-\abs{x-x_k}}$ at the instant in question, together with characteristic curves (also lifted to the surface) given by the ghostpeakon formulas from Corollary~\ref{cor:CH-add-a-ghostpeakon}. The parameters in the solution formulas are given by \eqref{eq:CH-peakon-antipeakon-parameters-symmetric} with $c_1=2$ and $c_2=1$, so that the collision takes place at $(x,t)=(0,0)$ and the asymptotic velocities are $2$ and~$-1$. Before the collision, $x_1(t) < x_2(t)$ with $m_1(t) > 0$ and $m_2(t) < 0$. As $t \to 0^-$, $m_1(t) \to +\infty$ and $m_2(t) \to -\infty$. The limiting wave profile at the instant of collision (where $x_1(0)=x_2(0)=0$) consists of just a single peakon: $u(x,0) = {\rm e}^{-\abs{x}}$. After the collision, $x_1(t) < x_2(t)$ holds again, but now with $m_1(t) < 0$ and $m_2(t) > 0$. The dimensions of the box are $\abs{x} \le 8$, $\abs{t} \le 7$ and $-1 \le u \le 2$.} \label{fig:CH-p-ap-wave}
\end{figure}

We will also derive explicit ghostpeakon formulas for two other PDEs, namely the Degasperis--Procesi (DP) equation
\cite{degasperis-holm-hone:2002:new-integrable-equation-DP, degasperis-procesi:1999:asymptotic-integrability}
\begin{gather} \label{eq:DP}
 m_t + m_x u + 3 m u_x = 0 ,\qquad m = u - u_{xx} ,
\end{gather}
and the Novikov equation~\cite{hone-wang:2008:cubic-nonlinearity,novikov:2009:generalizations-of-CH}
\begin{gather} \label{eq:Novikov}
 m_t + m_x u^2 + 3 m u u_x = 0 ,\qquad m = u - u_{xx} ,
\end{gather}
which like the Camassa--Holm equation~\eqref{eq:CH} are integrable systems in the sense of having Lax pairs, peaked multi\-soliton solutions of the form~\eqref{eq:multipeakons}, and infinitely many conservation laws. In terms of the shorthand notation
\begin{gather*}
 u(x_k) = \sum_{i=1}^N m_i {\rm e}^{-|x_k-x_i|} ,\qquad u_x(x_k) = - \sum_{i=1}^N m_i \sgn(x_i-x_k) {\rm e}^{-|x_k-x_i|} ,
\end{gather*}
the Camassa--Holm peakon ODEs~\eqref{eq:CH-peakon-ODEs} take the form
\begin{gather} \label{eq:CH-peakon-ODEs-shorthand}
 \dot x_k = u(x_k) ,\qquad \dot m_k = - m_k u_x(x_k) ,
\end{gather}
while the ODEs governing the dynamics of the multipeakon solutions~\eqref{eq:multipeakons} for the Degasperis--Procesi peakons are
\begin{gather} \label{eq:DP-peakon-ODEs-shorthand}
 \dot x_k = u(x_k) ,\qquad \dot m_k = - 2 m_k u_x(x_k) ,
\end{gather}
and the ODEs for Novikov peakons are
\begin{gather} \label{eq:Novikov-peakon-ODEs-shorthand}
 \dot x_k = u(x_k)^2 ,\qquad \dot m_k = - m_k u(x_k) u_x(x_k) .
\end{gather}
Explicit formulas for the general pure peakon solution were derived for the Degasperis--Procesi equation by Lundmark and Szmigielski~\cite{lundmark-szmigielski:2003:DPshort,lundmark-szmigielski:2005:DPlong}, and for the Novikov equation by Hone, Lundmark and Szmigielski~\cite{hone-lundmark-szmigielski:2009:novikov}. These solution formulas, which will be recalled in Section~\ref{sec:solution-formulas}, also provide valid peakon--antipeakon solutions for suitable choices of the parameters, but they do not quite give the most general solution in this case, since there are non-generic peakon--antipeakon configurations with non-simple eigenvalues that require separate formulas; see Remarks~\ref{rem:DP-antipeakons} and~\ref{rem:Novikov-antipeakons}.

The characteristic curves for a solution $u(x,t)$ of the Degasperis--Procesi equation~\eqref{eq:DP} are defined by the same equation $\dot \xi(t) = u( \xi(t), t)$ as before, while for the Novikov equation~\eqref{eq:Novikov}
the characteristics are given by $\dot \xi(t) = u ( \xi(t), t)^2$, with $u^2$ instead of~$u$ on the right-hand side. Thus, according to \eqref{eq:DP-peakon-ODEs-shorthand} and~\eqref{eq:Novikov-peakon-ODEs-shorthand},
the trajectories of the peakons (including ghostpeakons) are characteristics curves.

\begin{Remark} \label{rem:DP-Novikov-references} The Degasperis--Procesi and Novikov equations were not found through physical considerations but by purely mathematical means, using integrability tests to single out integrable systems of a form similar to the Camassa--Holm equation (although the DP equation has later been interpreted as water wave equation~\cite{constantin-lannes:2009:hydrodynamical-CH-DP,johnson:2003:classical-water-waves}). There are many interesting similarities and differences among these three PDEs. The Lax pair for \eqref{eq:CH} involves a second-order ODE in the $x$ direction, while those for \eqref{eq:DP} and~\eqref{eq:Novikov} are of order three. And any weak solution of~\eqref{eq:CH} or~\eqref{eq:Novikov} is by necessity continuous, whereas \eqref{eq:DP} also admits discontinuous weak solutions \cite{coclite-karlsen:2006:DPwellposedness,coclite-karlsen:2007:DPuniqueness,coclite-karlsen:2015:DPwellposedness-and-asymptotics}, in particular ``shockpeakons''~\cite{lundmark:2007:shockpeakons}. We may also mention some studies regarding finite-time blowup for solutions of the Degasperis--Procesi equation \cite{escher-liu-yin:2006:DP-globalweak-blowup,liu-yin:2006:DP-global-existence-and-blowup,liu-yin:2007:DP-blowup-phenomena, yin:2003:DP-cauchy-problem:Illinois}. Concerning Novikov's equation, Chen, Chen and Liu~\cite{chen-chen-liu:2018:novikov-global-conservative-weak-existence-uniqueness} have recently studied existence and uniqueness of global conservative weak solutions, using a characteristics-based approach. Novikov's equation always conserves the $H^1$ norm $E(t)$, even at peakon--antipeakon collisions, and they instead use the conservation of the quartic quantity
 \begin{gather*}
 F(t) = \int_{\R} \big(u^4 + 2 u^2 u_x^2 - \tfrac13 u_x^4\big) {\rm d}x
 \end{gather*}
 to single out the \emph{conservative} weak solutions. Several other works on the Novikov equation are also relevant in this context \cite{cai-chen-chen-shen:2018:novikov-lipschitz-metric, chen-guo-liu-qu:2016:DP-novikov-mCH-blowup,
 himonas-holliman-kenig:2018:novikov-2-peakon-illposedness,jiang-ni:2012:novikov-blowup, lai-li-wu:2013:novikov-global-solutions, ni-zhou:2011:novikov,tiglay:2011:novikov-periodic-cauchy-problem, wu-guo:2016:periodic-novikov-global-wellposedness, wu-yin:2011:novikov-global-weak-solutions, wu-yin:2012:novikov-wellposedness-global-existence, yan-li-zhang:2012:novikov-cauchy-problem}.
\end{Remark}

\begin{Remark} \label{rem:GX} The ideas developed in this article have turned out very useful in the study of ordinary (non-ghost) peakon solutions of the Geng--Xue equation \cite{geng-xue:2009:GX-peakon-equation-cubic-nonlinearity}, a coupled two-component generalization of the Novikov equation:
 \begin{alignat}{3}
& m_t + (m_xu + 3mu_x)v = 0, \qquad && m = u - u_{xx},&\nonumber\\
& n_t + (n_xv + 3nv_x)u = 0, \qquad && n = v - v_{xx}. & \label{eq:GX}
 \end{alignat}
In this system, a peakon in one component~$u(x,t)$ is not allowed to occupy the same position as a peakon in the other component~$v(x,t)$, and therefore there are many inequivalent peakon configurations, depending on the order in which the peakons in $u$ and~$v$ occur relative to each other. The solution formulas for the \emph{interlacing} configuration, where the peakons alternate (one peakon in~$u$, then one in~$v$, then one in~$u$ again, then one in~$v$, and so on) have been derived by inverse spectral methods \cite{lundmark-szmigielski:2016:GX-inverse-problem,lundmark-szmigielski:2017:GX-dynamics-interlacing}. However, it is not obvious how to apply these techniques in non-interlacing cases, since then the two Lax pairs for the Geng--Xue equation do not seem to provide sufficiently many constants of motion to make it possible to integrate the peakon ODEs. Instead, we derive the solution formulas for an arbitrary configuration by starting from an interlacing solution (with a larger number of peakons) and driving selected amplitudes to zero by taking suitable limits, thus turning some of the peakons into ghostpeakons, in such a~way that the remaining ordinary peakons occur in the desired configuration. The details, which are quite technical, are described in a separate article~\cite{shuaib-lundmark:2018p:GX-noninterlacing}.
\end{Remark}

\section[Review of formulas for $N$-peakon solutions with nonzero amplitudes]{Review of formulas for $\boldsymbol{N}$-peakon solutions\\ with nonzero amplitudes}\label{sec:solution-formulas}

In this section we collect the known formulas for the $N$-peakon solutions of the Camassa--Holm, Degasperis--Procesi and Novikov equations, where it is assumed that all peakons have nonzero amplitude~$m_k(t)$. The notation defined here will also be used to state our new formulas for ghostpeakon solutions in later sections, and the $N$-peakon solution formulas will be needed in the proofs.

\subsection{Camassa--Holm peakons}\label{sec:CH-N}

First we formulate the explicit formulas for Camassa--Holm $N$-peakon solutions. The solutions for $N=1$ and $N=2$ were computed already in the original Camassa--Holm paper~\cite{camassa-holm:1993:CH-orginal-paper}. For $N \ge 3$ direct integration seems difficult, but the solution for arbitrary~$N$ was found with the help of inverse spectral methods by Beals, Sattinger and Szmigielski~\cite{beals-sattinger-szmigielski:1999:stieltjes,beals-sattinger-szmigielski:2000:moment}. Note that they use a different normalization of the Camassa--Holm equation, so their formulas differ by various factors of~$2$ from the ones that we state here, and they also use the opposite sign convention for the parameters~$\lambda_k$.

First some notation. For any integers $a$ and~$k$, and for a fixed $N \ge 1$, let
\begin{gather} \label{eq:Delta-evaluated}
 \Delta_k^a = \Delta_k^a(N) =
 \begin{cases}
 \displaystyle \sum_{1 \le i_1 < \dotsb < i_k \le N} \Delta ( \lambda_{i_1}, \dots, \lambda_{i_k} )^2 \left( \prod_{r=1}^k \lambda_{i_r}^a b_{i_r} \right) ,& \text{if $1 \le k \le N$} ,\vspace{1mm}\\
1,& \text{if $k=0$} ,\\
 0,& \text{otherwise} ,
 \end{cases}
\end{gather}
where $\Delta$ on the right-hand side denotes the Vandermonde determinant:
\begin{gather} \label{eq:Vandermonde}
 \Delta(x_1,x_2,\dots,x_k) = \prod_{1 \le i < j \le k} (x_j - x_i) .
\end{gather}
The superscript $a$ in $\Delta_k^a$ is just a label, whereas in $\lambda_i^a$ it denotes the $a$th power of the constant~$\lambda_i$. (The expression~\eqref{eq:Delta-evaluated} arises by evaluation of a certain Hankel determinant; see \eqref{eq:Delta-determinant} in Remark~\ref{rem:Matsuno}.)

In terms of these quantities $\Delta_k^a$, the Beals--Sattinger--Szmigielski formulas for the general solution $\{ x_k(t), m_k(t) \}_{k=1}^N$ of the Camassa--Holm $N$-peakon ODEs~\eqref{eq:CH-peakon-ODEs} (or \eqref{eq:CH-peakon-ODEs-shorthand} in shorthand notation), with the proviso that $m_k \neq 0$ for all~$k$, are
\begin{gather} \label{eq:CH-generalsolution}
 x_{N+1-k}(t) = \ln\frac{\Delta_k^0}{\Delta_{k-1}^2} ,\qquad m_{N+1-k}(t) = \frac{\Delta_k^0 \Delta_{k-1}^2}{\Delta_k^1 \Delta_{k-1}^1} ,\qquad k = 1, \dots, N ,
\end{gather}
where $\Delta_k^a$ depends on~$t$ via
\begin{gather*}
 b_k = b_k(t) = b_k(0) {\rm e}^{t/\lambda_k} .
\end{gather*}
We will usually let this time dependence be understood, and write just $\Delta_k^a$ and~$b_k$ instead of~$\Delta_k^a(t)$ and~$b_k(t)$.

The parameters $\lambda_k$ are constant, i.e., time-independent. They are the eigenvalues of a certain symmetric $N \times N$-matrix, hence they are real, and it can also be proved~\cite{beals-sattinger-szmigielski:2000:moment} that they are nonzero and distinct, and that in fact the number of positive (negative) eigenvalues $\lambda_k$ equals the number of positive (negative) amplitudes~$m_k$. The parameters $b_k(0)$ are always positive; they appear as the residues of the so-called modified Weyl function \cite{beals-sattinger-szmigielski:2000:moment}.

\begin{Remark} The set of parameters
 \begin{gather*}
 \{ \lambda_k, b_k(0) \}_{k=1}^N ,
 \end{gather*}
 which is referred to as the \emph{spectral data} (\emph{eigenvalues}~$\lambda_k$ and \emph{residues}~$b_k$), is uniquely determined by the set of initial data
 \begin{gather*}
 \{ x_i(0), m_i(0) \}_{i=1}^N
 \end{gather*}
 up to ordering. The solution formulas \eqref{eq:CH-generalsolution} are given by symmetric functions which are invariant under relabeling $(\lambda_k,b_k) \mapsto (\lambda_{\sigma(k)},b_{\sigma(k)})$ for any permutation $\sigma \in S_N$. So we may prescribe an ordering if we like, for example
 \begin{gather*}
 \lambda_1 < \dots < \lambda_n .
 \end{gather*}
 Since the eigenvalues~$\lambda_k$ are determined by a polynomial equation of degree~$N$, there is no explicit formula for this correspondence in terms of radicals (except for small~$N$), but the \emph{inverse} map is explicitly provided by~\eqref{eq:CH-generalsolution} (with $t=0$).
\end{Remark}

\begin{Example}[CH one-peakon solution] For $N=1$, the solution formulas \eqref{eq:CH-generalsolution} reduce to
 \begin{gather*}
 x_1 = \ln\frac{\Delta_1^0}{\Delta_0^2} = \ln\frac{b_1}{1} ,\qquad m_1 = \frac{\Delta_1^0 \Delta_0^2}{\Delta_1^1 \Delta_0^1} = \frac{b_1 \cdot 1}{\lambda_1 b_1 \cdot 1} .
 \end{gather*}
 Taking the time dependence $b_1 = b_1(t) = b_1(0) {\rm e}^{t/\lambda_1}$ into account, we see that the one-peakon solution $u = m_1 {\rm e}^{-\abs{x-x_1}}$ is just a travelling wave with constant velocity $\dot x_1 = 1/\lambda_1$ and constant amplitude~$m_1=1/\lambda_1 \neq 0$:
 \begin{gather*}
 x_1(t) = \frac{t}{\lambda_1} + \ln b_1(0) ,\qquad m_1(t) = \frac{1}{\lambda_1} .
 \end{gather*}
 This of course also follows immediately from direct integration of the Camassa--Holm peakon ODEs~\eqref{eq:CH-peakon-ODEs}, which for $N=1$ are just $\dot x_1=m_1$ and $\dot m_1=0$.
\end{Example}

\begin{Example}[CH two-peakon solution] \label{ex:CH-twopeakon-solution} Letting $N=2$ in~\eqref{eq:CH-generalsolution} we get the two-peakon solution, $u = m_1 {\rm e}^{-\abs{x-x_1}} + m_2 {\rm e}^{-\abs{x-x_2}}$, with
 \begin{gather*}
 x_1 = \ln\frac{\Delta_2^0}{\Delta_1^2} ,\quad x_2 = \ln\frac{\Delta_1^0}{\Delta_0^2} ,\qquad m_1 = \frac{\Delta_2^0 \Delta_1^2}{\Delta_2^1 \Delta_1^1} ,\qquad m_2 = \frac{\Delta_1^0 \Delta_0^2}{\Delta_1^1 \Delta_0^1} ,
 \end{gather*}
 where
 \begin{gather*}
 \Delta_0^a = 1 ,\qquad \Delta_1^a = \lambda_1^a b_1 + \lambda_2^a b_2 ,\qquad \Delta_2^a = (\lambda_1 - \lambda_2)^2 \lambda_1^a \lambda_2^a b_1 b_2 .
 \end{gather*}
 Written out explicitly, the formulas are
 \begin{alignat}{3}
 & x_1(t) = \ln \frac{(\lambda_1-\lambda_2)^2 b_1 b_2}{\lambda_1^2 b_1 + \lambda_2^2 b_2},\qquad && m_1(t)= \frac{\lambda_1^2 b_1 + \lambda_2^2 b_2}{\lambda_1 \lambda_2 (\lambda_1 b_1 + \lambda_2 b_2)} , &\nonumber\\
& x_2(t) = \ln (b_1+b_2) ,\qquad && m_2(t) = \frac{b_1+b_2}{\lambda_1 b_1 + \lambda_2 b_2} ,&\label{eq:CH-twopeakon-solution}
 \end{alignat}
 with the time dependence
 $b_k = b_k(t) = b_k(0) {\rm e}^{t/\lambda_k}$.
\end{Example}

\begin{Remark} \label{rem:CH-twopeakon-parameters} The Camassa--Holm equation is invariant with respect to translations $x \mapsto x-x_0$, $t \mapsto t-t_0$, and we can use this freedom to reduce the number of parameters by two in the solution formulas. It turns out the eigenvalues~$\lambda_k$ are unaffected by such translations, while the residues~$b_k$ are rescaled. So for the two-peakon solution, the eigenvalues $\lambda_1$ and $\lambda_2$ are the only essential parameters, and we can make $b_1(0)$ and $b_2(0)$ take any positive values by a suitable translation. A particularly useful choice in the pure peakon case, say with $0 < \lambda_1 < \lambda_2$, is to take
 \begin{gather*}
 c_1 = \frac{1}{\lambda_1} > c_2 = \frac{1}{\lambda_2} > 0 ,\qquad b_1(0) = \frac{c_1}{c_1-c_2} ,\qquad b_2(0) = \frac{c_2}{c_1-c_2} .
 \end{gather*}
 Then the two-peakon solution takes the following form, which is symmetric with respect to the reversal $(x,t) \mapsto (-x,-t)$ since $x_1(t) = -x_2(-t)$ and $m_1(t) = m_2(-t)$:
 \begin{alignat*}{3}
& x_1(t)= -\ln \frac{c_1 {\rm e}^{-c_1 t} + c_2 {\rm e}^{-c_2 t}}{c_1-c_2} ,\qquad && m_1(t) = \frac{c_1 {\rm e}^{-c_1 t} + c_2 {\rm e}^{-c_2 t}}{{\rm e}^{-c_1 t} + {\rm e}^{-c_2 t}} ,& \\
& x_2(t)= \ln \frac{c_1 {\rm e}^{c_1 t} + c_2 {\rm e}^{c_2 t}}{c_1-c_2} ,\qquad && m_2(t) = \frac{c_1 {\rm e}^{c_1 t} + c_2 {\rm e}^{c_2 t}}{{\rm e}^{c_1 t} + {\rm e}^{c_2 t}}.
 \end{alignat*}
 For the peakon--antipeakon case, say with $\lambda_1 > 0 > \lambda_2$, we can take
 \begin{gather} \label{eq:CH-peakon-antipeakon-parameters-symmetric}
 c_1 = \frac{1}{\lambda_1} > 0 ,\qquad c_2 = \frac{-1}{\lambda_2} > 0 ,\qquad b_1(0) = \frac{c_1}{c_1+c_2} ,\qquad b_2(0) = \frac{c_2}{c_1+c_2}
 \end{gather}
 to get
 \begin{alignat}{3}
& x_1(t)= -\ln \frac{c_1 {\rm e}^{-c_1 t} + c_2 {\rm e}^{c_2 t}}{c_1+c_2} ,\qquad && m_1(t) = \frac{c_1 {\rm e}^{-c_1 t} + c_2 {\rm e}^{c_2 t}}{{\rm e}^{-c_1 t} - {\rm e}^{c_2 t}} ,& \nonumber\\
& x_2(t)= \ln \frac{c_1 {\rm e}^{c_1 t} + c_2 {\rm e}^{-c_2 t}}{c_1+c_2} ,\qquad && m_2(t) = \frac{c_1 {\rm e}^{c_1 t} + c_2 {\rm e}^{-c_2 t}}{{\rm e}^{c_1 t} - {\rm e}^{-c_2 t}} . & \label{eq:CH-peakon-antipeakon-solution-symmetric}
 \end{alignat}
This will make the peakon--antipeakon collision take place at the origin $(x,t)=(0,0)$. Even though $m_1(t)$ and $m_2(t)$ are undefined at $t=0$, the function $u(x,t)$ given for $t \neq 0$ by $u(x,t) = m_1 {\rm e}^{-\abs{x-x_1}} + m_2 {\rm e}^{-\abs{x-x_2}}$ together with~\eqref{eq:CH-peakon-antipeakon-solution-symmetric}, can be extended continuously by setting
 \begin{gather*}
 u(x,0) = (c_1-c_2) {\rm e}^{-\abs{x}} .
 \end{gather*}
This provides the \emph{conservative} continuation past the collision (see Remark~\ref{rem:characteristics-review} and Fig.~\ref{fig:CH-p-ap-wave}). If we instead let $u(x,t) = (c_1-c_2) {\rm e}^{-\abs{x - (c_1-c_2)t}}$ for $t \ge 0$, then we get the \emph{dissipative} solution; if $c_1 \neq c_2$, the peakon and the antipeakon merge into a single peakon or antipeakon which continues on its own after the collision, whereas if $c_1=c_2$ they annihilate completely, so that $u(x,t)=0$ for all~$t \ge 0$. The \emph{$\alpha$-dissipative} solution~\cite{grunert-holden:2016:CH-peakon-antipeakon-alpha-dissipative} with $0 < \alpha < 1$ is obtained by defining $u(x,t)$ for $t>0$ using~\eqref{eq:CH-peakon-antipeakon-solution-symmetric} with $c_1$ and $c_2$ replaced by new constants $d_1$ and~$d_2$ which are determined from the equations
 \begin{gather*}
 d_1 - d_2 = c_1 - c_2 ,\qquad d_1^2 + d_2^2 = c_1^2 + c_2^2 - 2 c_1 c_2 \alpha ,
 \end{gather*}
which reflect the conservation of momentum and the loss of a fraction $\alpha$ of the energy concentrated at the collision. (For the Camassa--Holm equation, the momentum $\int_{\R} u {\rm d}x = \sum\limits_{i=1}^N m_i$ is always conserved.)
\end{Remark}

\begin{Example}[CH three-peakon solution] \label{ex:CH-threepeakon-solution} Letting $N=3$ in~\eqref{eq:CH-generalsolution} gives the three-peakon solution formulas:
 \begin{subequations} \label{eq:CH-3-explicit}
 \begin{alignat}{4} \label{eq:CH-x-3-explicit}
& x_1 = \ln\frac{\Delta_3^0}{\Delta_{2}^2} ,\qquad && x_2 = \ln\frac{\Delta_2^0}{\Delta_{1}^2} ,\qquad && x_3 = \ln\frac{\Delta_1^0}{\Delta_{0}^2} = \ln \Delta_1^0 ,&\\
 \label{eq:CH-m-3-explicit}
& m_1 = \frac{\Delta_3^0 \Delta_{2}^2}{\Delta_3^1 \Delta_{2}^1} ,\qquad&& m_2 = \frac{\Delta_2^0 \Delta_{1}^2}{\Delta_2^1 \Delta_{1}^1} ,\qquad && m_3 = \frac{\Delta_1^0 \Delta_{0}^2}{\Delta_1^1 \Delta_{0}^1} ,&
 \end{alignat}
 \end{subequations}
 where
 \begin{gather*}
 \Delta_0^a = 1 ,\\
 \Delta_1^a = \lambda_1^a b_1 + \lambda_2^a b_2 + \lambda_3^a b_3 ,\\
 \Delta_2^a = (\lambda_1 \lambda_2)^a (\lambda_1-\lambda_2)^2 b_1 b_2 + (\lambda_1 \lambda_3)^a (\lambda_1-\lambda_3)^2 b_1 b_3 + (\lambda_2 \lambda_3)^a (\lambda_2-\lambda_3)^2 b_2 b_3 ,\\
 \Delta_3^a = (\lambda_1 \lambda_2 \lambda_3)^a (\lambda_1-\lambda_2)^2 (\lambda_1-\lambda_3)^2 (\lambda_2-\lambda_3)^2 b_1 b_2 b_3 ,
\end{gather*}
 with the time dependence $b_k = b_k(t) = b_k(0) {\rm e}^{t/\lambda_k}$.

 Thus, for example, $x_3(t)$ is given by the explicit formula
 \begin{gather*}
 x_3(t) = \ln \Delta_1^0(t) = \ln ( b_1(t) + b_2(t) + b_3(t) ) \\
 \hphantom{x_3(t)}{} = \ln \bigl( b_1(0) {\rm e}^{t/\lambda_1} + b_2(0) {\rm e}^{t/\lambda_2} + b_3(0) {\rm e}^{t/\lambda_3} \bigr) ,
 \end{gather*}
 and similarly for the other variables, but with more involved expressions:
 \begin{gather*}
 x_2(t) = \ln \frac{(\lambda_1 - \lambda_2)^2 b_1(t) b_2(t) + (\lambda_1 - \lambda_3)^2 b_1(t) b_3(t) + (\lambda_2 - \lambda_3)^2 b_2(t) b_3(t) }{\lambda_1^2 b_1(t) + \lambda_2^2 b_2(t) + \lambda_3^2 b_3(t)} ,
 \end{gather*}
 and so on.
\end{Example}

\begin{Remark} \label{rem:CH-threepeakon-parameters} In the three-peakon case, for given eigenvalues $\lambda_1$, $\lambda_2$, $\lambda_3$, the translations $x \mapsto x-x_0$, $t \mapsto t-t_0$ can be used to make two out of the three parameters $b_1(0)$, $b_2(0)$, $b_3(0)$ take any values that we like, but the third one remains as an essential parameter. That is, unlike the two-peakon case, where (up to translation) there is a unique solution for each pair of eigenvalues, there is a one-parameter family of inequivalent three-peakon solutions for each triple of eigenvalues. By taking
 \begin{gather}
 \lambda_k = \frac{1}{c_k} \qquad\text{with}\qquad c_1 > c_2 > c_3, \qquad c_1 \neq 0, \qquad c_2 \neq 0, \qquad c_3 \neq 0,\nonumber \\
 b_1(0) = \frac{c_1^2}{(c_1-c_2)(c_1-c_3)} ,\qquad
 b_2(0) = \frac{c_2^2}{(c_1-c_2)(c_2-c_3)} {\rm e}^K ,\nonumber\\
 b_3(0) = \frac{c_3^2}{(c_1-c_3)(c_2-c_3)},\label{eq:CH-threepeakon-symmetric}
 \end{gather}
we obtain the three-peakon solution in a form which depends on the asymptotic velocities~$c_k$ and one more essential parameter $K \in \R$, and which has the property that the asymptotes for the curves $x=x_1(t)$ and $x=x_3(t)$ lie symmetrically with respect to the origin, whereas the asymptotes for the curve $x=x_2(t)$ are symmetrically placed with respect to the point $(x,t)=(K,0)$. Changing the sign of~$K$ corresponds to the reversal $(x,t) \mapsto (-x,-t)$, so $K=0$ gives a three-peakon solution which is symmetric with respect to this reversal.
\end{Remark}

\subsection{Degasperis--Procesi peakons}\label{sec:DP-N}

Next, we present the $N$-peakon solution formulas for the Degasperis--Procesi equation. Readers who are only interested in the Camassa--Holm case can proceed directly to Section~\ref{sec:CH-ghost} about Camassa--Holm ghostpeakons.

For $N=1$ and $N=2$, the solution of the Degasperis--Procesi $N$-peakon ODEs \eqref{eq:DP-peakon-ODEs-shorthand} was found using direct integration by Degasperis, Holm and Hone~\cite{degasperis-holm-hone:2002:new-integrable-equation-DP}, and the general solution for arbitrary~$N$ was derived by Lundmark and Szmigielski~\cite{lundmark-szmigielski:2005:DPlong} using inverse spectral methods. To avoid certain complications, we consider first the pure peakon case where all amplitudes~$m_k$ are \emph{positive}. Then the solution is given in terms of spectral data $\{ \lambda_k, b_k \}_{k=1}^N$ where the eigenvalues~$\lambda_k$ are positive and distinct, and the residues $b_k$ are positive and have the time dependence $b_k(t) = b_k(0) {\rm e}^{t/\lambda_k}$, just like for Camassa--Holm peakons. The solution formulas are
\begin{gather} \label{eq:DP-generalsolution}
 x_{N+1-k}(t) = \ln \frac{U_k^0}{U_{k-1}^1}, \qquad m_{N+1-k}(t) = \frac{\bigl( U_k^0 U_{k-1}^1 \bigr)^2}{W_k W_{k-1}} ,
\end{gather}
for $k=1,\dots,N$, where
\begin{gather} \label{eq:Uk}
 U_k^a = \begin{cases}
 \displaystyle \sum_{1 \le i_1 < \dotsb < i_k \le N} \left( \prod_{r=1}^k \lambda_{i_r}^a b_{i_r} \right) \frac{\Delta( \lambda_{i_1}, \dots, \lambda_{i_k} )^2}{\Gamma( \lambda_{i_1}, \dots, \lambda_{i_k})} ,& 1 \le k \le N ,\vspace{1mm}\\
 1,& k=0 ,\\
 0,& \text{otherwise} ,
 \end{cases}
\end{gather}
and
\begin{gather} \label{eq:Wk}
 W_k = \begin{vmatrix} U_{k}^0 & U_{k-1}^1 \\ U_{k+1}^0 & U_{k}^1 \end{vmatrix},
\end{gather}
with $\Gamma$ similar to the Vandermonde determinant~$\Delta$ in~\eqref{eq:Vandermonde}, but with plus instead of minus:
\begin{gather*}
 \Gamma(x_1,x_2,\dots,x_k) = \prod_{1 \le i < j \le k} (x_i + x_j) .
\end{gather*}
The superscript $a$ in $U_k^a$ is just a label, whereas in $\lambda_i^a$ it denotes the $a$th power of the constant~$\lambda_i$.

\begin{Example}[DP two-peakon solution] \label{ex:DP-twopeakon-solution} For convenience, let us write
 \begin{gather*}
 U_k = U_k^0 ,\qquad V_k = U_k^1 ,
 \end{gather*}
 which is the original notation used by Lundmark and Szmigielski~\cite{lundmark-szmigielski:2005:DPlong}. Then the Degasperis--Procesi two-peakon solution is
 \begin{gather}
 x_1(t) = \ln\frac{U_2}{V_1} = \ln \frac{\frac{(\lambda_1-\lambda_2)^2}{\lambda_1+\lambda_2}b_1 b_2}{\lambda_1 b_1 + \lambda_2 b_2} ,\nonumber \\
 x_2(t) = \ln\frac{U_1}{V_0} = \ln (b_1+b_2) ,\nonumber \\
 m_1(t) = \frac{(U_2 V_1)^2}{W_2 W_1} = \frac{(\lambda_1 b_1 + \lambda_2 b_2)^2}{\lambda_1 \lambda_2 \big( \lambda_1 b_1^2 + \lambda_2 b_2^2 + \frac{4 \lambda_1 \lambda_2}{\lambda_1+\lambda_2}b_1 b_2 \big)} , \nonumber\\
 m_2(t) = \frac{(U_1 V_0)^2}{W_1 W_0} = \frac{(b_1+b_2)^2}{\lambda_1 b_1^2 + \lambda_2 b_2^2 + \frac{4 \lambda_1 \lambda_2}{\lambda_1+\lambda_2}b_1 b_2} ,\label{eq:DP-twopeakon-solution}
 \end{gather}
 where $b_k = b_k(t) = b_k(0) {\rm e}^{t/\lambda_k}$.
\end{Example}

\begin{Remark} \label{rem:DP-twopeakon-parameters} Just as for the Camassa--Holm equation (Remark~\ref{rem:CH-twopeakon-parameters}), the eigenvalues are the only essential parameters in the Degasperis--Procesi two-peakon solution; changing $b_1(0)$ and $b_2(0)$ only amounts to a translation in the $(x,t)$ plane. With
 \begin{gather*}
 c_1 = \frac{1}{\lambda_1} > c_2 = \frac{1}{\lambda_2} > 0 ,\qquad b_1(0) = \frac{\sqrt{c_1 (c_1+c_2)}}{c_1-c_2} ,\qquad b_2(0) = \frac{\sqrt{c_2(c_1+c_2)}}{c_1-c_2} ,
 \end{gather*}
 the pure two-peakon solution takes the symmetric form
 \begin{gather*}
 x_2(t) = -x_1(-t) = \frac12\ln\frac{c_1+c_2}{(c_1-c_2)^2} + \ln \bigl( \sqrt{c_1} {\rm e}^{c_1 t} + \sqrt{c_2} {\rm e}^{c_2 t} \bigr) ,\\
 m_2(t) = m_1(-t) = \frac{\bigl( \sqrt{c_1} {\rm e}^{c_1 t} + \sqrt{c_2} {\rm e}^{c_2 t} \bigr)^2}{{\rm e}^{2c_1 t} + {\rm e}^{2c_2 t} + \frac{4 \sqrt{c_1 c_2}}{c_1+c_2} {\rm e}^{(c_1+c_2)t}} .
 \end{gather*}
\end{Remark}

\begin{Example}[DP three-peakon solution] \label{ex:DP-threepeakon-solution} For $N=3$ the solution becomes
 \begin{alignat*}{3}
 & x_1(t) = \ln\frac{U_3}{V_2},\qquad && m_1(t) = \frac{(U_3 V_2)^2}{W_3 W_2} = \frac{(V_2)^2}{\lambda_1 \lambda_2 \lambda_3 W_2},& \\
 & x_2(t) = \ln\frac{U_2}{V_1},\qquad && m_2(t) = \frac{(U_2 V_1)^2}{W_2 W_1},& \\
& x_3(t) = \ln\frac{U_1}{V_0} = \ln U_1,\qquad && m_3(t) = \frac{(U_1 V_0)^2}{W_1 W_0} = \frac{(U_1)^2}{W_1},&
 \end{alignat*}
 where
 \begin{gather*}
 U_{-1} = V_{-1} = 0, \qquad  U_0 = V_0 = 1, \qquad  U_1 = b_1+b_2+b_3, \qquad
 V_1 = \lambda_1 b_1+\lambda_2 b_2+\lambda_3 b_3, \\
 U_2 = \frac{(\lambda_1-\lambda_2)^2}{\lambda_1+\lambda_2} b_1 b_2 +\frac{(\lambda_1-\lambda_3)^2}{\lambda_1+\lambda_3} b_1 b_3 +\frac{(\lambda_2-\lambda_3)^2}{\lambda_2+\lambda_3} b_2 b_3, \\
 V_2 = \frac{(\lambda_1-\lambda_2)^2}{\lambda_1+\lambda_2} \lambda_1 \lambda_2 b_1 b_2 +\frac{(\lambda_1-\lambda_3)^2}{\lambda_1+\lambda_3} \lambda_1 \lambda_3 b_1 b_3 +\frac{(\lambda_2-\lambda_3)^2}{\lambda_2+\lambda_3} \lambda_2 \lambda_3 b_2 b_3, \\
 U_3 = \frac{(\lambda_1-\lambda_2)^2 (\lambda_1-\lambda_3)^2 (\lambda_2-\lambda_3)^2} {(\lambda_1+\lambda_2) (\lambda_1+\lambda_3) (\lambda_2+\lambda_3)} b_1 b_2 b_3, \\
 V_3 = \frac{(\lambda_1-\lambda_2)^2 (\lambda_1-\lambda_3)^2 (\lambda_2-\lambda_3)^2} {(\lambda_1+\lambda_2) (\lambda_1+\lambda_3) (\lambda_2+\lambda_3)} \lambda_1 \lambda_2 \lambda_3 b_1 b_2 b_3, \\
 U_4 = V_4 = 0,
\end{gather*}
 and consequently
 \begin{gather*}
 W_0 = 1, \\
 W_1 = U_1 V_1 - U_2 V_0 = \lambda_1 b_1^2 + \lambda_2 b_2^2 + \lambda_3 b_3^2 + \frac{4 \lambda_1 \lambda_2}{\lambda_1+\lambda_2} b_1 b_2 +
 \frac{4 \lambda_1 \lambda_3}{\lambda_1+\lambda_3} b_1 b_3 + \frac{4 \lambda_2 \lambda_3}{\lambda_2+\lambda_3} b_2 b_3, \\
 W_2 = U_2 V_2 - U_3 V_1 = \frac{(\lambda_1-\lambda_2)^4}{(\lambda_1+\lambda_2)^2} \lambda_1 \lambda_2 (b_1 b_2)^2
 + \frac{(\lambda_1-\lambda_3)^4}{(\lambda_1+\lambda_3)^2} \lambda_1 \lambda_3 (b_1 b_3)^2\nonumber\\
 \hphantom{W_2 =}{} + \frac{(\lambda_2-\lambda_3)^4}{(\lambda_2+\lambda_3)^2} \lambda_2 \lambda_3 (b_2 b_3)^2
 + \frac{4 \lambda_1 \lambda_2 \lambda_3 b_1 b_2 b_3} {(\lambda_1+\lambda_2)(\lambda_1+\lambda_3)(\lambda_2+\lambda_3)}\\
 \hphantom{W_2 =}{}
 \times \bigl( (\lambda_1-\lambda_2)^2 (\lambda_1-\lambda_3)^2 b_1 + (\lambda_2-\lambda_1)^2 (\lambda_2-\lambda_3)^2 b_2 + (\lambda_3-\lambda_1)^2 (\lambda_3-\lambda_2)^2 b_3 \bigr) , \\
 W_3 = U_3 V_3 = \lambda_1 \lambda_2 \lambda_3 (U_3)^2 .
 \end{gather*}
\end{Example}

\begin{Remark} \label{rem:DP-threepeakon-parameters} A symmetric way of writing the Degasperis--Procesi pure three-peakon solution, analogous to what we saw in Remark~\ref{rem:CH-threepeakon-parameters} for the Camassa--Holm equation, is obtained by taking
\begin{gather*}
\lambda_k = \frac{1}{c_k} \qquad\text{with}\quad c_1 > c_2 > c_3 > 0 ,\\
b_1(0) = \frac{c_1 \sqrt{(c_1+c_2)(c_1+c_3)}}{(c_1-c_2)(c_1-c_3)} ,\qquad
b_2(0) = \frac{c_2 \sqrt{(c_1+c_2)(c_2+c_3)}}{(c_1-c_2)(c_2-c_3)} {\rm e}^K ,\\
b_3(0) = \frac{c_3 \sqrt{(c_1+c_3)(c_2+c_3)}}{(c_1-c_3)(c_2-c_3)} .
\end{gather*}
In this form, the solution depends on the asymptotic velocities $c_1$, $c_2$, $c_3$, together with one more essential parameter~$K\in \R$, where $K=0$ gives a solution symmetric under the reversal $(x,t) \mapsto (-x,-t)$.
\end{Remark}

\begin{Remark} \label{rem:DP-antipeakons} The solution formulas \eqref{eq:DP-generalsolution} also work for pure antipeakon solutions, with all amplitudes negative, the only difference being that all $\lambda_k$ are negative in this case. The situation for mixed peakon--antipeakon solutions is considerably more complicated. To begin with, the peakon solution formulas, with parameters determined by initial data $\{ x_k(0), m_k(0) \}_{k=1}^N$, only provide the solution of the initial-value problem up until the first peakon--antipeakon collision, since at that time the PDE solution $u(x,t)$ develops a jump discontinuity, so to continue past the collision one has to go outside the peakon world and consider \emph{shockpeakons}~\cite{lundmark:2007:shockpeakons} (cf.~Examples~\ref{ex:DP-p-ap-symm} and~\ref{ex:DP-p-ap-asymm} below). Secondly, the eigenvalues~$\lambda_k$ (which in the Degasperis--Procesi case are eigenvalues of a non-symmetric matrix) need not be real and simple anymore. If the eigen\-va\-lues are complex and simple, and $\lambda_i + \lambda_j \neq 0$ for all $i$ and~$j$, then the formulas work without modification; the eigenvalues occur in complex-conjugate pairs, and whenever $\lambda_j = \overline{\lambda_i}$, then also $b_j = \overline{b_i}$, which will make all $U_k^a$ real-valued, and the formulas provide a solution of the initial value problem until the first collision. If there are eigenvalues of multiplicity greater than one, then the formulas must be modified, taking into account that the partial fraction decomposition of the Weyl function will involve coefficients $b_k^{(i)}$ whose time dependence is given by a polynomial in~$t$ times the exponential ${\rm e}^{t/\lambda_k}$; see Szmigielski and Zhou~\cite{szmigielski-zhou:2013:DP-colliding-peakons-shock-formation,szmigielski-zhou:2013:DP-peakon-antipeakon}. There is also the possibility of resonant cases, where $\lambda_i + \lambda_j = 0$ for one or more pairs $(i,j)$. Such cases can be handled using a limiting procedure, but as far as we are aware the resulting formulas have not been published, with the exception of the symmetric two-peakon solution with $\lambda_1+\lambda_2=0$, which is easily computed by direct integration (see Example~\ref{ex:DP-p-ap-symm} below).
\end{Remark}

\subsection{Novikov peakons}\label{sec:Novikov-N}

As far as pure peakon solutions are concerned, Novikov's equation is fairly similar to the Degasperis--Procesi equation. The $N$-peakon solutions are governed by the ODEs~\eqref{eq:Novikov-peakon-ODEs-shorthand}, and the general solution for positive amplitudes $m_k$ was derived by Hone, Lundmark and Szmigielski~\cite{hone-lundmark-szmigielski:2009:novikov}. It is once again given in terms of spectral data $\{ \lambda_k, b_k \}_{k=1}^N$ where the eigenvalues $\lambda_k$ are positive and distinct, and the residues $b_k$ are positive and have time dependence $b_k(t) = b_k(0) {\rm e}^{t/\lambda_k}$.

Recall the notation $U_k^a$ and $W_k$ from \eqref{eq:Uk} and~\eqref{eq:Wk}, and also let
\begin{gather*}
 Z_k = \begin{vmatrix} U_{k}^{-1} & U_{k-1}^0 \\ U_{k+1}^{-1} & U_{k}^0 \end{vmatrix} = \begin{vmatrix} T_{k} & U_{k-1} \\ T_{k+1} & U_{k} \end{vmatrix},
\end{gather*}
where $T_k = U_k^{-1}$ and $U_k = U_k^0$ is the notation used by Hone, Lundmark and Szmigielski~\cite{hone-lundmark-szmigielski:2009:novikov}. In terms of these quantities, the solution formulas are
\begin{gather}
 \label{eq:Novikov-generalsolution}
 x_{N+1-k}(t) = \frac12 \ln\frac{Z_{k}}{W_{k-1}}
 ,\qquad
 m_{N+1-k}(t) = \frac{\sqrt{Z_{k} W_{k-1}}}{U_{k} U_{k-1}}
 ,
\end{gather}
for $k=1,\dots,N$.
The fact that $W_k$ and $Z_k$ are positive follows in the pure peakon
case from an explicit combinatorial formula due to Lundmark and
Szmigielski which expresses them as sums of positive quantities~\cite[Lemma 2.20]{lundmark-szmigielski:2005:DPlong}.
A different argument, which also works in the mixed peakon--antipeakon case,
is given by Kardell and Lundmark~\cite{kardell-lundmark:2016p:novikov-peakon-antipeakon}.

\begin{Example}[Novikov two-peakon solution] \label{ex:Novikov-twopeakon-solution} The Novikov two-peakon solution is
 \begin{gather*}
x_1(t) = \frac12 \ln\frac{Z_2}{W_1} = \frac12 \ln \frac{ \frac{(\lambda_1-\lambda_2)^4}{(\lambda_1+\lambda_2)^2 \lambda_1 \lambda_2} b_1^2 b_2^2}{ \lambda_1 b_1^2 + \lambda_2 b_2^2 + \frac{4 \lambda_1 \lambda_2}{\lambda_1+\lambda_2} b_1 b_2} , \\
 x_2(t) = \frac12 \ln\frac{Z_1}{W_0} = \frac12 \ln \left( \frac{b_1^2}{\lambda_1} + \frac{b_2^2}{\lambda_2} + \frac{4}{\lambda_1+\lambda_2} b_1 b_2 \right) , \\
 m_1(t) = \frac{\sqrt{Z_2 W_1}}{U_2 U_1} = \frac{ \left[ \frac{(\lambda_1 - \lambda_2)^4 b_1^2 b_2^2}{(\lambda_1 + \lambda_2)^2 \lambda_1 \lambda_2} \left( \lambda_1 b_1^2 + \lambda_2 b_2^2 + \frac{4 \lambda_1 \lambda_2}{\lambda_1+\lambda_2} b_1 b_2 \right) \right]^{1/2}}{ \frac{(\lambda_1 - \lambda_2)^2 b_1 b_2}{\lambda_1 + \lambda_2} (b_1+b_2)} , \\
 m_2(t) = \frac{\sqrt{Z_1 W_0}}{U_1 U_0}= \frac{\left( \frac{b_1^2}{\lambda_1} + \frac{b_2^2}{\lambda_2} + \frac{4}{\lambda_1+\lambda_2} b_1 b_2 \right)^{1/2}}{b_1+b_2} ,
 \end{gather*}
 where the expression for $m_1$ can be simplified to
 \begin{gather*}
 m_1(t) = \frac{ \left( \lambda_1 b_1^2 + \lambda_2 b_2^2 + \frac{4 \lambda_1 \lambda_2}{\lambda_1+\lambda_2} b_1 b_2 \right)^{1/2}}{\sqrt{\lambda_1 \lambda_2} (b_1+b_2)}
 \end{gather*}
 in the pure peakon case, since then all spectral data are positive.
\end{Example}

\begin{Remark} \label{rem:Novikov-antipeakons} Pure antipeakon solutions are obtained from pure peakon solutions simply by keeping all~$x_k(t)$ and changing the signs of all~$m_k(t)$, which in terms of spectral data is accomplished by keeping all~$\lambda_k$ and changing the signs of all~$b_k(t)$. Note that both peakons and antipeakons move to the right, since $\dot x_k = u(x_k)^2 \ge 0$. Mixed peakon--antipeakon solutions have been studied in detail by Kardell and Lundmark~\cite{kardell-lundmark:2016p:novikov-peakon-antipeakon}, and it turns out that the behaviour of the Novikov equation differs from that of the Camassa--Holm and Degasperis--Procesi equations. The eigenvalues~$\lambda_k$ may be complex, but as long as they are simple and have positive real part (which is the generic case), the solutions will still be described by the same formulas~\eqref{eq:Novikov-generalsolution} as in the pure peakon case. Despite everything moving to the right, there will be peakon--antipeakon collisions where a faster peakon (or antipeakon) catches up with a slower antipeakon (or peakon). At the collision, the corresponding amplitudes blow up to~$\pm\infty$, but the wave profile $u(x,t)$, which is defined by the formulas for $x_k(t)$ and $m_k(t)$ for all~$t$ except the instants of collision (which are isolated), extends to a continuous function defined for all $t \in \R$, providing a global conservative weak solution; cf.\ Remark~\ref{rem:DP-Novikov-references}, and see Example~\ref{ex:Novikov-1p-4cluster} for an illustration. For some initial conditions, the eigenvalues have multiplicity greater than one. The solution formulas in those non-generic cases are also known, although we will not state them here, since that would require quite a lot of additional notation. For complex eigenvalues, the peakon--antipeakon solutions exhibit periodic or quasi-periodic behaviour, with peakons colliding, separating, and colliding again, infinitely many times. The eigenvalues always lie in the right half of the complex plane, $\operatorname{Re} \lambda_k \ge 0$. For $N \ge 3$ there may be conjugate pairs lying on the imaginary axis, leading to resonances where some $\lambda_i + \lambda_j = 0$; that non-generic case is also covered by modified solution formulas which are known (but not described here).
\end{Remark}

This concludes our review of the known solution formulas for peakons (with nonzero amplitude), and we now turn to our new results about ghostpeakons and characteristic curves.

\section{Camassa--Holm ghostpeakons}\label{sec:CH-ghost}

In this section, we will state and prove the explicit formulas for Camassa--Holm multipeakon solutions where one or several amplitudes $m_k(t)$ are identically zero. As explained above Defi\-ni\-tion~\ref{def:ghostpeakon}, ghostpeakons do not interact with each other, so it is enough to find the solution for the case with $N$ ordinary peakons and just \emph{one} ghostpeakon, say at position $N+1-p$ for some $0 \le p \le N$, which is the case treated in Theorem~\ref{thm:CH-ghost}. If there are several ghostpeakons, the position of each one of them can be obtained by disregarding the other ghostpeakons and applying Theorem~\ref{thm:CH-ghost} to the particular ghostpeakon in question.

The trajectories $x = x_k(t)$ of the ghostpeakons
are characteristic curves for the multipeakon solution
$u(x,t) = \sum m_i {\rm e}^{-\abs{x-x_i}}$
formed by the nonzero-amplitude peakons;
see the discussion in connection with equation~\eqref{eq:u-with-ghostpeakon-at-site-k}.
Corollary~\ref{cor:CH-add-a-ghostpeakon} reformulates Theorem~\ref{thm:CH-ghost}
from that point of view.

Recall the notation $\Delta_k^a=\Delta_k^a(N)$ from \eqref{eq:Delta-evaluated}.

\begin{Theorem} \label{thm:CH-ghost} Fix some $p$ with $0 \le p \le N$. The solution of the Camassa--Holm $(N+1)$-peakon ODEs \eqref{eq:CH-peakon-ODEs} with $x_1 < \dotsb < x_{N+1}$ and all amplitudes $m_k(t)$ nonzero except for $m_{N+1-p}(t) = 0$ is as follows: the position of the ghostpeakon is given by
 \begin{gather} \label{eq:CH-ghost-general}
 x_{N+1-p}(t) = \ln \frac{\Delta_{p+1}^0 + \theta \Delta_{p}^0}{\Delta_{p}^2 + \theta \Delta_{p-1}^2 } ,\qquad 0 < \theta < \infty ,
 \end{gather}
 while the other peakons are given by the $N$-peakon solution formulas~\eqref{eq:CH-generalsolution} up to relabeling $($shift the index by one for the peakons to the right of the ghostpeakon$)$:
 \begin{gather} \label{eq:CH-N-peakons-renumbered}
 x_{N+1-k}(t) =
 \begin{cases}
 \ln \dfrac{\Delta_{k+1}^0}{\Delta_{k}^2}, & 0 \le k < p, \vspace{1mm}\\
 \ln \dfrac{\Delta_{k}^0}{\Delta_{k-1}^2}, & p < k \le N,
 \end{cases}\qquad
 m_{N+1-k}(t) =
 \begin{cases}
 \dfrac{\Delta_{k+1}^0 \Delta_{k}^2}{\Delta_{k+1}^1 \Delta_{k}^1 }, & 0 \le k < p, \vspace{1mm}\\
 \dfrac{\Delta_{k}^0 \Delta_{k-1}^2}{\Delta_{k}^1 \Delta_{k-1}^1 }, & p < k \le N.
 \end{cases}
 \end{gather}
 Here $\theta \in (0,\infty)$ is a constant in one-to-one correspondence with the ghostpeakon's initial position $x_{N+1-p}(0)$, while the quantities $\{ \lambda_k, b_k \}_{k=1}^{N}$ appearing in the expressions $\Delta_k^a = \Delta_k^a(N)$ have the usual time dependence $\lambda_k = {\rm const}$, $b_k(t) = b_k(0) {\rm e}^{t/\lambda_k}$.
\end{Theorem}

\begin{proof} Let
 \begin{gather*}
 \lambda_1, \dots, \lambda_N ,\qquad b_1(0), \dots, b_N(0)
 \end{gather*}
be the spectral data corresponding to the initial data $\{ x_k(0), m_k(0) \}_{k \neq p}$ of the peakons with nonzero amplitudes, and consider an $(N+1)$-peakon solution obtained by augmenting these spectral data with some arbitrary positive constants $\lambda_{N+1}$ and $b_{N+1}(0)$. (We will later let these constants tend to $+\infty$ and $0$, respectively, with $\lambda_{N+1}^{2p} b_{N+1}(0)$ held constant.) The formulas for this $(N+1)$-peakon solution are given by \eqref{eq:CH-generalsolution} with $N+1$ instead of~$N$ (and $k+1$ instead of~$k$):
 \begin{gather*}
 x_{N+1-k} = \ln\frac{\tilde \Delta_{k+1}^0}{\tilde \Delta_{k}^2} ,\qquad m_{N+1-k} = \frac{\tilde \Delta_{k+1}^0 \tilde \Delta_{k}^2}{\tilde \Delta_{k+1}^1 \tilde \Delta_{k}^1} ,\qquad 0 \le k \le N ,
 \end{gather*}
 where $\tilde \Delta_k^a = \Delta_k^a(N+1)$. These formulas depend on the constant parameters
 \begin{gather*}
 \lambda_1, \dots, \lambda_N, \lambda_{N+1} ,\qquad b_1(0), \dots, b_N(0), b_{N+1}(0) ,
 \end{gather*}
where the $\lambda_k$ are real, distinct and nonzero and all $b_k(0)$ are positive. But we may just as well express the solution in terms of the equivalent set of parameters
 \begin{gather*}
 \lambda_1, \dots, \lambda_N, \varepsilon ,\qquad b_1(0), \dots, b_N(0), \theta ,
 \end{gather*}
 where
 \begin{gather*}
 \varepsilon = \frac{1}{\lambda_{N+1}} ,\qquad \theta = \lambda_{N+1}^{2p} b_{N+1}(0) .
 \end{gather*}
 Thus, $\lambda_{N+1}=1/\varepsilon$ and
 \begin{gather*}
 b_{N+1} = b_{N+1}(t) = b_{N+1}(0) {\rm e}^{t/\lambda_{N+1}} = \varepsilon^{2p} \theta {\rm e}^{\varepsilon t} = \varepsilon^{2p} \Theta ,
 \end{gather*}
 where
 \begin{gather*}
 \Theta = \Theta(t) = \theta {\rm e}^{\varepsilon t} .
 \end{gather*}
When we perform these substitutions in the definition \eqref{eq:Delta-evaluated} of $\tilde \Delta_k^a = \Delta_k^a(N+1)$, split the sum according to whether $i_k \le N$ or $i_k = N+1$, and write $\Delta_k^a = \Delta_k^a(N)$, we obtain (for $1 \le k \le N+1$)
 \begin{gather*}
 \tilde \Delta_k^a = \sum_{1 \le i_1 < \dotsb < i_k \le N+1} \left( \prod_{r=1}^k \lambda_{i_r}^a b_{i_r} \right) \Delta ( \lambda_{i_1}, \dots, \lambda_{i_k} )^2 \\
 \hphantom{\tilde \Delta_k^a}{} =
 \sum_{1 \le i_1 < \dotsb < i_k \le N} \left( \prod_{r=1}^k \lambda_{i_r}^a b_{i_r} \right) \Delta ( \lambda_{i_1}, \dots, \lambda_{i_k} )^2 \\
 \hphantom{\tilde \Delta_k^a=}{} + \sum_{1 \le i_1 < \dotsb < i_{k-1} \le N} \left( \prod_{r=1}^{k-1} \lambda_{i_r}^a b_{i_r} \right) \Delta ( \lambda_{i_1}, \dots, \lambda_{i_{k-1}} )^2 \lambda_{N+1}^a b_{N+1} \prod_{s=1}^{k-1} (\lambda_{i_s} - \lambda_{N+1})^2 \\
 \hphantom{\tilde \Delta_k^a}{} = \Delta_k^a + \sum_{1 \le i_1 < \dotsb < i_{k-1} \le N} \left( \prod_{r=1}^{k-1} \lambda_{i_r}^a b_{i_r} \right) \Delta ( \lambda_{i_1}, \dots, \lambda_{i_{k-1}} )^2
 \varepsilon^{-a} \varepsilon^{2p} \Theta \prod_{s=1}^{k-1} \frac{(\varepsilon \lambda_{i_s} - 1)^2}{\varepsilon^2} \\
\hphantom{\tilde \Delta_k^a}{} = \Delta_k^a + \sum_{1 \le i_1 < \dotsb < i_{k-1} \le N} \left( \prod_{r=1}^{k-1} \lambda_{i_r}^a b_{i_r} \right) \Delta ( \lambda_{i_1}, \dots, \lambda_{i_{k-1}} )^2
 \Theta \frac{\varepsilon^{2p-a} }{(\varepsilon^{2})^{k-1}} \bigl( 1 + \order{\varepsilon} \bigr) \\
\hphantom{\tilde \Delta_k^a}{} = \Delta_k^a + \Delta_{k-1}^a \Theta \varepsilon^{2(p-k+1)-a} ( 1 + \order{\varepsilon}) \qquad (\text{as $\varepsilon \to 0$}) .
 \end{gather*}
 Consequently, expressed in terms of the new parameters (and in particular considered as being functions of~$\varepsilon$) the positions and amplitudes in the $(N+1)$-peakon solution take the form
 \begin{gather*}
 x_{N+1-k}(t;\varepsilon) = \ln\frac{\tilde \Delta_{k+1}^0}{\tilde \Delta_{k}^2}
 = \ln\frac{\Delta_{k+1}^0 + \Delta_{k}^0 \Theta \varepsilon^{2(p-k)}( 1 + \order{\varepsilon} )}{\Delta_{k}^2 + \Delta_{k-1}^2 \Theta \varepsilon^{2(p-k)} ( 1 + \order{\varepsilon} )} \\
 \hphantom{x_{N+1-k}(t;\varepsilon)}{} =
 \begin{cases}
 \ln\dfrac{\Delta_{k+1}^0 + \order{\varepsilon}}{\Delta_{k}^2 + \order{\varepsilon}}, & k < p, \vspace{1mm}\\
 \ln\dfrac{\Delta_{p+1}^0 + \Delta_{p}^0 \Theta + \order{\varepsilon}}{\Delta_{p}^2 + \Delta_{p-1}^2 \Theta + \order{\varepsilon}}, & k = p, \vspace{1mm}\\
 \ln\dfrac{\Delta_{k}^0 \Theta + \order{\varepsilon}}{\Delta_{k-1}^2 \Theta + \order{\varepsilon}}, & k > p,
 \end{cases}
 \end{gather*}
 and
 \begin{gather*}
 m_{N+1-k}(t;\varepsilon) = \frac{\tilde \Delta_{k+1}^0 \tilde \Delta_{k}^2}{\tilde \Delta_{k+1}^1 \tilde \Delta_{k}^1} \\
 \hphantom{m_{N+1-k}(t;\varepsilon)}{} =
 \frac{\Delta_{k+1}^0 + \Delta_{k}^0 \Theta \varepsilon^{2(p-k)} ( 1 + \order{\varepsilon} )}{\Delta_{k+1}^1 + \Delta_{k}^1 \Theta \varepsilon^{2(p-k)-1} ( 1 + \order{\varepsilon} )}
 \frac{\Delta_{k}^2 + \Delta_{k-1}^2 \Theta \varepsilon^{2(p-k)} ( 1 + \order{\varepsilon})}{\Delta_{k}^1 + \Delta_{k-1}^1 \Theta \varepsilon^{2(p-k+1)-1} ( 1 + \order{\varepsilon} )} \\
\hphantom{m_{N+1-k}(t;\varepsilon)}{} =
 \begin{cases}
 \dfrac{\Delta_{k+1}^0 \Delta_{k}^2 + \order{\varepsilon}}{\Delta_{k+1}^1 \Delta_{k}^1 + \order{\varepsilon}}, & k < p, \vspace{1mm}\\
 \dfrac{\varepsilon \bigl( \Delta_{p+1}^0 + \Delta_{p}^0 \Theta + \order{\varepsilon} \bigr) \bigl( \Delta_{p}^2 + \Delta_{p-1}^2 \Theta + \order{\varepsilon} \bigr)}{\bigl( \Delta_{p}^1 \Theta + \order{\varepsilon} \bigr) \bigl( \Delta_{p}^1 + \order{\varepsilon} \bigr)} , & k = p, \vspace{1mm}\\
 \dfrac{\Theta^2 \Delta_{k}^0 \Delta_{k-1}^2 + \order{\varepsilon}}{\Theta^2 \Delta_{k}^1 \Delta_{k-1}^1 + \order{\varepsilon}}, & k > p.
 \end{cases}
 \end{gather*}
 In the limit $\varepsilon \to 0^+$, these expressions reduce to those given in \eqref{eq:CH-ghost-general} and \eqref{eq:CH-N-peakons-renumbered}. Note in particular that $m_{N+1-p}(t;0) = 0$, i.e., we really kill the peakon at that position, as claimed.

The expressions which remain after letting $\varepsilon \to 0$ still satisfy the peakon ODEs. Indeed, this is a purely differential-algebraic issue. Due to the ordering property $x_1 < \dots < x_{N+1}$, we can remove the absolute value signs in the ODEs, and then the satisfaction of the ODEs by the proposed solution formulas boils down to certain meromorphic functions of~$t$ and of the parameters being identically zero. After the reparametrization, these meromorphic functions will have removable singularities at $\varepsilon=0$, so if they are zero for all $\varepsilon \neq 0$, they will remain zero also for $\varepsilon = 0$.
\end{proof}

If we rename
\begin{gather*}
 (x_1,x_2,\dots,x_{N+1-p},\dots,x_N,x_{N+1})
\end{gather*}
to
\begin{gather*}
 (x_1,x_2,\dots,x_{\text{ghost}},\dots,x_{N-1},x_N)
\end{gather*}
and similarly for $m_k$, we see that Theorem~\ref{thm:CH-ghost} tells us how to add a ghostpeakon to a given nonzero-amplitude $N$-peakon solution. (Or as many ghostpeakons as we like, since they are independent of each other.) Which formula to use for describing the ghostpeakon's trajectory $x = x_{\text{ghost}}(t)$ depends on which pair of peakons we want it to lie between. The collection of all possible such ghostpeakon trajectories,
together with the peakon trajectories $x = x_i(t)$ themselves, constitutes the family of characteristic curves for the $N$-peakon solution $u(x,t)$, i.e., the curves $x = \xi(t)$ such that $\dot \xi(t) = u(\xi(t),t)$. So we can rephrase the theorem as the following corollary.

\begin{Corollary} \label{cor:CH-add-a-ghostpeakon} For the Camassa--Holm $N$-peakon solution given by~\eqref{eq:CH-generalsolution}, namely
 \begin{gather*}
 x_{N+1-k}(t) = \ln\frac{\Delta_k^0}{\Delta_{k-1}^2} ,\qquad m_{N+1-k}(t) = \frac{\Delta_k^0 \Delta_{k-1}^2}{\Delta_k^1 \Delta_{k-1}^1} ,\qquad 1 \le k \le N,
 \end{gather*}
 the characteristic curves $x = \xi(t)$ in the $k$th interval from the right, i.e.,
 \begin{gather} \label{eq:kth-interval-from-right}
 x_{N-k}(t) < \xi(t) < x_{N+1-k}(t)
 \end{gather}
$($where $0 \le k \le N$, $x_0 = -\infty$, $x_{N+1} = +\infty)$, are given by
 \begin{gather} \label{eq:CH-add-a-ghostpeakon}
 \xi(t) = \ln\frac{\Delta_{k+1}^0 + \theta \Delta_{k}^0}{\Delta_{k}^2 + \theta \Delta_{k-1}^2} ,\qquad \theta > 0.
 \end{gather}
\end{Corollary}

\begin{Remark} Note that the function $\xi(t)$ given by \eqref{eq:CH-add-a-ghostpeakon} ranges over all values between
 \begin{gather*}
 x_{N-k}(t) = \ln\frac{\Delta_{k+1}^0}{\Delta_{k}^2} = \lim_{\theta \to 0} \xi(t) \qquad\text{and}\qquad x_{N+1-k}(t) = \ln\frac{\Delta_k^0}{\Delta_{k-1}^2} = \lim_{\theta \to \infty} \xi(t)
 \end{gather*}
 as the parameter $\theta$ ranges over all positive numbers.
\end{Remark}

\begin{Remark} Let us say a few words about the philosophy behind our proof. The obvious first thing to try is direct integration, which is tricky but not impossible; see Example~\ref{ex:CH-N3-stuck} and Remarks~\ref{rem:Matsuno} and~\ref{rem:Matsuno-continued}. Another idea is to fix all initial data $x_i(0)$ and $m_i(0)$ except for one amplitude $m_k(0) = \varepsilon \neq 0$ that we allow to vary. Then all the spectral data $\{ \lambda_i, b_i(0) \}_{i=1}^{N+1}$ in the $(N+1)$-peakon solution formulas will be functions of~$\varepsilon$. Taking the limit as $\varepsilon \to 0$, these solution formulas must reduce to the usual (already known) $N$-peakon solution formulas for the functions $x_i(t)$ and $m_i(t)$ with $i \neq k$, plus the trivial formula $m_k(t)=0$ and one additional (previously unknown) formula which gives the ghostpeakon's position $x_k(t)$. However, this is easily doable only when $N=1$, the trivial case with one peakon and one ghostpeakon, since in this case the eigenvalues $\lambda_1$ and $\lambda_2$ are the roots of a quadratic equation with coefficients depending on the initial data (including~$\varepsilon$), so we can write $\lambda_1(\varepsilon)$ and $\lambda_2(\varepsilon)$ explicitly with formulas involving nothing worse than square roots. So instead of taking a limit in the space of physical variables, we do it the space of spectral variables: keep the parameters $\{ \lambda_i, b_i(0) \}_{i=1}^{N}$, replace the parameters $\lambda_{N+1} > 0$ and $b_{N+1}(0) > 0$ with the two new parameters $\varepsilon > 0$ and $\theta > 0$ defined by
 \begin{gather*}
 \lambda_{N+1} = \frac{1}{\varepsilon} ,\qquad b_{N+1}(0) = \varepsilon^{2p} \theta
 \end{gather*}
for some suitably chosen integer~$p$, and let $\varepsilon$ vary while keeping all other spectral parameters (including~$\theta$) fixed. All the functions $x_i(t;\varepsilon)$ and $m_i(t;\varepsilon)$ will then depend on~$\varepsilon$, and it turns out that each one of them will have a removable singularity at $\varepsilon=0$. Moreover, the value at~$\varepsilon = 0$ will be $m_k(t;0)=0$ for exactly one index~$k$, namely $k = N+1-p$, and the limiting formula for $x_k(t;0)$ then gives us the position of a ghostpeakon at that site, while the other $x_i(t;0)$ and $m_i(t;0)$ reduce to the $N$-peakon solution with parameters $\{ \lambda_i, b_i(0) \}_{i=1}^{N}$. Clearly the main difficulty lies in figuring out the suitable power of $\varepsilon$ to use in the substitution for $b_{N+1}(0)$ in order to obtain the desired result; once that is done, the verification that it works is straightforward.
\end{Remark}

\begin{Example}[CH two-peakon characteristics by direct integration] \label{ex:CH-N3-stuck} Grunert and Holden~\cite{grunert-holden:2016:CH-peakon-antipeakon-alpha-dissipative} computed the $\alpha$-dissipative continuation of the solution after a peakon--antipeakon collision (with $N=2$, i.e., there are no other peakons present). A quote from the introduction of their article emphasizes how essential it is to have explicit formulas for the characteristic curves in order to rigorously verify that the function~$u$ really satisfies the (rather involved) definition of an $\alpha$-dissipative solution:
 \begin{quote}
 ``It is somewhat surprising that the non-symmetric case allows an explicit, albeit not simple, solution. [ \ldots] The crux of the calculation is that one can solve exactly the equation for the characteristics.''
 \end{quote}
Our Corollary~\ref{cor:CH-add-a-ghostpeakon} provides these formulas immediately for any~$N$, but for comparison, let us review what the direct calculations look like for $N=2$.

 From \eqref{eq:CH-peakon-ODEs}, the Camassa--Holm three-peakon ODEs are
 \begin{alignat*}{3}
& \dot x_1 = m_1 + m_2 E_{12} + m_3 E_{13} ,\qquad && \dot m_1= m_1 ( - m_2 E_{12} - m_3 E_{13} ) , &\\
& \dot x_2 = m_1 E_{12} + m_2 + m_3 E_{23} ,\qquad && \dot m_2= m_2 ( m_1 E_{12} - m_3 E_{23} ) , &\\
& \dot x_3 = m_1 E_{13} + m_2 E_{23} + m_3 ,\qquad && \dot m_3 = m_3 ( m_1 E_{13} + m_2 E_{23} ) ,&
 \end{alignat*}
 where we assume $x_1 < x_2 < x_3$ as always, and use the abbreviations $E_{12} = {\rm e}^{x_1-x_2}$, $E_{13} = {\rm e}^{x_1-x_3}$ and $E_{23} = {\rm e}^{x_2-x_3}$.

 Suppose first that we seek the solution with a ghostpeakon at~$x_3$. With $m_3$ identically zero, the equations for $x_1(t)$, $x_2(t)$, $m_1(t)$, $m_2(t)$ reduce to the Camassa--Holm two-peakon ODEs
 \begin{alignat*}{3}
 & \dot x_1 = m_1 + m_2 E_{12} ,\qquad && \dot m_1 = - m_1 m_2 E_{12} ,&\\
 & \dot x_2 = m_1 E_{12} + m_2 ,\qquad && \dot m_2 = m_2 m_1 E_{12} ,&
 \end{alignat*}
 so we obtain those functions from the two-peakon solution formulas~\eqref{eq:CH-twopeakon-solution}:
 \begin{alignat*}{3}
 & x_1(t)= \ln \frac{(\lambda_1-\lambda_2)^2 b_1 b_2}{\lambda_1^2 b_1 + \lambda_2^2 b_2} ,\qquad && m_1(t) = \frac{\lambda_1^2 b_1 + \lambda_2^2 b_2}{\lambda_1 \lambda_2 (\lambda_1 b_1 + \lambda_2 b_2)} ,&\\
 & x_2(t) = \ln (b_1+b_2) ,\qquad && m_2(t) = \frac{b_1+b_2}{\lambda_1 b_1 + \lambda_2 b_2} , &
 \end{alignat*}
 where $b_k = b_k(t) = b_k(0) {\rm e}^{t/\lambda_k}$. The remaining equation for $x_3(t)$ becomes
 \begin{gather*}
 \dot x_3 = m_1 E_{13} + m_2 E_{23} = \bigl( m_1 {\rm e}^{x_1} + m_2 {\rm e}^{x_2} \bigr) {\rm e}^{-x_3} = \left( \frac{b_1}{\lambda_1} + \frac{b_2}{\lambda_2} \right) {\rm e}^{-x_3} ,
 \end{gather*}
 so that
 \begin{gather*}
 \frac{{\rm d}}{{\rm d}t} {\rm e}^{x_3(t)} = {\rm e}^{x_3(t)} \dot x_3(t) = \frac{1}{\lambda_1} b_1(t) + \tfrac{1}{\lambda_2} b_2(t) = \frac{1}{\lambda_1} b_1(0) {\rm e}^{t/\lambda_1} + \tfrac{1}{\lambda_2} b_2(0) {\rm e}^{t/\lambda_2} ,
 \end{gather*}
 which is easily integrated to give the desired ghostpeakon formula,
 \begin{gather} \label{eq:CH-2p-rightmost-char-directly-x3}
 x_3(t) = \ln\bigl( b_1(0) {\rm e}^{t/\lambda_1} + b_2(0) {\rm e}^{t/\lambda_2} + \theta \bigr) = \ln\bigl( b_1(t) + b_2(t) + \theta \bigr) .
 \end{gather}
Comparison with the formula $x_2(t) = \ln ( b_1(t)+b_2(t))$ shows that the constant of integration~$\theta$ must be positive, in order for $x_2 < x_3$ to hold. When $\theta$ runs through all positive values,
 equation~\eqref{eq:CH-2p-rightmost-char-directly-x3} thus gives the family of characteristic curves $x=\xi(t)$ for the two-peakon solution in the outer right region $x > x_2(t)$ in the $(x,t)$-plane:
 \begin{gather*}
 \xi(t) = \ln\bigl( \Delta_{1}^0 + \theta \bigr) = \ln\bigl( b_1 + b_2 + \theta \bigr) ,\qquad \theta > 0 .
 \end{gather*}
 The case with a ghostpeakon at~$x_1$ is just as easy. However, the middle case with a ghostpeakon at~$x_2$ leads to the equation
 \begin{gather} \label{eq:CH-N3-stuck}
 \dot x_2 = m_1 E_{12} + m_3 E_{23} = m_1(t) {\rm e}^{x_1(t)} {\rm e}^{-x_2} + m_3(t) {\rm e}^{ - x_3(t)} {\rm e}^{x_2} ,
 \end{gather}
where the functions $x_1(t)$, $x_3(t)$, $m_1(t)$, $m_3(t)$ are explicitly given by the two-peakon solution formulas~\eqref{eq:CH-twopeakon-solution} after relabeling, with $x_3$ and $m_3$ taking the places of~$x_2$ and~$m_2$. Here we have~${\rm e}^{-x_2}$ in one of the terms on the right-hand side and~${\rm e}^{x_2}$ in the other term, so the variables $x_2$ and~$t$ do not separate, but Grunert and Holden~\cite[Section~4, below~(4.14)]{grunert-holden:2016:CH-peakon-antipeakon-alpha-dissipative} managed to integrate equation~\eqref{eq:CH-N3-stuck} via a sequence of rather ingenious substitutions. In our notation, the resulting formula for the characteristics in the region between the two peakons is
 \begin{gather*}
 \xi(t) = \ln\frac{\Delta_{2}^0 + \theta \Delta_{1}^0}{\Delta_{1}^2 + \theta \Delta_{0}^2} = \ln\frac{(\lambda_1-\lambda_2)^2 b_1 b_2 + \theta (b_1+b_2)}{\lambda_1^2 b_1 + \lambda_2^2 b_2 + \theta} ,\qquad \theta > 0 .
 \end{gather*}
 (Note that this reduces to $x_1(t) = \ln \frac{(\lambda_1-\lambda_2)^2 b_1 b_2}{\lambda_1^2 b_1 + \lambda_2^2 b_2}$ and $x_2(t) = \ln ( b_1+b_2)$ as $\theta \to 0$ or $\theta \to \infty$, respectively.)
\end{Example}

\begin{Remark} \label{rem:Matsuno} One of the referees suggested the following argument, which shows how one may actually succeed in integrating the ODE for the characteristics directly, for any~$N$, thus giving an independent proof of Corollary~\ref{cor:CH-add-a-ghostpeakon}.

The calculations that follow make use of some identities from an article by Matsuno~\cite{matsuno:2007:CH-multisolitons-peakon-limit}, where it is proved that, in the limit as $\kappa \to 0$, the smooth $N$-soliton solutions $u(x,t)$ of the Camassa--Holm equation~\eqref{eq:CH-kappa} with $\kappa > 0$ (which are known in a parametric form where both the dependent variable~$u(t,\eta)$ and the independent variable~$x(t,\eta)$ are expressed in terms of an additional parameter~$\eta$) reduce to the $N$-peakon solutions of the Camassa--Holm equation~\eqref{eq:CH} with $\kappa = 0$, given by the Beals--Sattinger--Szmigielski formulas.

The identities in question are proved using the Desnanot--Jacobi determinant identity, also known as the Lewis Carroll identity, using the fact that the expression~$\Delta_{k}^a$, that we defined by the formula~\eqref{eq:Delta-evaluated}, is originally a $k \times k$ Hankel determinant,
 \begin{gather} \label{eq:Delta-determinant}
 \Delta_k^a = \det ( A_{a+i+j})_{i,j=0}^{k-1} = \begin{vmatrix}
 A_{a} & A_{a+1} & \dots & A_{a+k-1} \\
 A_{a+1} & A_{a+2} & \dots & A_{a+k} \\
 \vdots & & & \\
 A_{a+k-1} & A_{a+k} & \dots & A_{a+2k-2} \\
 \end{vmatrix}
 ,\qquad
 A_m = \sum_{i=1}^N \lambda_i^m b_i .
 \end{gather}
 Matsuno uses the notation $D_{k}^{(a)}$ for a corresponding Hankel determinant, with a somewhat different normalization for $\lambda_i$ and~$b_i$, which causes a discrepancy by some power of~$2$, but all his identities among these determinants are derived solely from the above expression, and therefore we obtain correct formulas simply by substituting our $\Delta_k^a$ for his~$D_k^{(a)}$, except that $\frac{{\rm d}}{{\rm d}t} D_k^{(a)}$ should be replaced by $2 \frac{{\rm d}}{{\rm d}t} \Delta_k^a$. Thus, for example, the Desnanot--Jacobi identity applied to $\Delta_{k+1}^0$ immediately gives
 \begin{gather} \label{eq:Matsuno-4.5}
 \Delta_{k+1}^0 \Delta_{k-1}^2 = \Delta_{k}^0 \Delta_{k}^2 - \bigl( \Delta_k^1 \bigr)^2 ,
 \end{gather}
 which is a special case of equation~(4.5) in~\cite{matsuno:2007:CH-multisolitons-peakon-limit}.

Equation~(4.27) in~\cite{matsuno:2007:CH-multisolitons-peakon-limit} (with adjustments for the differing conventions) gives the expression
 \begin{gather} \label{eq:Matsuno-u-between-peakons}
 u(x,t) = a(t) {\rm e}^x + b(t) {\rm e}^{-x} ,\qquad a = \frac{\Delta_{k-1}^3}{\Delta_{k}^1} ,\qquad b = \frac{\Delta_{k+1}^{-1}}{\Delta_{k}^1} ,
 \end{gather}
for the multipeakon solution in the $k$th interval from the right (see~\eqref{eq:kth-interval-from-right} in Corollary~\ref{cor:CH-add-a-ghostpeakon}). Thus the ODE $\dot \xi = u(\xi)$ for the characteristics in that interval is equivalent to the Riccati equation
 \begin{gather} \label{eq:Riccati}
 \dot y = a y^2 + b
 \end{gather}
for the function $y(t) = {\rm e}^{\xi(t)}$, and since we know the two solutions
 \begin{gather*}
 y_1 = {\rm e}^{x_{N-k}} = \frac{p}{q} ,\qquad y_2 = {\rm e}^{x_{N+1-k}} = \frac{r}{s},
 \end{gather*}
 where
 \begin{gather*}
 p = \Delta_{k+1}^0 ,\qquad q = \Delta_{k}^2 ,\qquad r = \Delta_{k}^0 ,\qquad s = \Delta_{k-1}^2 ,
 \end{gather*}
 experience with the form of the solutions of Riccati equations (cf.\ Remark~\ref{rem:Matsuno-continued}) may lead one to try the ansatz
 \begin{gather*}
 y = \frac{p + \theta r}{q + \theta s}
 \end{gather*}
 with $\theta$ constant. It is not obvious that this is going to work, but in fact it does, as can be verified by substitution into the ODE:
\begin{gather*}
 \frac{(\dot p + \theta \dot r) (q + \theta s) - (p + \theta r) (\dot q + \theta \dot s)}{(q + \theta s)^2} = \dot y = a y^2 + b = \frac{a (p + \theta r)^2 + b (q + \theta s)^2}{(q + \theta s)^2} .
 \end{gather*}
Here the coefficients of $\theta^0$ and $\theta^2$ in the numerators agree since $y_1$ and $y_2$ are solutions, and the coefficients of $\theta^1$ agree provided that
 \begin{gather*}
 \dot p s + \dot r q - p \dot s - r \dot q = 2 (a p r + b q s) ,
 \end{gather*}
which is true. Indeed, equation~\eqref{eq:Matsuno-4.5} says that $qr-ps = c^2$ where $c = \Delta_k^1$, adding Matsuno's (4.25c) and (4.25d) gives $\dot p s - p \dot s = \dot r q - r \dot q =: Q$, and his (4.25a) and (4.25b) read $a c^2 = \dot q s - q \dot s$ and $b c^2 = \dot p r - p \dot r$, so the right-hand side times $c^2$ equals $2 (a p r + b q s) c^2 = 2 a c^2 \cdot p r + 2 b c^2 \cdot q s = 2 (\dot q s - q \dot s) p r + 2 (\dot p r - p \dot r) q s = 2 (\dot p s - p \dot s) q r - 2 (\dot r q - r \dot q) p s = 2 Q q r - 2 Q p s = (Q + Q) c^2$, which is the left-hand side times~$c^2$.
\end{Remark}

\begin{Remark} \label{rem:Matsuno-continued} We would like to add yet another way of solving the Riccati equation~\eqref{eq:Riccati} in Remark~\ref{rem:Matsuno}. We will need the expression for $a$ from~\eqref{eq:Matsuno-u-between-peakons}, but not the one for~$b$. Following the standard method for Riccati equations, we make the substitution $y = -\dot w / (wa)$ to obtain the linear second-order ODE $\ddot w - (\dot a / a) \dot w + ab w = 0$. We claim that under this substitution, the known solutions $y_1 = p/q$ and $y_2 = r/s$ correspond to $w_1 = q/c$ and $w_2 = s/c$, respectively. In other words, the claim is that $\dot c q - c \dot q = p a c$ and $\dot c s - c \dot s = r a c$, or, written out,
 \begin{gather} \label{eq:determinant-identities}
 \dot \Delta_k^1 \Delta_k^2 - \Delta_k^1 \dot \Delta_k^2 = \Delta_{k+1}^0 \Delta_{k-1}^3 ,\qquad \dot \Delta_k^1 \Delta_{k-1}^2 - \Delta_k^1 \dot \Delta_{k-1}^2 = \Delta_{k}^0 \Delta_{k-1}^3 .
 \end{gather}
The first of the formulas~\eqref{eq:determinant-identities} follows directly from applying the Desnanot--Jacobi identity to the determinant $\Delta_{k+1}^0$ with its first column moved to the far right. (Note that since $\dot A_a = A_{a-1}$ due to $\dot b_k = b_k / \lambda_k$, the derivative $\dot \Delta_k^a$ is given by the same determinant as~$\Delta_k^a$ except that all indices in the first column are reduced by one.) The second one is more difficult, since the matrix sizes do not match the Desnanot--Jacobi identity directly, but it is a special case of equation~(A.25) in Lundmark and Szmigielski~\cite{lundmark-szmigielski:2016:GX-inverse-problem}, with $n=k-1$,
 $z_i=A_i$ for $1 \le i \le k-1$, $z_k=A_0$, $w_i=A_{i+1}$ for $1 \le i \le k-1$, $w_k=A_1$, $X_{ij} = A_{i+j+1}$ for $1 \le i \le k-2$ and $1 \le j \le k-1$, and $X_{k-1,j} = A_{j+1}$ for $1 \le j \le k-1$.

Since the two functions $w_1$ and~$w_2$ satisfy the linear ODE above, so does $w = w_1 + \theta w_2$ (with~$\theta$ constant), and therefore a one-parameter family of solutions to our Riccati equation is
\begin{gather*}
y = - \frac{\dot w}{wa} = - \frac{\dot w_1 + \theta \dot w_2}{(w_1 + \theta w_2) a} = \frac{y_1 w_1 + \theta y_2 w_2}{w_1 + \theta w_2} = \frac{\frac{p}{q} \frac{q}{c} + \theta \frac{r}{s} \frac{s}{c}}{\frac{q}{c} + \theta \frac{s}{c}} = \frac{p + \theta r}{q + \theta s} ,
 \end{gather*}
 as desired.
\end{Remark}

We now turn to examples illustrating in some special cases what the multipeakon solutions and their characteristic curves look like. For the three-dimensional plots of $u(x,t)$, we have made use of the technique
described in Remark~\ref{rem:plotting}.

\begin{figure}[t] \centering
 \begin{tikzpicture}
 \node[anchor=south west,inner sep=0] at (0,0)
 {\includegraphics[width=127mm]{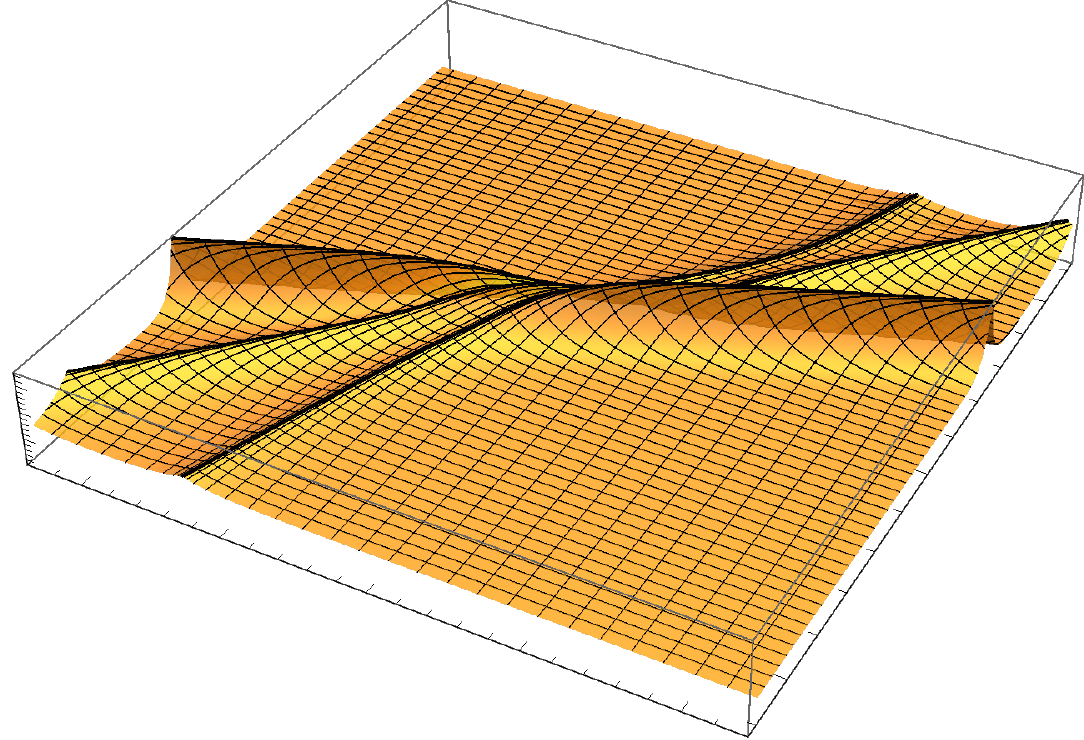}};

 \draw (4.10,1.48) node {$x$};
 \draw (2.34,2.19) node {\scriptsize $-5$};
 \draw (5.98,0.86) node {\scriptsize $5$};

 \draw (11.11,3.06) node {$t$};
 \draw (9.5,0.6) node {\scriptsize $-10$};
 \draw (12.5,5.1) node {\scriptsize $10$};

 \draw (0,3.55) node {$u$};
 \end{tikzpicture}

 \caption{Graph of a pure three-peakon solution $u(x,t) = \sum\limits_{k=1}^3 m_k(t) {\rm e}^{-\abs{x-x_k(t)}}$ of the Camassa--Holm equation, computed from exact formulas as described in Example~\ref{ex:CH-3p-ghost}. The parameter values are given by~\eqref{eq:CH-3p-ghost-parameters}. The graph is plotted as described in Remark~\ref{rem:plotting}, over a mesh consisting of lines $t={\rm const}$ together with the characteristic curves~\eqref{eq:CH-3-ghost} obtained from Corollary~\ref{cor:CH-add-a-ghostpeakon}. The dimensions of the box are $\abs{x} \le 12$, $\abs{t} \le 12$ and $-1 \le u \le 5/2$.} \label{fig:CH-3p-wave}
\end{figure}
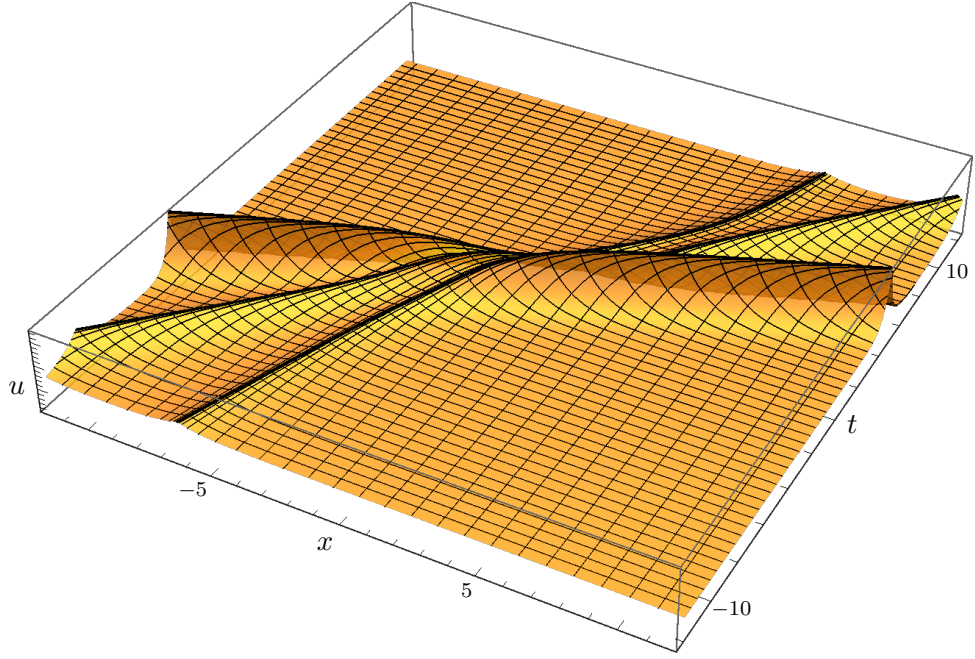

\begin{figure}[t] \centering

 \includegraphics{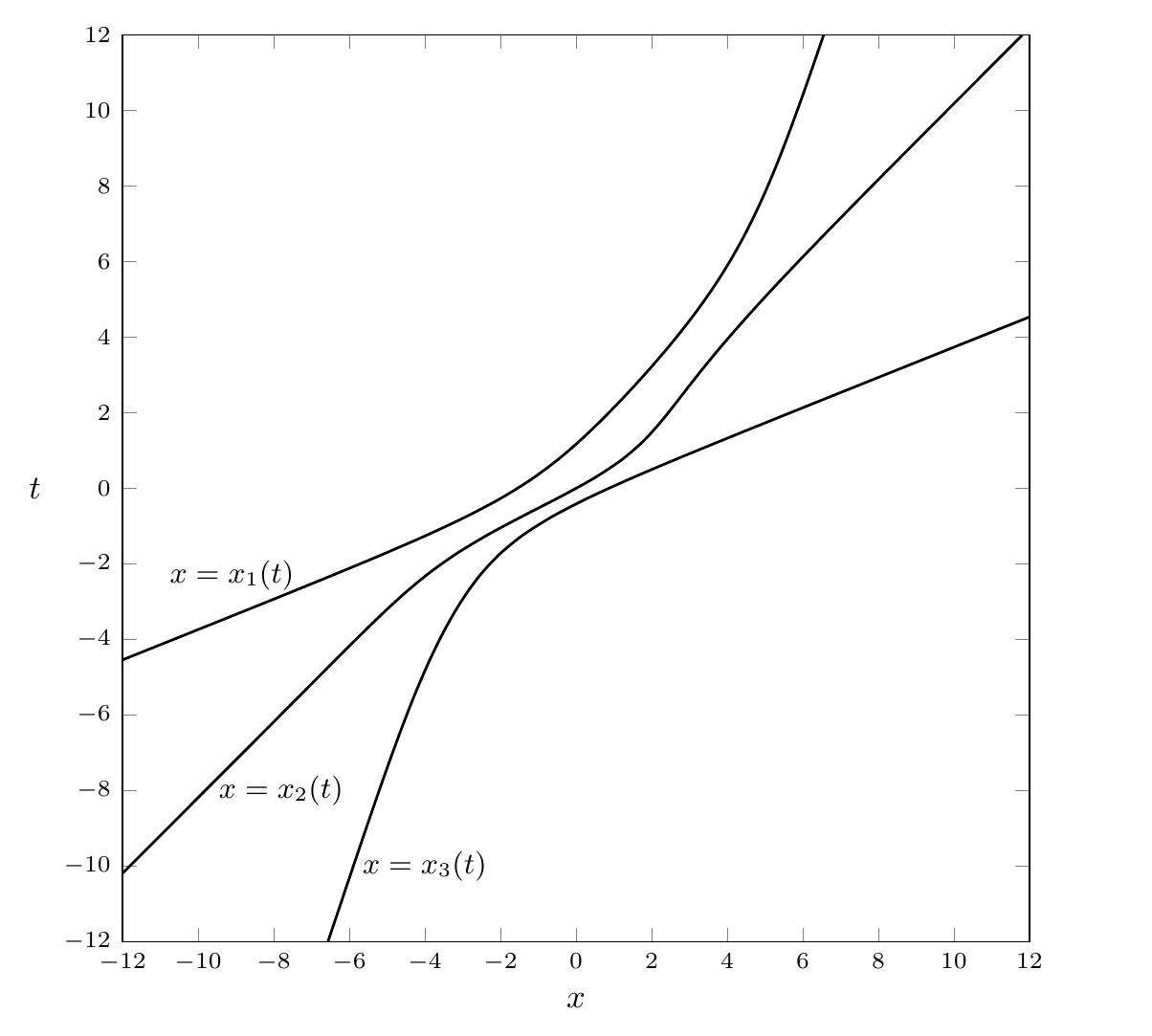}

\caption{Spacetime plot of the peakon trajectories $x=x_k(t)$ for the Camassa--Holm pure three-peakon solution shown in Fig.~\ref{fig:CH-3p-wave}. Since this is a pure peakon solution, $x_1(t) < x_2(t) < x_3(t)$ holds for all $t \in \R$.} \label{fig:CH-3p}
\end{figure}

\begin{figure}[t] \centering

\includegraphics{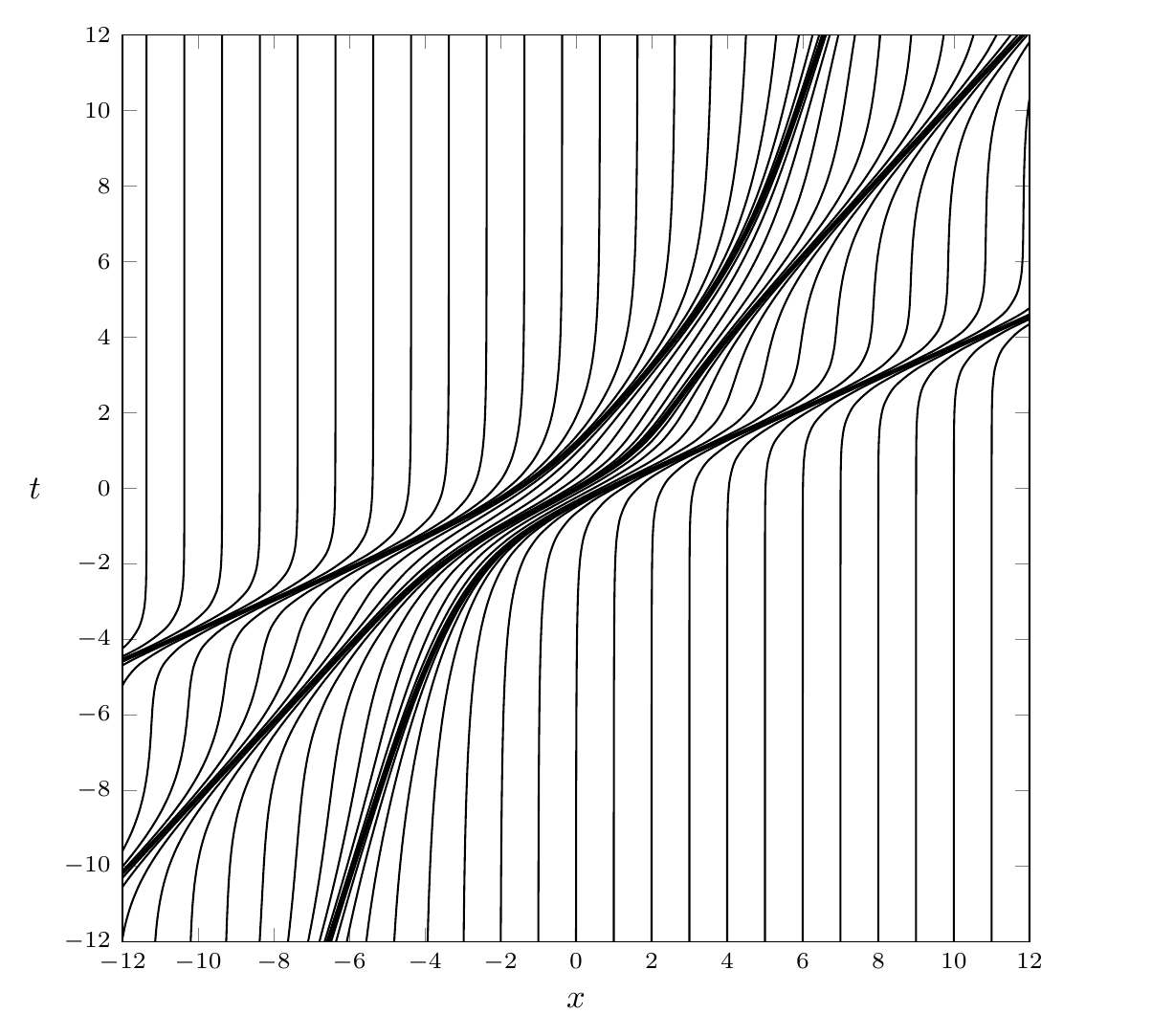}

\caption{A selection of characteristic curves $x = \xi(t)$ for the Camassa--Holm pure three-peakon solution shown in Figs.~\ref{fig:CH-3p-wave} and~\ref{fig:CH-3p}. These curves are given by the formulas~\eqref{eq:CH-3-ghost} in Example~\ref{ex:CH-3p-ghost}.} \label{fig:CH-3p-ghost}
\end{figure}

\begin{Example}[CH three-peakon characteristics]\label{ex:CH-3p-ghost} Fig.~\ref{fig:CH-3p-wave} shows the graph of a pure three-peakon solution of the Camassa--Holm equation,
 \begin{gather*}
 u(x,t) = \underbrace{m_1(t)}_{>0} {\rm e}^{-\abs{x-x_1(t)}} + \underbrace{m_2(t)}_{>0} {\rm e}^{-\abs{x-x_2(t)}} + \underbrace{m_3(t)}_{>0} {\rm e}^{-\abs{x-x_3(t)}} ,
 \end{gather*}
where the positions~$x_k(t)$ and amplitudes~$m_k(t)$ are given by the formulas~\eqref{eq:CH-3-explicit} in Example~\ref{ex:CH-threepeakon-solution}, with the parameter values
\begin{gather} \label{eq:CH-3p-ghost-parameters}
 \lambda_1 = \frac{2}{5} ,\qquad \lambda_2 = 1 ,\qquad \lambda_3 = 3 ,\qquad b_1(0) = \frac{25}{13} ,\qquad b_2(0) = \frac{1}{e} ,\qquad b_3(0) = \frac{1}{13} ,
 \end{gather}
 obtained by taking $c_1 = 5/2$, $c_2=1$, $c_3=1/3$ and $K=-1$ in~\eqref{eq:CH-threepeakon-symmetric}.

Fig.~\ref{fig:CH-3p} shows a plot of the peakon trajectories $x = x_1(t)$, $x = x_2(t)$ and $x = x_3(t)$ in the $(x,t)$ plane. This picture is what the ``mountain ridges'' in Fig.~\ref{fig:CH-3p-wave} would look like if viewed straight from above. The solution formulas express each~$x_k(t)$ as the logarithm of a rational function in a number of exponentials ${\rm e}^{\alpha t}$ with different growth rates~$\alpha$. As $t \to \pm\infty$, when a~single exponential term dominates in each $\Delta_k^a$, the peakons asymptotically travel in straight lines with constant velocities $c_1 = 1/\lambda_1 = 5/2$, $c_2 = 1/\lambda_2 = 1$ and $c_3 = 1/\lambda_3 = 1/3$, and those numbers are also the limiting values of the amplitudes~$m_k(t)$ as $t \to \pm\infty$.

Fig.~\ref{fig:CH-3p-ghost} shows a selection of characteristic curves given by the ghostpeakon formulas from our Theorem~\ref{thm:CH-ghost} or Corollary~\ref{cor:CH-add-a-ghostpeakon}; these curves were also used as mesh lines in Fig.~\ref{fig:CH-3p-wave}. The characteristics $x = \xi(t)$ in the leftmost region $x < x_1(t)$ are given by the formula
 \begin{subequations} \label{eq:CH-3-ghost}
 \begin{gather} \label{eq:CH-3-ghost0}
 \xi(t) = \ln\frac{\overbrace{\Delta_{4}^0}^{=0} + \theta \Delta_{3}^0}{\Delta_{3}^2 + \theta \Delta_{2}^2} = \ln\frac{\theta \Delta_{3}^0}{\Delta_{3}^2 + \theta \Delta_{2}^2} ,
 \end{gather}
 those in the region $x_1(t) < x < x_2(t)$ are
 \begin{gather} \label{eq:CH-3-ghost1}
 \xi(t) = \ln\frac{\Delta_{3}^0 + \theta \Delta_{2}^0}{\Delta_{2}^2 + \theta \Delta_{1}^2} ,
 \end{gather}
 for $x_2(t) < x < x_3(t)$ we have
 \begin{gather} \label{eq:CH-3-ghost2}
 \xi(t) = \ln\frac{\Delta_{2}^0 + \theta \Delta_{1}^0}{\Delta_{1}^2 + \theta \Delta_{0}^2} ,
 \end{gather}
 and in the rightmost region $x_3(t) < x$ the formula is
 \begin{gather} \label{eq:CH-3-ghost3}
 \xi(t) = \ln\frac{\Delta_{1}^0 + \theta \Delta_{0}^0}{\Delta_{0}^2 + \theta \underbrace{\Delta_{-1}^2}_{=0}} = \ln\big(\Delta_1^0 + \theta\big) ,
 \end{gather}
 \end{subequations}
 where in each case $\theta$ is a positive parameter. As $\theta$ runs through the values $0 < \theta < \infty$, the corresponding characteristic curves sweep out their respective regions, and in the limit as $\theta \to 0^+$ or $\theta \to \infty$, the formula for $\xi(t)$ reduces to the appropriate $x_k(t)$ (or to $-\infty$ or~$+\infty$, in the exterior regions $x<x_1$ and $x>x_3$).

 In Fig.~\ref{fig:CH-3p-ghost} (and Fig.~\ref{fig:CH-3p-wave}) the values of~$\theta$ were taken in geometric progressions with ratio~$e$, in order for the curves to be approximately evenly spaced.
\end{Example}

\begin{figure}[t] \centering

 \begin{tikzpicture}
 \node[anchor=south west,inner sep=0] at (0,0)
 {\includegraphics[width=127mm]{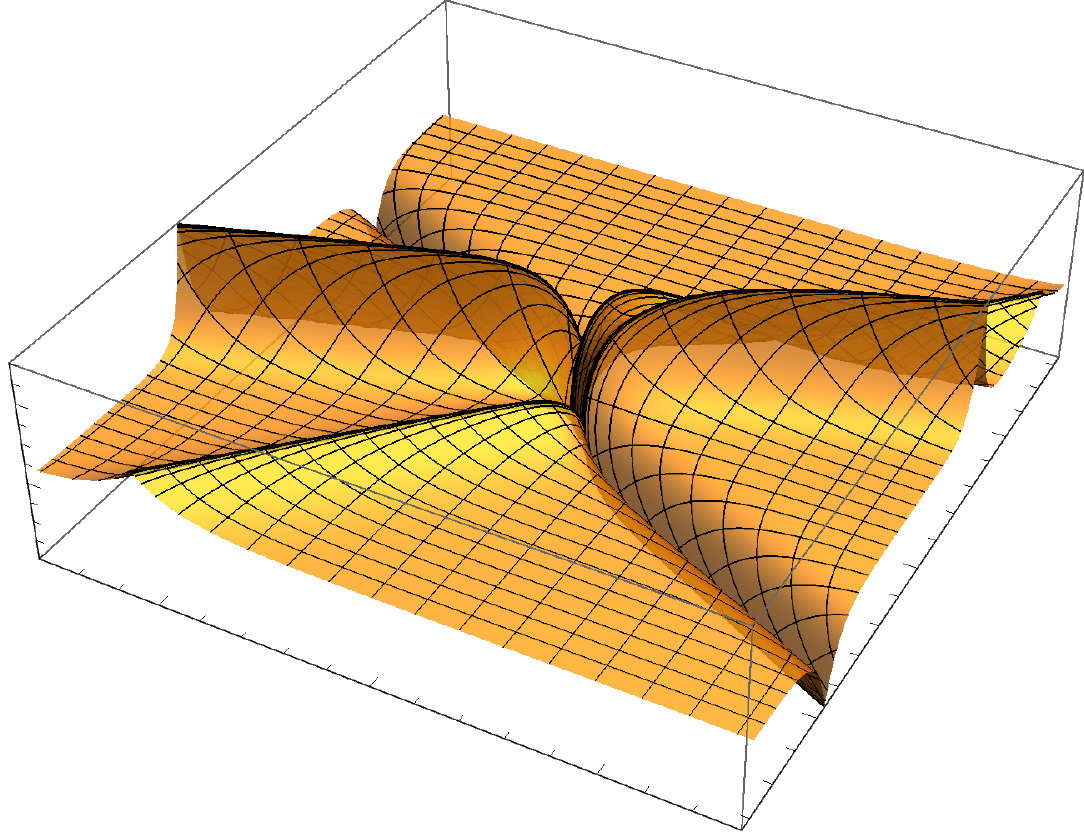}};

 \draw (4.15,1.5) node {$x$};
 \draw (1.67,2.47) node {\scriptsize $-5$};
 \draw (6.9,0.5) node {\scriptsize $5$};

 \draw (11,3) node {$t$};
 \draw (9.75,1.2) node {\scriptsize $-5$};
 \draw (12.0,4.6) node {\scriptsize $5$};

 \draw (0,4.1) node {$u$};
 \end{tikzpicture}

 \caption{Graph of a \emph{conservative} Camassa--Holm solution $u(x,t)$ with two peakons (positive amplitude)
 and one antipeakon (negative amplitude),
 computed from exact formulas
 as described in Example~\ref{ex:CH-2p-1ap-ghost}.
 The dimensions of the box are $\abs{x} \le 8$, $\abs{t} \le 8$ and $-2 \le u \le 3$.} \label{fig:CH-2p-1ap-wave}
\end{figure}

\begin{figure}[t] \centering

 \includegraphics{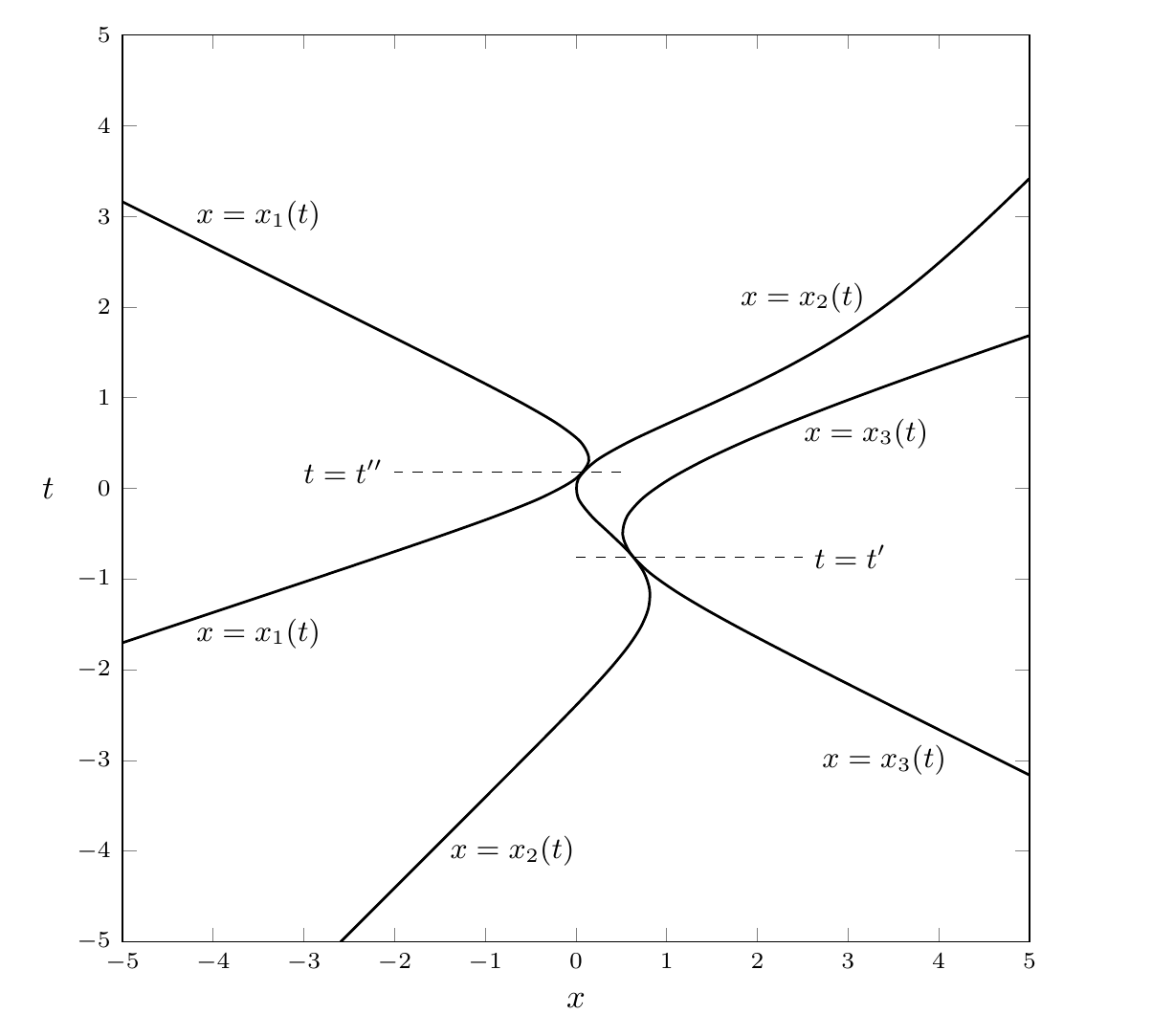}

 \caption{Spacetime plot of the peakon trajectories $x=x_k(t)$ for the $N=3$ conservative peakon--antipeakon Camassa--Holm solution shown in Fig.~\ref{fig:CH-2p-1ap-wave}. The dashed lines indicate the times $t=t'$ and $t=t''$ of the two peakon--antipeakon collisions that occur; for all other~$t$ the strict ordering $x_1(t) < x_2(t) < x_3(t)$ holds. The solution starts out with $m_1$ and~$m_2$ positive and $m_3$~negative, so that the antipeakon is on the far right for $t<t'$. At the first collision, where $x_2(t')=x_3(t')$, the corresponding amplitude factors~$m_2$ and~$m_3$ blow up to $\pm\infty$ and change their signs so that it is the middle peakon that plays the role of the antipeakon for $t' < t < t''$ ($m_1$~positive, $m_2$~negative, $m_3$~positive). Similarly, $m_1$ and~$m_2$ switch signs at the second collision where $x_1(t'')=x_2(t'')$, and the solution ends up with the antipeakon on the far left for $t>t''$ ($m_1$~negative, $m_2$ and~$m_3$ positive).} \label{fig:CH-2p-1ap}
\end{figure}

\begin{figure}[t] \centering

\includegraphics{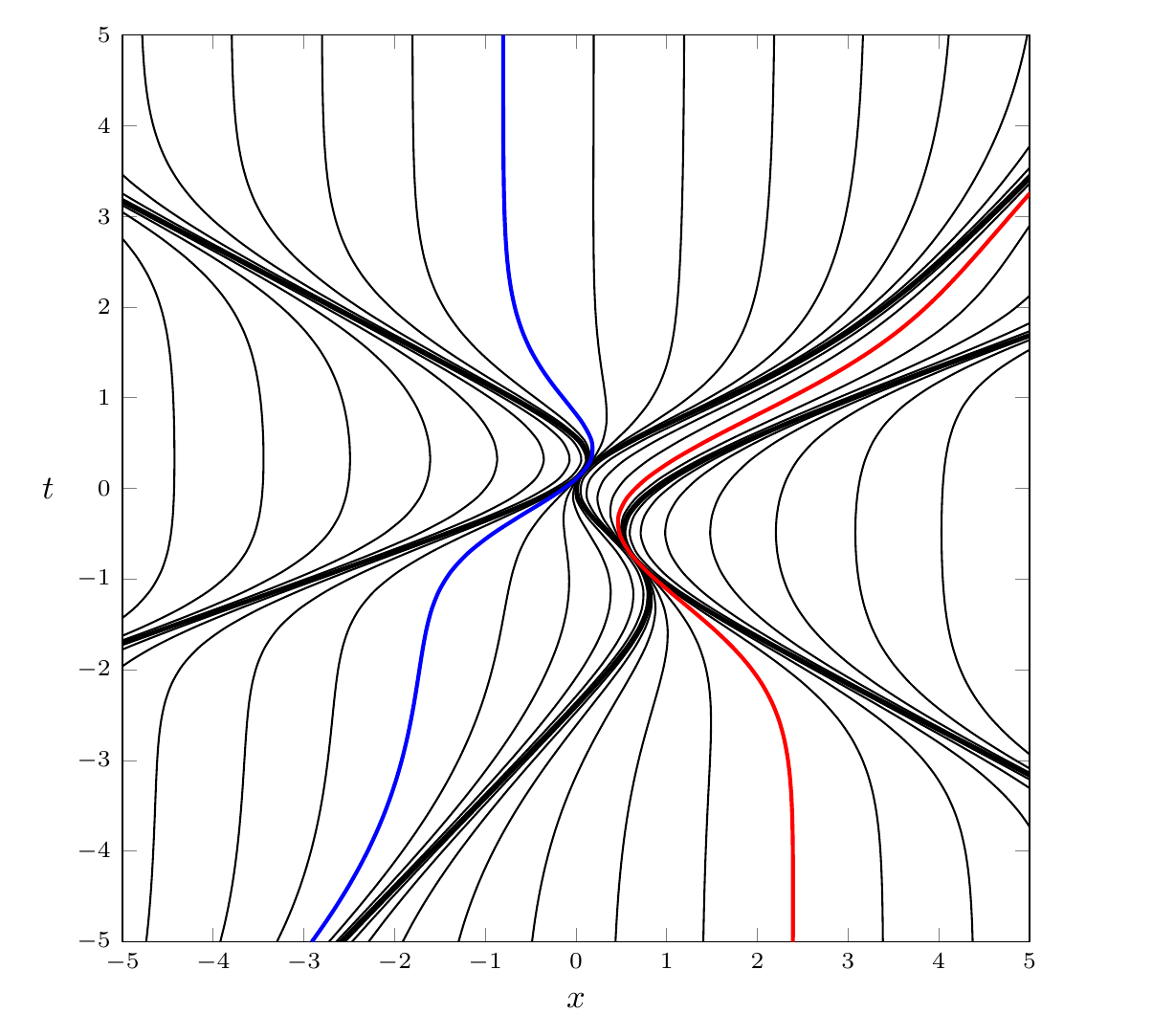}

\caption{A selection of characteristic curves for the $N=3$ conservative peakon--antipeakon Camassa--Holm solution shown in Figs.~\ref{fig:CH-2p-1ap-wave} and~\ref{fig:CH-2p-1ap}. All curves are computed from exact formulas as described in Example~\ref{ex:CH-2p-1ap-ghost}. The blue curve is a typical characteristic curve in the region $x_1 < x < x_2$, and the red curve is a typical characteristic curve in the region $x_2 < x < x_3$.}\label{fig:CH-2p-1ap-ghost}
\end{figure}

\begin{Example}[conservative CH peakon solution] \label{ex:CH-2p-1ap-ghost} From the same formulas as in Example~\ref{ex:CH-3p-ghost}, but using the parameter values
\begin{gather} \label{eq:CH-2p-1ap-ghost-parameters}
 \lambda_1 = \frac{1}{3} ,\qquad \lambda_2 = 1 ,\qquad \lambda_3 = - \frac{1}{2} ,\qquad b_1(0) = \frac{9}{10} ,\qquad b_2(0) = \frac{{\rm e}^2}{6} ,\qquad b_3(0) = \frac{4}{15}
\end{gather}
instead of~\eqref{eq:CH-3p-ghost-parameters}, we get a mixed peakon--antipeakon solution instead of a pure peakon solution. The asymptotic velocities (and amplitudes) as $t \to \pm\infty$ are
 \begin{gather*}
 c_1 = \frac{1}{\lambda_1} = 3 ,\qquad c_2 = \frac{1}{\lambda_2} = 1 ,\qquad c_3 = \frac{1}{\lambda_3} = -2 .
 \end{gather*}
Since $c_1 > c_2 > 0 > c_3$, there are two peakons and one antipeakon. (The values for $b_k(0)$ in~\eqref{eq:CH-2p-1ap-ghost-parameters} were obtained by taking these numbers $c_k$ together with $K=2$ in the formulas~\eqref{eq:CH-threepeakon-symmetric}.)

The wave profile $u(x,t)$ is illustrated in Fig.~\ref{fig:CH-2p-1ap-wave}, while the peakon trajectories $x = x_k(t)$ are plotted in Fig.~\ref{fig:CH-2p-1ap}. A selection of characteristic curves $x = \xi(t)$ are plotted in Fig.~\ref{fig:CH-2p-1ap-ghost}; they are also given by the same formulas as in Example~\ref{ex:CH-3p-ghost} but with the new parameter values~\eqref{eq:CH-2p-1ap-ghost-parameters}.

 In contrast to the pure peakon case, the solution formulas for $m_k(t)$ are not globally defined; there is an instant
 \begin{gather*}
 t = t' \approx -0.758404
 \end{gather*}
 (easily determined numerically) when the quantity
 \begin{gather*}
 \Delta_1^1(t) = \lambda_1 b_1(t) + \lambda_2 b_2(t) + \lambda_3 b_3(t) = \frac{3 {\rm e}^{3t}}{10} + \frac{{\rm e}^{2+t}}{6} -\frac{2 {\rm e}^{-2t}}{15}
 \end{gather*}
 becomes zero, which causes $m_2$ and~$m_3$ to blow up, since $\Delta_1^1$ occurs in the denominator of the formulas for $m_2$ and~$m_3$ in~\eqref{eq:CH-m-3-explicit}. Similarly, there is another instant
 \begin{gather*}
 t = t'' \approx 0.182763,
 \end{gather*}
 when the expression
\begin{gather*}
 \Delta_2^1(t) = \lambda_1 \lambda_2 (\lambda_1-\lambda_2)^2 b_1(t) b_2(t) + \lambda_1 \lambda_3 (\lambda_1-\lambda_3)^2 b_1(t) b_3(t) + \lambda_2 \lambda_3 (\lambda_2-\lambda_3)^2 b_2(t) b_3(t) \\
 \hphantom{\Delta_2^1(t)}{} = \frac{{\rm e}^{2+4t}}{45} - \frac{{\rm e}^t}{36} - \frac{{\rm e}^{2-t}}{20}
\end{gather*}
occurring in the denominator of $m_1$ and $m_2$ becomes zero. These times $t'$ and~$t''$, when two amplitudes $m_k$ blow up (one to $+\infty$ and the other to $-\infty$), are also the instants of \emph{collisions} between the corresponding peakon and antipeakon:
 \begin{gather*}
 x_1(t') < x_2(t') = x_3(t') ,\qquad x_1(t'') = x_2(t'') < x_3(t'') .
 \end{gather*}
 For all other $t \in \R$, the strict ordering $x_1(t) < x_2(t) < x_3(t)$ holds.

From the point of the peakon ODEs, the solution $\{ x_k(t), m_k(t) \}_{k=1}^3$ given by the explicit formulas for $t \in \R \setminus \{t',t''\}$ is obtained by analytic continuation in the complex $t$ plane around simple poles at $t'$ and~$t''$.

More importantly, from the point of view of the PDE, there is cancellation which causes the function
 \begin{gather*}
 u(x,t) = \sum_{k=1}^3 m_k(t) {\rm e}^{-\abs{x-x_k(t)}} ,\qquad t \in \R \setminus \{ t', t'' \} ,
 \end{gather*}
 with $\{ x_k(t), m_k(t) \}_{k=1}^3$ given by the explicit formulas, to have finite limits as $t \to t'$ and as $t \to t''$. More specifically, $\Delta_1^1$ cancels in the sum
 \begin{gather*}
 m_2 + m_3 = \frac{\Delta_2^0 \Delta_{1}^2}{\Delta_2^1 \Delta_{1}^1} + \frac{\Delta_1^0 \Delta_{0}^2}{\Delta_1^1} \\
 \hphantom{m_2 + m_3}{}
 = \frac{(\lambda_1-\lambda_2)^2 (\lambda_1+\lambda_2) b_1 b_2 + (\lambda_1-\lambda_3)^2 (\lambda_1+\lambda_3) b_1 b_3 + (\lambda_2-\lambda_3)^2 (\lambda_2+\lambda_3) b_2 b_3}{\Delta_2^1} ,
 \end{gather*}
so that this expression remains finite at $t=t'$, and $\Delta_2^1$ cancels in the sum
 \begin{gather*}
 m_1 + m_2 = \frac{\Delta_3^0 \Delta_{2}^2}{\Delta_3^1 \Delta_{2}^1} + \frac{\Delta_2^0 \Delta_{1}^2}{\Delta_2^1 \Delta_{1}^1} =
 \frac{\lambda_1^2 (\lambda_2+\lambda_3) b_1 + \lambda_2^2 (\lambda_1+\lambda_3) b_2 + \lambda_3^2 (\lambda_1+\lambda_2) b_3 }{\lambda_1 \lambda_2 \lambda_3 \Delta_1^1},
\end{gather*}
 so that this expression is finite at $t=t''$. Extending $u$ continuously by defining
 \begin{gather} \label{eq:u-at-first-collision}
 u(x,t') = \lim_{t \to t'} u(x,t) = m_1(t') {\rm e}^{-\abs{x-x_1(t')}} + (m_2+m_3)(t') {\rm e}^{-\abs{x-x_2(t')}}
 \end{gather}
 and
 \begin{gather*}
 u(x,t'') = \lim_{t \to t''} u(x,t) = (m_1+m_2)(t'') {\rm e}^{-\abs{x-x_1(t'')}} + m_3(t'') {\rm e}^{-\abs{x-x_3(t'')}} ,
 \end{gather*}
 we get a global weak solution $u(x,t)$ of the Camassa--Holm equation, which has been continued past the singularities at $t'$ and~$t''$, where the derivative $u_x$ blows up, although $u$ itself remains bounded. This is the so-called \emph{conservative} multipeakon solution; cf. Remark~\ref{rem:DP-Novikov-references}.
\end{Example}

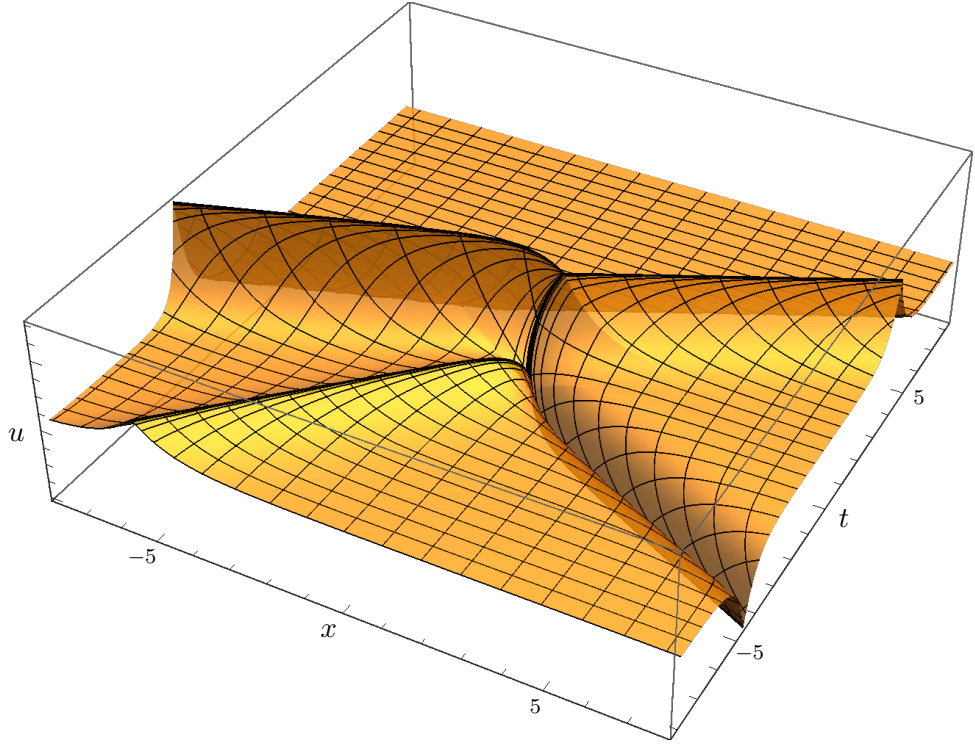
\begin{figure}[t] \centering

 \begin{tikzpicture}
 \node[anchor=south west,inner sep=0] at (0,0)
 {\includegraphics[width=127mm]{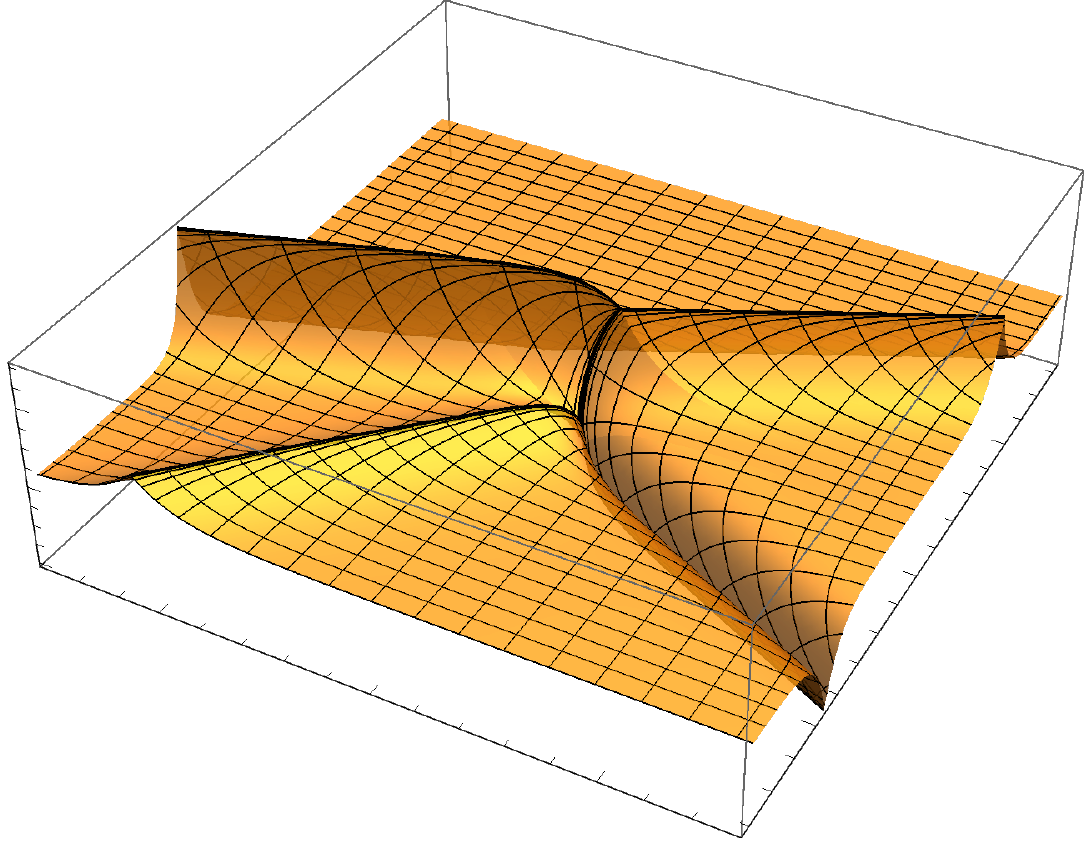}};

 \draw (4.15,1.5) node {$x$};
 \draw (1.67,2.47) node {\scriptsize $-5$};
 \draw (6.9,0.5) node {\scriptsize $5$};

 \draw (11,3) node {$t$};
 \draw (9.75,1.2) node {\scriptsize $-5$};
 \draw (12.0,4.6) node {\scriptsize $5$};

 \draw (0,4.1) node {$u$};
 \end{tikzpicture}

\caption{Graph of a \emph{dissipative} Camassa--Holm solution $u(x,t)$. The solution starts out identical to the conservative solution in Fig.~\ref{fig:CH-2p-1ap-wave}, with two peakons and one antipeakon, but here they merge at collisions so that only one peakon remains in the end. See Example~\ref{ex:CH-2p-1ap-ghost} for details. The dimensions of the box are $\abs{x} \le 8$, $\abs{t} \le 8$ and $-2 \le u \le 3$.} \label{fig:CH-dissipative-2p-1ap-wave}
\end{figure}

\begin{figure}[t] \centering

 \includegraphics{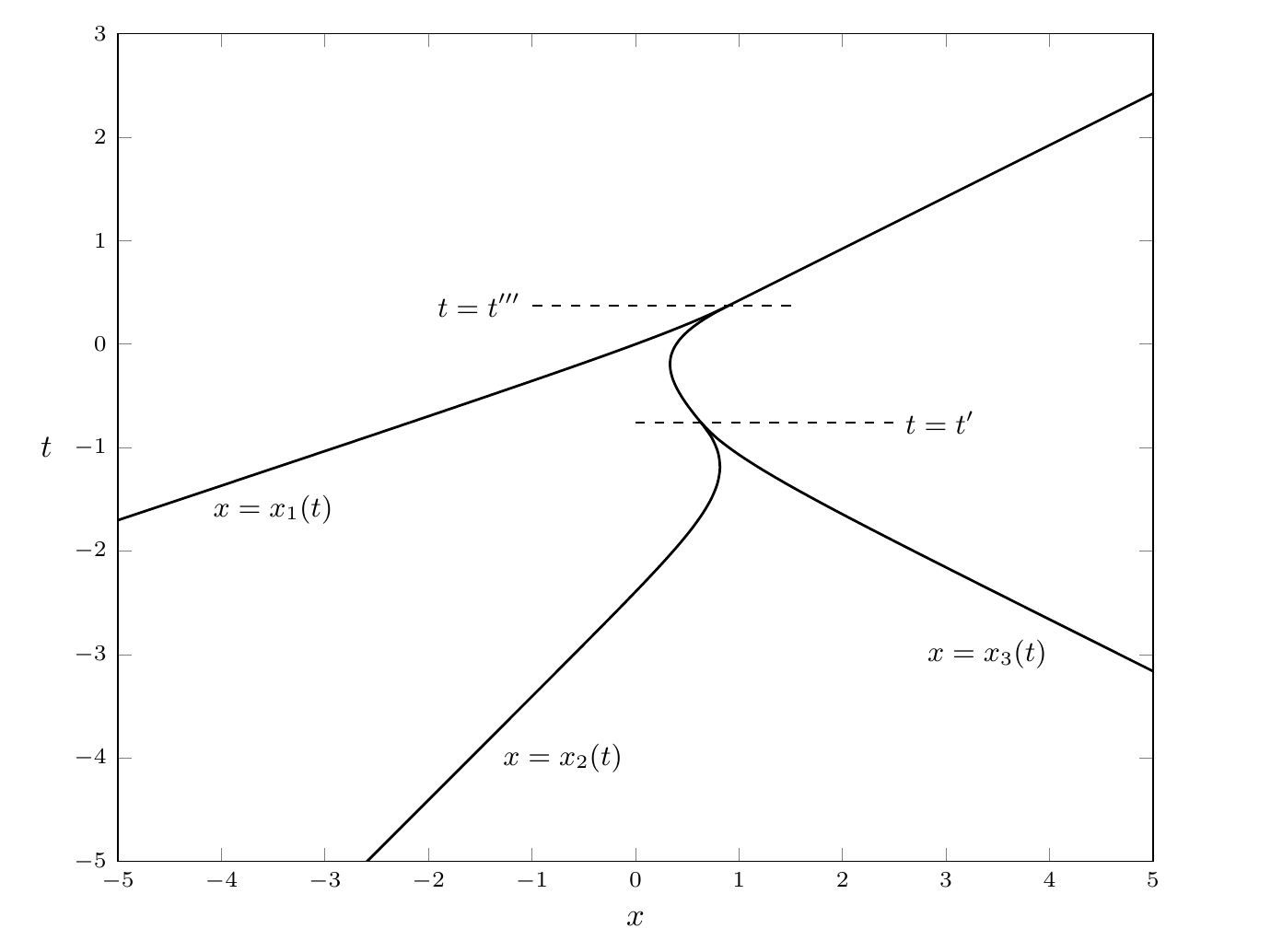}

\caption{Spacetime plot of the peakon trajectories $x=x_k(t)$ for the dissipative Camassa--Holm solution shown in Fig.~\ref{fig:CH-dissipative-2p-1ap-wave}. The solution starts out with two peakons on the left (at $x_1$ and~$x_2$) and one antipeakon on the right (at~$x_3$). At the first collision, when $x_2(t')=x_3(t')$, the colliding peakon and antipeakon merge into a single antipeakon, which then collides with the leftmost peakon at time $t=t'''$, after which the solution continues as a single peakon. For $t \le t'$, the solution is given by the three-peakon solution formulas, and for $t' < t \le t'''$ it is given by the two-peakon solution formulas, with new spectral data computed as described in Example~\ref{ex:CH-dissipative-2p-1ap}. And for $t''' < t$, the solution is just a~one-peakon solution~-- a~travelling wave where the wave profile~$u(x,t''')$ is translated with constant speed.}\label{fig:CH-dissipative-2p-1ap}
\end{figure}

\begin{figure}[t] \centering

\includegraphics{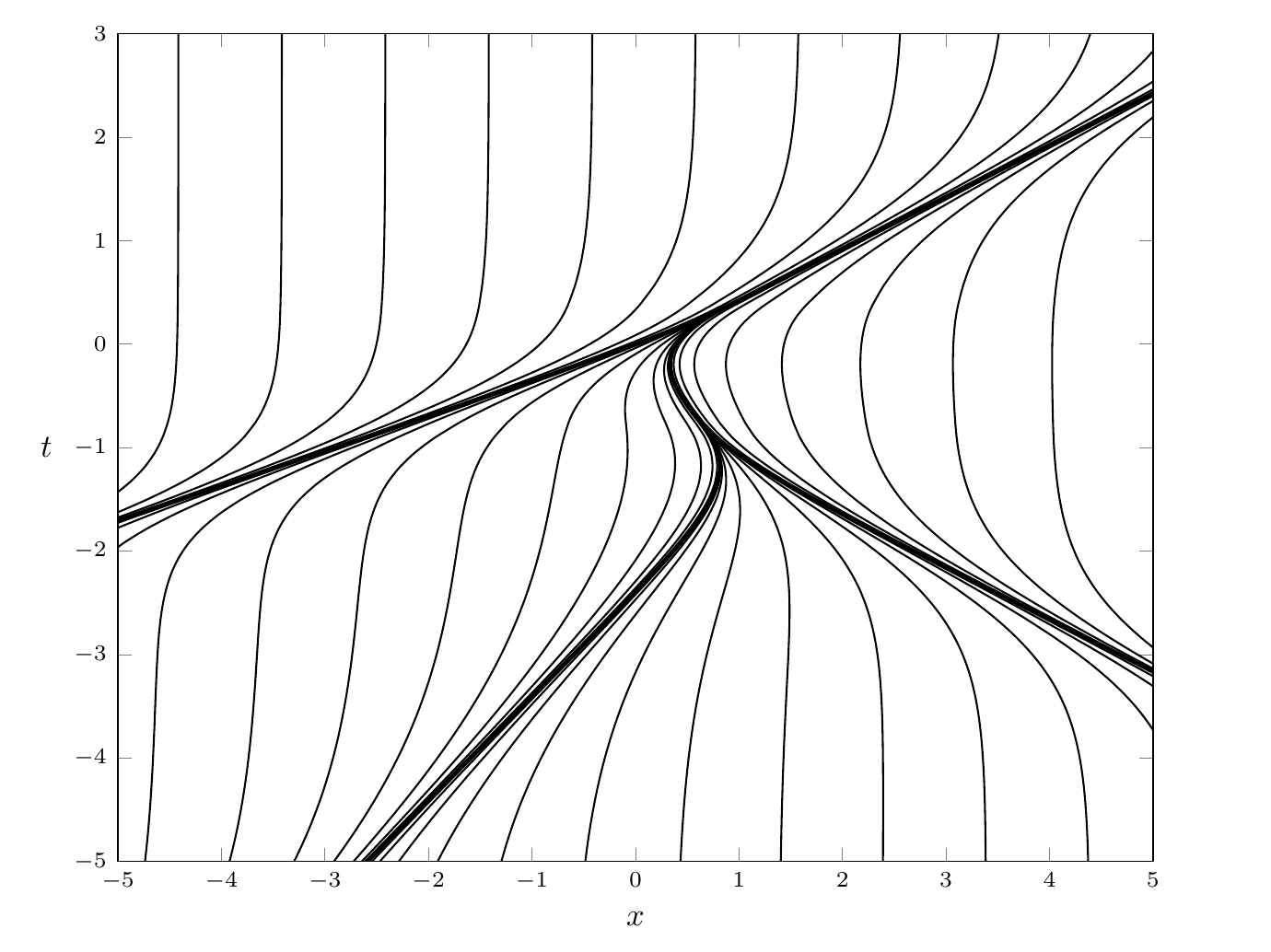}

\caption{A selection of characteristic curves (plotted from the exact formulas) for the $N=3$ dissipative peakon--antipeakon Camassa--Holm solution shown in Figs.~\ref{fig:CH-dissipative-2p-1ap-wave} and~\ref{fig:CH-dissipative-2p-1ap}. See Example~\ref{ex:CH-dissipative-2p-1ap}.} \label{fig:CH-dissipative-2p-1ap-ghost}
\end{figure}

\begin{Example}[dissipative CH peakon solution] \label{ex:CH-dissipative-2p-1ap} The \emph{dissipative} multipeakon solution is another global weak solution, where $u(x,t)$ is continued past singularities in a different way, as illustrated in Fig.~\ref{fig:CH-dissipative-2p-1ap-wave}; cf.\ Remark~\ref{rem:DP-Novikov-references}. In the scenario of the previous example, the conservative and the dissipative solutions agree for $t \le t'$, up until the time of the first collision. The wave profile $u(x,t')$ defined by \eqref{eq:u-at-first-collision} has the two-peakon form
\begin{gather*}
 u(x,t') = \tilde m_1(t') {\rm e}^{-\abs{x-\tilde x_1(t')}} + \tilde m_2(t') {\rm e}^{-\abs{x-\tilde x_2(t')}}
 \end{gather*}
 with
 \begin{alignat}{3}
& \tilde x_1(t') = x_1(t') \approx -2.1819596 ,\qquad && \tilde m_1(t') = m_1(t') \approx 3.0170075 ,& \nonumber\\
& \tilde x_2(t')= x_2(t') = x_3(t') \approx 0.6337836 ,\qquad  &&  \tilde m_2(t') = (m_2+m_3)(t') \approx -1.0170075 ,\!\!\!\!&\label{eq:pos-and-amp-at-first-collision}
 \end{alignat}
where the values are obtained from exact formulas, except that $t'$ has to be found by solving the equation $\Delta_1^1(t)=0$ numerically. In contrast to the conservative solution above, which again contains three peakons immediately after the collision, the dissipative solution continues as a~two-peakon solution, given by the usual explicit formulas~\eqref{eq:CH-twopeakon-solution}, but with new parameters~$\tilde\lambda_1$, $\tilde\lambda_2$, $\tilde b_1(t')$ and $\tilde b_2(t')$ determined from the values~\eqref{eq:pos-and-amp-at-first-collision} via the general relations
 \begin{gather*}
 \frac{1}{\lambda_1} + \frac{1}{\lambda_2} = m_1 + m_2 ,\qquad \frac{1}{\lambda_1 \lambda_2} = m_1 m_2 (1-{\rm e}^{x_1-x_2})
 \end{gather*}
 and
 \begin{gather*}
 b_1 + b_2 = {\rm e}^{x_2} ,\qquad \frac{b_1}{\lambda_1} + \frac{b_2}{\lambda_2} = m_1 {\rm e}^{x_1} + m_2 {\rm e}^{x_2}.
 \end{gather*}
(These formulas originate in the study of the forward spectral problem; we will not go into the details here. It is by solving these equations for $x_1$, $x_2$, $m_1$, $m_2$ that the Camassa--Holm two-peakon solution formulas~\eqref{eq:CH-twopeakon-solution} are obtained.) This gives
 \begin{alignat*}{3}
 & \tilde\lambda_1\approx 0.3365925 ,\qquad && \tilde b_1(t') = \tilde b_1(0) {\rm e}^{t'/\tilde \lambda_1} \approx 0.06432834 ,& \\
 & \tilde\lambda_2\approx -1.02991776 ,\qquad && \tilde b_2(t') = \tilde b_2(0) {\rm e}^{t'/\tilde \lambda_2} \approx 1.82039986 .&
 \end{alignat*}
 The remaining peakon and the antipeakon then collide when
 \begin{gather*}
 \tilde \Delta_1^1(t) = \tilde \lambda_1 \tilde b_1(0) {\rm e}^{t/\tilde \lambda_1} + \tilde \lambda_2 \tilde b_2(0) {\rm e}^{t/\tilde \lambda_2}
 \end{gather*}
 vanishes, namely at
 \begin{gather*}
 t = t''' \approx 0.373326 ,\qquad x_1(t''') = x_2(t''') \approx 0.9013432 ,
 \end{gather*}
 and merge into a single peakon of amplitude~$2$:
 \begin{gather*}
 u(x,t) = 2 {\rm e}^{-\abs{x-x_1(t''')-2t}} \qquad \text{for $t \ge t'''$} .
 \end{gather*}
(The amplitude here is exactly~$2$, which is due to the momentum $\int_{\R} u {\rm d}x = \sum m_i$ being conserved even for dissipative solutions. Initially, the momentum is $m_1+m_2+m_3=\tfrac{1}{\lambda_1}+\tfrac{1}{\lambda_2}+\tfrac{1}{\lambda_3} = 3+1+(-2)=2$, and so it stays equal to~$2$.) The peakon trajectories for the dissipative solution are plotted in Fig.~\ref{fig:CH-dissipative-2p-1ap}, while Fig.~\ref{fig:CH-dissipative-2p-1ap-ghost} shows a selection of characteristic curves.
\end{Example}

\section{Degasperis--Procesi ghostpeakons}\label{sec:DP-ghost}

We now leave the Camassa--Holm equation, and move on to the Degasperis--Procesi peakon ODEs~\eqref{eq:DP-peakon-ODEs-shorthand}. In addition to deriving ghostpeakons formulas analogous to those for the Camassa--Holm equation
in Section~\ref{sec:CH-ghost}, we also use direct integration to compute the characteristics for the one-shockpeakon solution formed at a peakon--antipeakon collision; see Example~\ref{ex:DP-p-ap-symm} for the symmetric case
and Example~\ref{ex:DP-p-ap-asymm} for the asymmetric case.

For simplicity, we will formulate Theorem~\ref{thm:DP-ghost} and Corollary~\ref{cor:DP-add-a-ghostpeakon} for pure peakon solutions, but to some extent they remain valid also for mixed peakon--antipeakon solutions;
see Remark~\ref{rem:DP-ghost-antipeakons}.

\begin{Theorem} \label{thm:DP-ghost} Fix some $p$ with $0 \le p \le N$. The solution of the Degasperis--Procesi $(N+1)$-peakon ODEs \eqref{eq:DP-peakon-ODEs-shorthand} with $x_1 < \dotsb < x_{N+1}$ and all amplitudes $m_k(t)$ positive except for $m_{N+1-p}(t) = 0$ is as follows: the position of the ghostpeakon is given by
 \begin{gather} \label{eq:DP-ghost-general}
 x_{N+1-p}(t) = \ln \frac{U_{p+1}^0 + \theta U_{p}^0}{U_{p}^1 + \theta U_{p-1}^1} ,\qquad 0 < \theta < \infty ,
 \end{gather}
 while the other peakons are given by the general solution formulas~\eqref{eq:DP-generalsolution}, up to renumbering:
 \begin{gather}
 x_{N+1-k}(t) = \begin{cases}
 \ln \dfrac{U_{k+1}^0}{U_{k}^1}, & 0 \le k < p, \vspace{1mm}\\
 \ln \dfrac{U_{k}^0}{U_{k-1}^1}, & p < k \le N,
 \end{cases} \qquad
 m_{N+1-k}(t) =
 \begin{cases}
 \dfrac{(U_{k+1}^0 U_{k}^1)^2}{W_{k+1} W_{k}}, & 1 \le k < p, \vspace{1mm}\\
 \dfrac{(U_{k}^0 U_{k-1}^1)^2}{W_{k} W_{k-1}}, & p < k \le N.
 \end{cases}\!\!\!\!\label{eq:DP-N-peakons-renumbered}
 \end{gather}
 Here $\theta \in (0,\infty)$ is a constant in one-to-one correspondence with the ghostpeakon's initial position $x_{N+1-p}(0)$, while the quantities $\{ \lambda_k, b_k \}_{k=1}^{N-1}$ appearing in the expressions
 $U_{k}^a$ and~$W_{k}$ have the usual time dependence $\lambda_k = {\rm const}$, $b_k(t) = b_k(0) {\rm e}^{t/\lambda_k}$.
\end{Theorem}

\begin{proof} The (pure) $(N+1)$-peakon solution of the Degasperis--Procesi equation is given by \eqref{eq:DP-generalsolution} with $N+1$ instead of~$N$:
 \begin{gather*}
 x_{N+1-k} = \ln \frac{\tilde U_{k+1}^0}{\tilde U_{k}^1} ,\qquad m_{N+1-k} = \frac{\big(\tilde U_{k+1}^0 \tilde U_{k}^1\big)^2}{\tilde W_{k+1} \tilde W_{k}} ,\qquad 0 \le k \le N,
 \end{gather*}
 where $\tilde U_k^a = U_k^a(N+1)$ and $\tilde W_k = W_k(N+1)$. Replace the parameters $\lambda_{N+1}$ and~$b_{N+1}(0)$ with
 \begin{gather} \label{eq:DP-reparametrization}
 \varepsilon = \frac{1}{\lambda_{N+1}} ,\qquad \theta = \lambda_{N+1}^{p} b_{N+1}(0) ,
 \end{gather}
 and let
 \begin{gather*}
 \Theta = \Theta(t) = \theta {\rm e}^{\varepsilon t} = \lambda_{N+1}^{p} b_{N+1}(t) .
 \end{gather*}
 Then
 \begin{gather*}
 \tilde U_{k}^a = \sum_{1 \le i_1 < \dotsb < i_k \le N+1} \left( \prod_{r=1}^k \lambda_{i_r}^a b_{i_r} \right) \frac{\Delta( \lambda_{i_1}, \dots, \lambda_{i_k} )^2}{\Gamma( \lambda_{i_1}, \dots, \lambda_{i_k} )} \\
 \hphantom{\tilde U_{k}^a}{} =
\sum_{1 \le i_1 < \dotsb < i_k \le N} \left( \prod_{r=1}^k \lambda_{i_r}^a b_{i_r} \right) \frac{\Delta( \lambda_{i_1}, \dots, \lambda_{i_k} )^2}{\Gamma( \lambda_{i_1}, \dots, \lambda_{i_k} )}
 \\ \hphantom{\tilde U_{k}^a=}{}+ \displaystyle \sum_{1 \le i_1 < \dotsb < i_{k-1} \le N} \left( \prod_{r=1}^{k-1} \lambda_{i_r}^a b_{i_r} \right) \frac{\Delta( \lambda_{i_1}, \dots, \lambda_{i_{k-1}} )^2}{\Gamma( \lambda_{i_1}, \dots, \lambda_{i_{k-1}} )} \lambda_{N+1}^a b_{N+1} \prod_{s=1}^{k-1} \frac{(\lambda_{i_s} - \lambda_{N+1})^2}{(\lambda_{i_s} + \lambda_{N+1})} \\
 \hphantom{\tilde U_{k}^a}{} = U_k^a + \displaystyle \sum_{1 \le i_1 < \dotsb < i_{k-1} \le N} \left( \prod_{r=1}^{k-1} \lambda_{i_r}^a b_{i_r} \right) \frac{\Delta( \lambda_{i_1}, \dots, \lambda_{i_{k-1}} )^2}{\Gamma( \lambda_{i_1}, \dots, \lambda_{i_{k-1}} )} \varepsilon^{-a} \varepsilon^p \Theta
 \prod_{s=1}^{k-1} \frac{(\varepsilon \lambda_{i_s} - 1)^2}{(\varepsilon \lambda_{i_s} + 1) \varepsilon} \\
 \hphantom{\tilde U_{k}^a}{} = U_k^a +
 \displaystyle \sum_{1 \le i_1 < \dotsb < i_{k-1} \le N} \left( \prod_{r=1}^{k-1} \lambda_{i_r}^a b_{i_r} \right) \frac{\Delta( \lambda_{i_1}, \dots, \lambda_{i_{k-1}} )^2}{\Gamma( \lambda_{i_1}, \dots, \lambda_{i_{k-1}} )} \Theta \frac{\varepsilon^{p-a}}{\varepsilon^{k-1}}
( 1 + \order{\varepsilon} )
 \\ \hphantom{\tilde U_{k}^a}{} = U_{k}^a + U_{k-1}^a \Theta \varepsilon^{p-k-a+1} ( 1 + \order{\varepsilon}) \qquad (\text{as $\varepsilon \to 0$}) ,
 \end{gather*}
 and hence
 \begin{gather*}
 \tilde W_{k} =
 \begin{vmatrix} \tilde U_{k}^0 & \tilde U_{k-1}^1 \\ \tilde U_{k+1}^0 & \tilde U_{k}^1 \end{vmatrix} \\
 \hphantom{\tilde W_{k}}{} =
 \begin{vmatrix} U_{k}^0 + U_{k-1}^0 \Theta \varepsilon^{p-k+1} ( 1 + \order{\varepsilon} ) & U_{k-1}^1 + U_{k-2}^1 \Theta \varepsilon^{p-k+1} ( 1 + \order{\varepsilon} ) \\ U_{k+1}^0 + U_{k}^0 \Theta \varepsilon^{p-k} ( 1 + \order{\varepsilon} ) & U_{k}^1 + U_{k-1}^1 \Theta \varepsilon^{p-k} ( 1 + \order{\varepsilon} ) \end{vmatrix}
 \\ \hphantom{\tilde W_{k}}{} =
 W_k +
 \begin{vmatrix} U_{k-1}^0 & U_{k-2}^1 \\ U_{k+1}^0 & U_{k}^1 \end{vmatrix}
 \Theta \varepsilon^{p-k+1} ( 1 + \order{\varepsilon} )
 + W_{k-1} \Theta^2 \varepsilon^{2(p-k)+1} ( 1 + \order{\varepsilon})
 \\ \hphantom{\tilde W_{k}}{} =
 \begin{cases}
 W_{k} + \order{\varepsilon}, & k \le p, \vspace{1mm}\\
 \dfrac{W_{k-1} \Theta^2 + \order{\varepsilon}}{\varepsilon^{2(k-p)-1}}, & k > p,
 \end{cases}
 \end{gather*}
 so that the peakon solution formulas take the form
 \begin{gather*}
 x_{N+1-k} = \ln \frac{\tilde U_{k+1}^0}{\tilde U_{k}^1}
 = \ln\frac{U_{k+1}^0 + U_{k}^0 \Theta \varepsilon^{p-k} ( 1 + \order{\varepsilon} )}{U_{k}^1 + U_{k-1}^1 \Theta \varepsilon^{p-k} ( 1 + \order{\varepsilon} )} \\
 \hphantom{x_{N+1-k}}{} =
 \begin{cases}
 \ln\dfrac{U_{k+1}^0 + \order{\varepsilon}}{U_{k}^1 + \order{\varepsilon}}, & k < p, \vspace{1mm}\\
 \ln\dfrac{U_{p+1}^0 + U_{p}^0 \Theta + \order{\varepsilon}}{U_{p}^1 + U_{p-1}^1 \Theta + \order{\varepsilon}}, & k = p, \vspace{1mm}\\
 \ln\dfrac{U_{k}^0 \Theta + \order{\varepsilon}}{U_{k-1}^1 \Theta + \order{\varepsilon}}, & k > p,
 \end{cases}
 \end{gather*}
 and
 \begin{gather*}
 m_{N+1-k} = \frac{\big(\tilde U_{k+1}^0 \tilde U_{k}^1\big)^2}{\tilde W_{k+1} \tilde W_{k}} = \frac{ \bigl( U_{k+1}^0 + U_{k}^0 \Theta \varepsilon^{p-k} ( 1 + \order{\varepsilon} ) \bigr)^2
 \bigl( U_{k}^1 + U_{k-1}^1 \Theta \varepsilon^{p-k} ( 1 + \order{\varepsilon} ) \bigr)^2
 }{\tilde W_{k+1} \tilde W_{k}} \\
 \hphantom{m_{N+1-k}}{} = \begin{cases}
 \dfrac{\bigl( U_{k+1}^0 + \order{\varepsilon} \bigr)^2 \bigl( U_{k}^1 + \order{\varepsilon} \bigr)^2}{( W_{k+1} + \order{\varepsilon}) ( W_{k} + \order{\varepsilon} )}, & k < p, \vspace{1mm}\\
 \dfrac{\varepsilon \bigl( U_{p+1}^0 + U_{p}^0 \Theta + \order{\varepsilon} \bigr)^2 \bigl( U_{p}^1 + U_{p-1}^1 \Theta + \order{\varepsilon} \bigr)^2}{\bigl( W_{p} \Theta^2 + \order{\varepsilon} \bigr) ( W_{p} + \order{\varepsilon} )}, & k = p, \vspace{1mm}\\
 \dfrac{\bigl( U_{k}^0 \Theta + \order{\varepsilon} \bigr)^2 \bigl( U_{k-1}^1 \Theta + \order{\varepsilon} \bigr)^2}{\bigl( W_{k} \Theta^2 + \order{\varepsilon} \bigr) \bigl( W_{k-1} \Theta^2 + \order{\varepsilon} \bigr)}, & k > p.
 \end{cases}
\end{gather*}
 In the limit $\varepsilon \to 0$, we see that $m_{N+1-p} \to 0$, while the other expressions reduce to those given in~\eqref{eq:DP-ghost-general} and~\eqref{eq:DP-N-peakons-renumbered}.
\end{proof}

In the same way as for Corollary~\ref{cor:CH-add-a-ghostpeakon}, we obtain the following result just by relabeling.

\begin{Corollary} \label{cor:DP-add-a-ghostpeakon} Write $U_k=U_k^0$ and $V_k=U_k^1$. For the Degasperis--Procesi pure $N$-peakon solution given by~\eqref{eq:DP-generalsolution}, namely
 \begin{gather*}
 x_{N+1-k}(t) = \ln \frac{U_k}{V_{k-1}}, \qquad m_{N+1-k}(t) = \frac{\bigl( U_k V_{k-1} \bigr)^2}{W_k W_{k-1}} ,\qquad 1 \le k \le N,
 \end{gather*}
 the characteristic curves $x = \xi(t)$ in the $k$th interval from the right,
 \begin{gather*}
 x_{N-k}(t) < \xi(t) < x_{N+1-k}(t)
 \end{gather*}
$($where $0 \le k \le N$, $x_0 = -\infty$, $x_{N+1} = +\infty)$, are given by
 \begin{gather*} \label{eq:DP-add-a-ghostpeakon}
 \xi(t) = \ln \frac{U_{k+1} + \theta U_{k}}{V_{k} + \theta V_{k-1}} ,\qquad \theta > 0.
 \end{gather*}
\end{Corollary}

\begin{Remark} \label{rem:DP-ghost-antipeakons} Although we have formulated Theorem~\ref{thm:DP-ghost} and Corollary~\ref{cor:DP-add-a-ghostpeakon} for pure peakon solutions, they are also valid (in the appropriate time interval) for mixed peakon--antipeakon solutions, in the generic case where all eigenvalues are simple and there are no resonances $\lambda_i + \lambda_j = 0$; see Remark~\ref{rem:DP-antipeakons}. The eigenvalues and residues may even be complex (appearing in complex-conjugate pairs). In the proof, given a ``target configuration'' of~$N$ eigenvalues and~$N$ residues for which we want to obtain the ghostpeakon formulas, we add a simple positive eigenvalue $\lambda_{N+1} = 1/\varepsilon > 0$ and a corresponding real nonzero residue $b_{N+1}$. However, it may happen that this residue $b_{N+1}$ must be \emph{negative} in order for $x_1 < \dots < x_{N+1}$ to hold, and according to~\eqref{eq:DP-reparametrization} this means that also $\theta$ must be negative. Such cases, where we need to use a~negative ghost parameter $\theta$, can be seen in Example~\ref{ex:DP-p-ap-asymm} below.
\end{Remark}

\begin{Example}[DP three-peakon characteristics] \label{ex:DP-3p-ghost} The trajectories in the pure three-peakon solution are given by
 \begin{gather*}
 x_1(t) = \ln\frac{U_3}{V_2} ,\qquad x_2(t) = \ln\frac{U_2}{V_1} ,\qquad x_3(t) = \ln\frac{U_1}{V_0} = \ln U_1 ,
 \end{gather*}
where the expressions $U_k$ and $V_k$ were written out in detail in Example~\ref{ex:DP-threepeakon-solution}. By Corollary~\ref{cor:DP-add-a-ghostpeakon}, the characteristic curves in the intervals outside and between the peakons are
 \begin{alignat*}{3}
& \xi(t)= \ln\frac{U_4 + \theta U_3}{V_3 + \theta V_2} = \ln\frac{\theta U_3}{V_3 + \theta V_2}\qquad && \text{in $( -\infty, x_1(t))$} ,& \\
& \xi(t)= \ln\frac{U_3 + \theta U_2}{V_2 + \theta V_1}\qquad && \text{in $( x_1(t), x_2(t) )$} ,& \\
& \xi(t)= \ln\frac{U_2 + \theta U_1}{V_1 + \theta V_0} = \ln\frac{U_2 + \theta U_1}{V_1 + \theta}\qquad && \text{in $( x_2(t), x_3(t))$} ,& \\
& \xi(t)= \ln\frac{U_1 + \theta U_0}{V_0 + \theta V_{-1}} = \ln(U_1 + \theta)\qquad && \text{in $( x_3(t), \infty)$} ,
 \end{alignat*}
where $0 < \theta < \infty$ in all cases. The curves look fairly similar to those for the Camassa--Holm case shown in Fig.~\ref{fig:CH-3p-ghost}.
\end{Example}

The following two examples illustrate what the characteristic curves look like for a peakon--antipeakon collision in the Degasperis--Procesi equation. Here our theorem only applies before the collision, since the solution continues after the collision in the form of a~so-called shockpeakon~\cite{lundmark:2007:shockpeakons}, and in that region we fall back to computing the characteristics by direct integration. There is also a degenerate symmetric case where our theorem does not apply at all, since already the peakon solution before the collision is given by exceptional formulas. We begin with this case, since it is a little simpler.

\begin{figure}[t] \centering

 \includegraphics{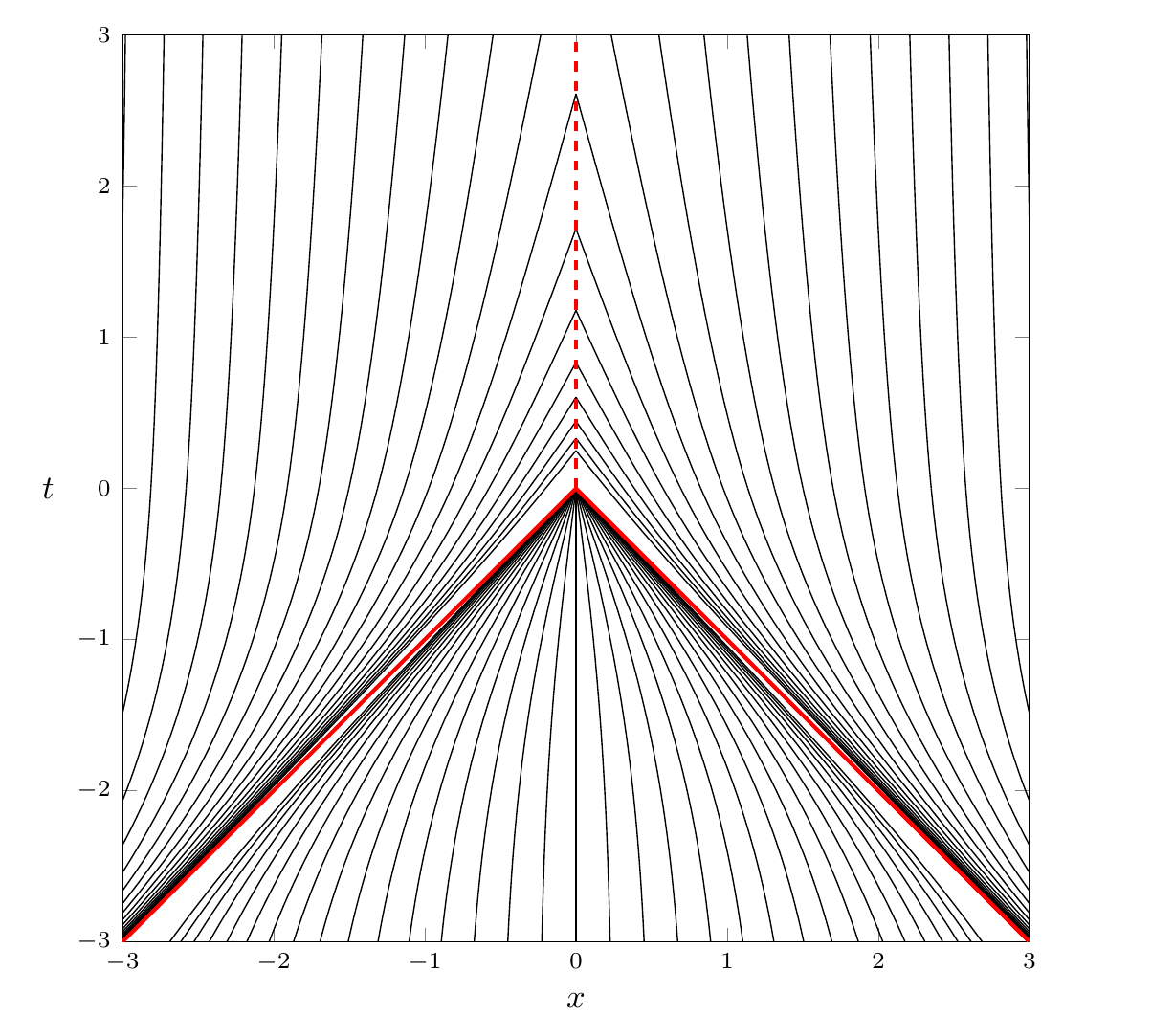}

\caption{Peakon trajectories and a selection of other characteristic curves for a symmetric peakon--antipeakon solution ($\lambda_1 = 1$, $\lambda_2=-1$) of the Degasperis--Procesi equation; see Example~\ref{ex:DP-p-ap-symm}. The peakon and the antipeakon travel with constant speed and collide head-on at the origin, forming a shockpeakon which remains stationary at $x=0$ (dashed line): $u(x,t) = \frac{- \sgn(x) {\rm e}^{-\abs{x}}}{\lambda + t}$ for $t \ge 0$. All characteristic curves have a finite lifespan; the curve through the point $(x,t) = (\xi_0,0)$ reaches the shock at time $t = \exp\bigl({\rm e}^{\abs{\xi_0}}-1\bigr) - 1$.} \label{fig:DP-p-ap-symm}
\end{figure}

\begin{Example}[DP symmetric collision]\label{ex:DP-p-ap-symm} Consider the solution of the Degasperis--Procesi equation where a peakon and an antipeakon of \emph{equal} strength collide and merge into a shockpeakon~\cite[Theorem~3.3]{lundmark:2007:shockpeakons}. This symmetric situation occurs when the eigenvalues satisfy $\lambda_{1,2}=\pm \lambda$, so that $\lambda_1+\lambda_2 = 0$. This is a resonant case where the usual formulas~\eqref{eq:DP-twopeakon-solution} do not apply (since they contain $\lambda_1 + \lambda_2$ in the denominators), but the solution can easily be found by direct integration instead. The results below are illustrated in Fig.~\ref{fig:DP-p-ap-symm} for $\lambda = 1$.

By translations along the $x$ and $t$ axes, we can arrange for the collision to occur at $(x,t)=(0,0)$; then the solution depends on the single parameter $\lambda > 0$. Before the collision, for $t<0$, we have $x_1(t) < 0 < x_2(t)$ and $m_1(t) > 0 > m_2(t)$, and the solution is
 \begin{gather*}
 x_1(t) = -x_2(t) = \frac{t}{\lambda} ,\qquad m_1(t) = -m_2(t) = \frac{1}{\lambda (1 - {\rm e}^{2t/\lambda})}, \qquad t < 0 .
 \end{gather*}
Thus,
 \begin{gather*}
 u(x,t) = m_1(t) {\rm e}^{-\abs{x-x_1(t)}} + m_2(t) {\rm e}^{-\abs{x-x_2(t)}} = \frac{{\rm e}^{-\abs{x-t/\lambda}} - {\rm e}^{-\abs{x+t/\lambda}}}{\lambda (1 - {\rm e}^{2t/\lambda})} \\
 \hphantom{ u(x,t)}{} =
 \begin{cases}
 \dfrac{{\rm e}^{x-t/\lambda} - {\rm e}^{x+t/\lambda}}{\lambda (1 - {\rm e}^{2t/\lambda})} = \dfrac{{\rm e}^{x-t/\lambda}}{\lambda} ,& x < t/\lambda ,\vspace{1mm}\\
 \dfrac{{\rm e}^{-x+t/\lambda} - {\rm e}^{x+t/\lambda}}{\lambda (1 - {\rm e}^{2t/\lambda})} = \dfrac{\sinh x}{\lambda \sinh(t/\lambda)} ,& t/\lambda \le x \le -t/\lambda ,\vspace{1mm}\\
 \dfrac{{\rm e}^{-x+t/\lambda} - {\rm e}^{-x-t/\lambda}}{\lambda (1 - {\rm e}^{2t/\lambda})} = \dfrac{-{\rm e}^{-x-t/\lambda}}{\lambda} ,& -t/\lambda < x ,
 \end{cases}
 \end{gather*}
 In particular, $u(x_1(t),t) = -u(x_2(t),t) = 1/\lambda$ for all $t<0$, in agreement with the fact that the peakons travel with the constant speed $1/\lambda$.

We can also easily find the characteristic curves by direct integration in this case. The characteristics $x = \xi(t)$ to the left of~$x_1$ are given by
 \begin{gather*}
 \dot\xi = \frac{{\rm e}^{\xi-t/\lambda}}{\lambda} \iff \frac{{\rm d}}{{\rm d}t} {\rm e}^{-\xi} = \frac{{\rm d}}{{\rm d}t} {\rm e}^{-t/\lambda}
 \iff \xi(t) = - \ln\bigl( \theta + {\rm e}^{-t/\lambda} \bigr) ,\qquad \theta = {\rm e}^{-\xi(0)} - 1 > 0 ,
 \end{gather*}
 and similarly to the right of~$x_2$,
 \begin{gather*}
 \xi(t) = \ln\bigl( \theta + {\rm e}^{-t/\lambda} \bigr) ,\qquad \theta = {\rm e}^{\xi(0)} - 1 > 0 .
 \end{gather*}
 In between, we have
 \begin{gather*}
 \dot\xi = \frac{\sinh \xi}{\lambda \sinh(t/\lambda)} \iff \int\frac{{\rm d}\xi}{\sinh \xi} = \int\frac{{\rm d}t}{\lambda \sinh(t/\lambda)}
 \iff \ln \abs{\tanh\frac{\xi}{2}} = \ln \abs{\tanh\frac{t}{2 \lambda}} + C \\
\hphantom{\dot\xi = \frac{\sinh \xi}{\lambda \sinh(t/\lambda)}}{} \iff \xi(t) = -2 \artanh\left( \tanh\frac{\theta}{2} \tanh\frac{t}{2\lambda} \right) ,
 \end{gather*}
 with $\theta = \lim\limits_{t \to -\infty} \xi(t) \in \R$.

The limiting wave profile at the collision is
 \begin{gather*}
 u(x,0) = \frac{- \sgn(x) {\rm e}^{-\abs{x}}}{\lambda} ,
 \end{gather*}
which has a jump of size $-2/\lambda$ at $x=0$. As explained in Lundmark's paper~\cite{lundmark:2007:shockpeakons}, the continuation of the solution picked out by a so-called entropy condition~\cite{coclite-karlsen:2006:DPwellposedness} is
 \begin{gather*}
 u(x,t) = \frac{- \sgn(x) {\rm e}^{-\abs{x}}}{\lambda + t} ,\qquad t \ge 0 ,
 \end{gather*}
a shockpeakon which just sits at $x=0$, with the size of the jump decaying to zero as $t \to \infty$. The characteristic curves for $t \ge 0$ to the left of the shock ($x < 0$) are given by
 \begin{gather*}
 \dot\xi = \frac{{\rm e}^{\xi}}{\lambda+t} \iff \xi(t) = -\ln \left( \theta - \ln \left( 1 + \frac{t}{\lambda} \right) \right) ,\qquad \theta = {\rm e}^{-\xi(0)} > 1 ,
 \end{gather*}
 and similarly to the right of the shock ($x > 0$) we have
 \begin{gather*}
 \xi(t) = \ln \left( \theta - \ln \left( 1 + \frac{t}{\lambda} \right) \right) ,\qquad \theta = {\rm e}^{\xi(0)} > 1 .
 \end{gather*}
 From this it follows that each characteristic curve has a finite lifespan; it ceases to exist when it collides with the shock, i.e., when $\xi(t)$ becomes zero, which happens when
 \begin{gather*}
 \frac{t}{\lambda} = {\rm e}^{\theta-1} - 1 = \exp\bigl({\rm e}^{\abs{\xi(0)}}-1\bigr) - 1 .
 \end{gather*}
\end{Example}

\begin{figure}[t] \centering

 \includegraphics{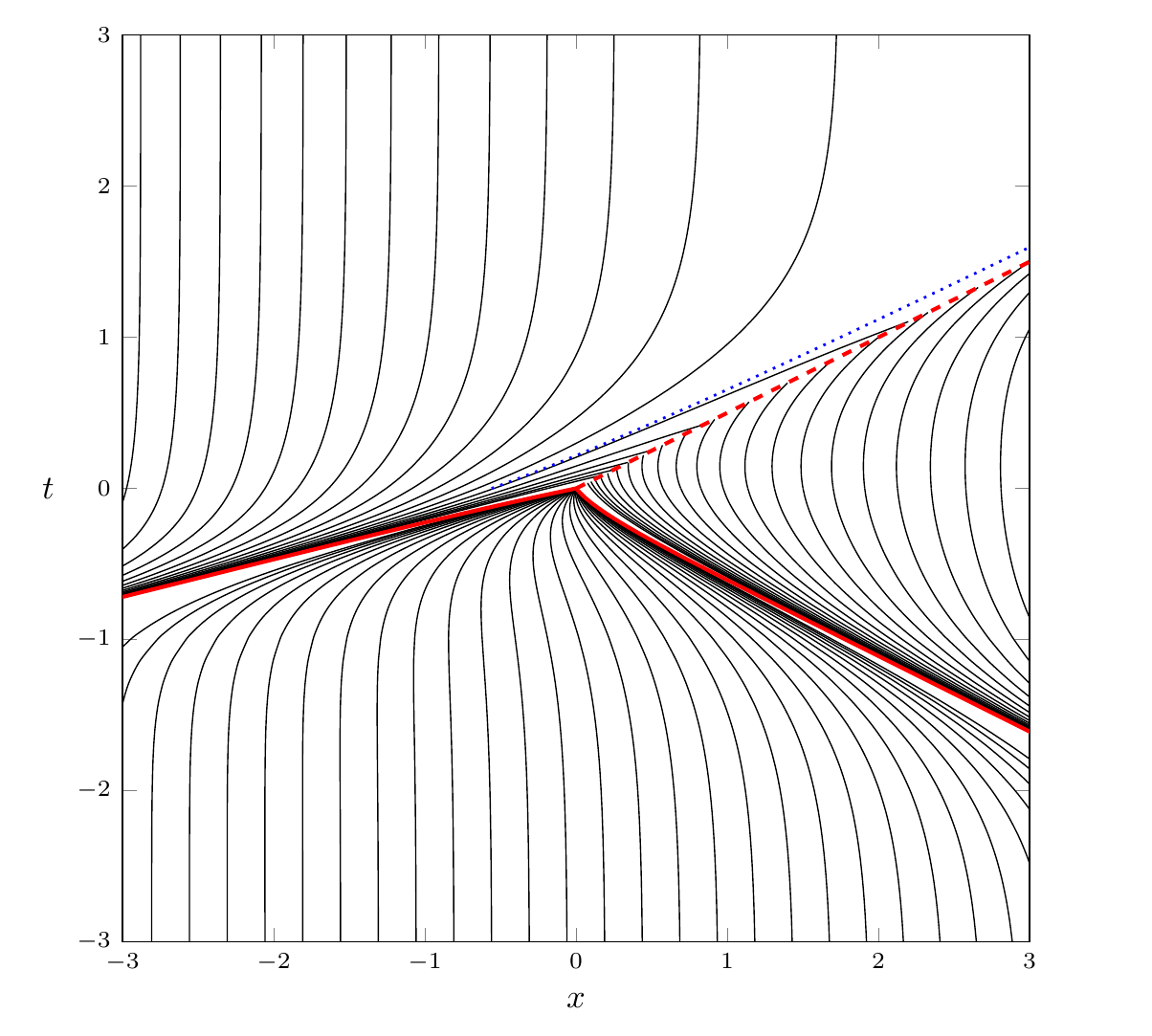}

\caption{Peakon trajectories and a selection of other characteristic curves for an asymmetric peakon--antipeakon solution ($\lambda_1 = 1/4$, $\lambda_2=-1/2$) of the Degasperis--Procesi equation; see Example~\ref{ex:DP-p-ap-asymm}. The limiting wave profile at the collision is a shockpeakon given by $u(x,0) = \big(v+\frac{1}{\sigma}\big) {\rm e}^x$ for $x<0$ and $u(x,0) = \big(v-\frac{1}{\sigma}\big) {\rm e}^{-x}$ for $x>0$, where $v = \frac{1}{\lambda_1} + \frac{1}{\lambda_2} = 2$ and $\sigma = \sqrt{-\lambda_1 \lambda_2} = \frac{1}{2 \sqrt 2}$. In particular, $u(0^\pm,0) = v \mp \frac{1}{\sigma}$, so there is a jump of size $-\frac{2}{\sigma}$ at the origin. The shockpeakon travels with constant velocity~$v$ (dashed line), and the jump at time~$t>0$ is $-\frac{2}{\sigma+t}$, which decays to zero as $t \to \infty$. The characteristic curve (dotted line) through the point $(x,t)=(\Xi,0)$, where $\Xi = - \ln \bigl( 1 - {\rm e}^{v\sigma} \Ei(-v\sigma) \bigr) \approx -0.55842$, approaches the shockpeakon trajectory $x=vt$ as $t \to \infty$. Characteristics to the left of that curve approach lines $x = {\rm const}$ as $t \to \infty$, and characteristics to the right of it hit the shock after finite time. All characteristics to the right of the shock turn from left to right at the same instant, namely $t = \frac{1}{v} - \sigma \approx 0.146$ (except, of course, those that have already hit the shock and ceased to exist).} \label{fig:DP-p-ap-asymm}
\end{figure}

\begin{Example}[DP asymmetric collision]\label{ex:DP-p-ap-asymm} Next, we look at the characteristics for the asymmetric peakon--antipeakon collision \cite[Theorem~3.5]{lundmark:2007:shockpeakons} which occurs when $\lambda_1 + \lambda_2 \neq 0$. Let us assume that the peakon at~$x_1$ is stronger than the antipeakon at~$x_2$, i.e., $m_1(t) > -m_2(t) > 0$. Then the shockpeakon which forms at the collision will move to the right. This case occurs when $b_1(0)$ and $b_2(0)$ are positive and $0 < \lambda_1 < -\lambda_2$, so that
\begin{gather*}
 \kappa := \sqrt{-\lambda_2/\lambda_1} > 1 .
 \end{gather*}
The results below are illustrated in Fig.~\ref{fig:DP-p-ap-asymm} for $\lambda_1 = 1/4$ and $\lambda_2 = -1/2$.

The solution is given by the same formulas~\eqref{eq:DP-twopeakon-solution} as in the pure peakon case, but only up until the time of collision, which is
 \begin{gather*}
 t_0 = \frac{1}{\lambda_1^{-1}-\lambda_2^{-1}} \log\left(\frac{\kappa^2 - \kappa}{\kappa+1}\frac{b_2(0)}{b_1(0)}\right) > 0 ,
 \end{gather*}
By translation, we can arrange for the collision to occur at $(x,t)=(0,0)$. This is accomplished by taking
 \begin{gather*}
 b_1(0) = \frac{\kappa^2 - \kappa}{\kappa^2 + 1} ,\qquad b_2(0) = \frac{\kappa + 1}{\kappa^2 + 1} .
 \end{gather*}
Indeed, to get the collision at $t=0$, we must take $\frac{b_1(0)}{b_2(0)} = \frac{\kappa^2-\kappa}{\kappa+1}$, and $x_2(0)=0$ if and only if $b_1(0) + b_2(0) = 1$. With these parameter values, the solution before the collision (i.e., for $t<0$) is given by
 \begin{gather*}
 x_1(t) = \ln \frac{\frac{(\lambda_1-\lambda_2)^2}{\lambda_1+\lambda_2}b_1 b_2}{\lambda_1 b_1 + \lambda_2 b_2} = -\ln \frac{\big(\kappa^2+\kappa\big) {\rm e}^{-t/\lambda_1} - (\kappa-1) {\rm e}^{-t/\lambda_2}}{\kappa^2+1} , \\
 x_2(t) = \ln (b_1+b_2) = \ln \frac{\big(\kappa^2-\kappa\big) {\rm e}^{t/\lambda_1} + (\kappa+1) {\rm e}^{t/\lambda_2}}{\kappa^2+1} , \\
 m_1(t) = \frac{(\lambda_1 b_1 + \lambda_2 b_2)^2}{\lambda_1 \lambda_2 \left( \lambda_1 b_1^2 + \lambda_2 b_2^2 + \frac{4 \lambda_1 \lambda_2}{\lambda_1+\lambda_2}b_1 b_2 \right)} \\
 \hphantom{m_1(t)}{} = \frac{1}{\lambda_1} \cdot \frac{\left( (\kappa-1) {\rm e}^{t/\lambda_1} - \big(\kappa^2+\kappa\big) {\rm e}^{t/\lambda_2} \right)^2}{\kappa^2 \left( (\kappa-1)^2 {\rm e}^{t/\lambda_1} + (\kappa+1)^2 {\rm e}^{t/\lambda_2} \right) \left( {\rm e}^{t/\lambda_2} - {\rm e}^{t/\lambda_1} \right)} > 0 , \\
 m_2(t) = \frac{(b_1+b_2)^2}{\lambda_1 b_1^2 + \lambda_2 b_2^2 + \frac{4 \lambda_1 \lambda_2}{\lambda_1+\lambda_2}b_1 b_2} \\
 \hphantom{m_2(t)}{} = \frac{1}{\lambda_2} \cdot
 \frac{\left( \big(\kappa^2-\kappa\big) {\rm e}^{t/\lambda_1} + (\kappa+1) {\rm e}^{t/\lambda_2} \right)^2}{\left( (\kappa-1)^2 {\rm e}^{t/\lambda_1} + (\kappa+1)^2 {\rm e}^{t/\lambda_2} \right) \left( {\rm e}^{t/\lambda_2} - {\rm e}^{t/\lambda_1} \right)} < 0 .
 \end{gather*}
Only the terms containing ${\rm e}^{t/\lambda_2}$ contribute asymptotically as $t \to -\infty$, since ${\rm e}^{t/\lambda_1} \to 0$ and ${\rm e}^{t/\lambda_2} \to \infty$. More precisely, ${\rm e}^{t/\lambda_1} = {\rm e}^{\delta t} {\rm e}^{t/\lambda_2}$, where $\delta = \frac{1}{\lambda_1} - \frac{1}{\lambda_2} > 0$, which implies that
 \begin{alignat*}{3}
 & x_1(t) = \frac{t}{\lambda_1} + \ln \frac{\kappa^2+1}{\kappa^2 + \kappa} + {\mathcal O} \big({\rm e}^{\delta t}\big) , \qquad && m_1(t) = \frac{1}{\lambda_1} + {\mathcal O} \big({\rm e}^{\delta t}\big) ,&\\
& x_2(t) = \frac{t}{\lambda_2} + \ln \frac{\kappa+1}{\kappa^2 + 1} + {\mathcal O} \big({\rm e}^{\delta t}\big) , \qquad && m_2(t) = \frac{1}{\lambda_2} + {\mathcal O} \big({\rm e}^{\delta t}\big) ,&
 \end{alignat*}
as $t \to -\infty$. To see the behaviour at the collision, we just use the Maclaurin expansions of the exponentials to get
 \begin{alignat*}{3}
& x_1(t) = \frac{\kappa^2+\kappa-1}{\kappa^2 \lambda_1} t + {\mathcal O}\big(t^2\big) , \qquad && m_1(t) = \frac{-1}{2t} + \frac{\kappa^2 + \kappa - 1}{2 \kappa^2 \lambda_1} + \order{t} ,&\\
& x_2(t) = \frac{\kappa^2-\kappa-1}{\kappa^2 \lambda_1} t + {\mathcal O}\big(t^2\big) , \qquad && m_2(t) = \frac{1}{2t} + \frac{\kappa^2 - \kappa - 1}{2 \kappa^2 \lambda_1} + \order{t} ,&
 \end{alignat*}
 as $t \to 0^-$. It follows that
 \begin{gather*}
 u(x_1(t),t) = m_1(t) + m_2(t) {\rm e}^{x_1(t)-x_2(t)} = \frac{1}{\lambda_1} + \frac{1}{\lambda_2} + \frac{1}{\sqrt{-\lambda_1 \lambda_2}} + \order{t} ,\\
 u(x_2(t),t) = m_1(t) {\rm e}^{x_1(t)-x_2(t)} + m_2(t) = \frac{1}{\lambda_1} + \frac{1}{\lambda_2} - \frac{1}{\sqrt{-\lambda_1 \lambda_2}} + \order{t} ,
 \end{gather*}
 as $t \to 0^-$. The limiting wave profile at the collision is thus
 \begin{gather*}
 u(x,0) = \left( v - \frac{\sgn(x)}{\sigma} \right) {\rm e}^{-\abs{x}} ,
 \end{gather*}
 where
 \begin{gather*}
 v = \frac{1}{\lambda_1} + \frac{1}{\lambda_2} > 0 ,\qquad \sigma = \sqrt{-\lambda_1 \lambda_2} > 0 .
 \end{gather*}
It has a jump of size $-2/\sigma$ at $x=0$, and continues as a shockpeakon moving to the right with constant velocity~$v$ and shock strength decaying to zero as $t\to \infty$:
 \begin{gather} \label{eq:DP-continuation-shockpeakon}
 u(x,t) = \left( v - \frac{\sgn(x-vt)}{\sigma + t} \right) {\rm e}^{-\abs{x-vt}} ,\qquad t \ge 0 .
 \end{gather}

The characteristic curves before the collision are given by our general ghostpeakon formula~\eqref{eq:DP-ghost-general}, except that we must let the parameter~$\theta$ be \emph{negative} in some cases; cf.\ Remark~\ref{rem:DP-ghost-antipeakons}. To the left of~$x_1$, we have
 \begin{gather*}
 \xi_{\text{left}}(t) = \ln \frac{\theta \frac{(\lambda_1-\lambda_2)^2}{\lambda_1+\lambda_2}b_1 b_2}{\frac{(\lambda_1-\lambda_2)^2}{\lambda_1+\lambda_2} \lambda_1 \lambda_2 b_1 b_2 + \theta(\lambda_1 b_1 + \lambda_2 b_2)} \\
\hphantom{\xi_{\text{left}}(t)}{} = -\ln \left( \frac{(\kappa^2+\kappa) {\rm e}^{-t/\lambda_1} - (\kappa-1) {\rm e}^{-t/\lambda_2}}{\kappa^2+1} + \frac{\lambda_1 \lambda_2}{\theta} \right) ,
\end{gather*}
where $-\infty<\theta<0$, with the asymptotics
 \begin{gather*}
 \xi_{\text{left}}(t) = \frac{t}{\lambda_1} + \ln \frac{\kappa^2+1}{\kappa^2 + \kappa} + {\mathcal O}\big({\rm e}^{t/\lambda_1}\big),\qquad t \to -\infty ,
 \end{gather*}
 and
 \begin{gather*}
 \xi_{\text{left}}(t) = -\ln \left( 1 + \frac{\lambda_1 \lambda_2}{\theta} \right) + \order{t} ,\qquad t \to 0^- .
 \end{gather*}
 The characteristics to the right of~$x_2$ are
 \begin{gather*}
 \xi_{\text{right}}(t) = \ln (b_1+b_2 + \theta) = \ln \left( \frac{(\kappa^2-\kappa) {\rm e}^{t/\lambda_1} + (\kappa+1) {\rm e}^{t/\lambda_2}}{\kappa^2+1} + \theta \right),
 \end{gather*}
 where $0<\theta<\infty$, with the asymptotics
 \begin{gather*}
 \xi_{\text{right}}(t) = \frac{t}{\lambda_2} + \ln \frac{\kappa+1}{\kappa^2 + 1} + {\mathcal O}\big({\rm e}^{-t/\lambda_2}\big) ,\qquad t \to -\infty ,
 \end{gather*}
 and
 \begin{gather*}
 \xi_{\text{right}}(t) = \ln(1+\theta) + \order{t} ,\qquad t \to 0^- .
 \end{gather*}
 And between $x_1$ and~$x_2$, the characteristics are
 \begin{gather*}
 \xi_{\text{mid}}(t) = \ln \frac{\frac{(\lambda_1-\lambda_2)^2}{\lambda_1+\lambda_2}b_1 b_2 + \theta(b_1+b_2)}{\lambda_1 b_1 + \lambda_2 b_2 + \theta} \\
 \hphantom{\xi_{\text{mid}}(t)}{} = \ln \frac{-\big(\kappa^3+\kappa\big) {\rm e}^{t/\lambda_1} {\rm e}^{t/\lambda_2} + \frac{\theta}{\lambda_1} \bigl( \big(\kappa^2-\kappa\big) {\rm e}^{t/\lambda_1} + (\kappa+1) {\rm e}^{t/\lambda_2} \bigr)}{(\kappa^2-\kappa) {\rm e}^{t/\lambda_1} - \big(\kappa^3+\kappa^2\big) {\rm e}^{t/\lambda_2} + \frac{\theta}{\lambda_1} (\kappa^2+1)},
 \end{gather*}
 where $-\infty<\theta<0$, with the asymptotics
 \begin{gather*}
 \xi_{\text{mid}}(t) = \ln \frac{-\theta (\kappa+1)}{\lambda_1 \big(\kappa^3+\kappa^2\big)} + {\mathcal O}\big({\rm e}^{-t/\lambda_2}\big) ,\qquad t \to -\infty ,
 \end{gather*}
 and
 \begin{gather*}
 \xi_{\text{mid}}(t) = \frac{\big(\kappa^2 + \kappa - 1\big) + \frac{-\theta}{\kappa \lambda_1} \big(\kappa^2 - \kappa - 1\big)}{\kappa^2 \lambda_1 \left( 1 + \frac{-\theta}{\kappa \lambda_1} \right)} t
 + {\mathcal O}\big(t^2\big) ,\qquad t \to 0^- .
 \end{gather*}
The rate of steepening of the slope~$u_x$ along this middle family of characteristic curves turns out to be
 \begin{gather*}
 u_x\bigl( \xi_{\text{mid}}(t),t) \bigr) = - m_1(t) {\rm e}^{x_1(t) - \xi_{\text{mid}}(t)} + m_2(t) {\rm e}^{\xi_{\text{mid}}(t) - x_2(t)}
 = \frac{1}{t} + \order{t} ,\qquad t \to 0^- ,
\end{gather*}
for every $\theta<0$. We may note that this agrees with the blowup rate for solutions with initial data in $H^s(\R)$ for $s>3/2$ given by Escher, Liu and Yin~\cite[Theorem~3.1]{escher-liu-yin:2006:DP-globalweak-blowup},
 even though peakons are not smooth enough to be in those function spaces.

Let us also investigate the characteristic curves of the shockpeakon solution~\eqref{eq:DP-continuation-shockpeakon}, which provides the continuation of the solution past the collision. Here we use direct integration again.
 The ODE for the characteristics starting to the left of the shock ($\xi(0)<0$) is
\begin{gather*}
 \dot\xi = \left( v + \frac{1}{\sigma + t} \right) {\rm e}^{\xi-vt} ,
\end{gather*}
 which gives
 \begin{gather*}
{\rm e}^{-\xi(t)} - {\rm e}^{-\xi(0)} = - \int_0^t \left( v + \frac{1}{\sigma + \tau} \right) {\rm e}^{-v \tau} {\rm d}\tau = - \int_0^{vt} \left( 1 + \frac{1}{v\sigma + s} \right) {\rm e}^{-s} {\rm d}s \\
\hphantom{{\rm e}^{-\xi(t)} - {\rm e}^{-\xi(0)}}{} = - {\rm e}^{v\sigma} \int_{v\sigma}^{v\sigma+vt} \left( 1 + \frac{1}{r} \right) {\rm e}^{-r} {\rm d}r = {\rm e}^{v\sigma} ( F(v\sigma+vt) - F(v\sigma)) ,
 \end{gather*}
 where $F$ is a decreasing function determined by
 \begin{gather*}
 F'(r) = - \left( 1 + \frac{1}{r} \right) {\rm e}^{-r} ,\qquad r > 0 .
 \end{gather*}
If we normalize by requiring that $F(r) \to 0$ as $r \to \infty$, we can write
 \begin{gather*}
 F(r) = {\rm e}^{-r} - \Ei(-r) ,
 \end{gather*}
where $\Ei$ is the standard exponential integral available in many software packages,
 \begin{gather*} \label{eq:Ei}
 \Ei(z) = - \int_{-z}^{\infty} \frac{{\rm e}^{-\zeta}}{\zeta} {\rm d}\zeta .
 \end{gather*}
The function $\Ei(z)$ has the following properties for real~$z$: it is negative and decreasing for $z < 0$, with
 \begin{gather*}
 \lim_{z \to -\infty} \Ei(z) = 0 ,\qquad \lim_{z \to 0^-} \Ei(z) = -\infty ,
 \end{gather*}
and it is increasing for $z > 0$, with
 \begin{gather*}
 \lim_{z \to 0^+} \Ei(z) = -\infty ,\qquad \Ei(z)={\rm e}^z \bigl( z^{-1} + \mathcal{O}\big(z^{-2}\big) \bigr) \qquad \text{ as $z \to \infty$} .
 \end{gather*}
 There is a unique positive zero $z_0 \approx 0.372507$.

 Since $\Ei(-r) < 0$ for $r > 0$, we have
 \begin{gather} \label{eq:F-exp-ineqality}
 F(r) > {\rm e}^{-r} .
 \end{gather}
 Thus, the characteristics to the left of the shock are
 \begin{gather*}
 x = \xi(t) = - \ln \bigl( {\rm e}^{-\xi(0)} + {\rm e}^{v\sigma} ( F(v\sigma+vt) - F(v\sigma)) \bigr) .
 \end{gather*}
 This is an increasing function of $\xi(0)$, so different characteristic curves never intersect. Note that the expression in brackets,
 \begin{gather*}
 {\rm e}^{-\xi(0)} + {\rm e}^{v\sigma} ( F(v\sigma+vt) - F(v\sigma)) ,
 \end{gather*}
is a decreasing function (which agrees with $\xi(t)$ being increasing), and it tends to
 \begin{gather*}
 {\rm e}^{-\xi(0)} - {\rm e}^{v\sigma} F(v\sigma)
 \end{gather*}
as $t \to \infty$. Since ${\rm e}^{v\sigma} F(v\sigma) > 1$ by~\eqref{eq:F-exp-ineqality}, there is exactly one initial value
 \begin{gather*}
 \xi(0) = - \ln \bigl( {\rm e}^{v\sigma} F(v\sigma) \bigr) = - \ln \bigl( 1 - {\rm e}^{v\sigma} \Ei(-v\sigma) \bigr) =: \Xi < 0
 \end{gather*}
for which this limit equals zero. This particular characteristic curve,
 \begin{gather*}
 x = \xi(t) = - \ln \bigl( {\rm e}^{v\sigma} F(v\sigma+vt) \bigr) ,
 \end{gather*}
 therefore tends to $\infty$ as $t \to \infty$, and by \eqref{eq:F-exp-ineqality} it satisfies
 \begin{gather*}
 \xi(t) < -\ln \bigl( {\rm e}^{v\sigma} {\rm e}^{-(v\sigma+vt)} \bigr) = vt ,
 \end{gather*}
so that it always stays to the left of the shock (which travels along the curve $x=vt$). In fact, asymptotically it catches up with the shock:
 \begin{gather*}
 \xi(t) = - \ln \bigl( {\rm e}^{v\sigma} \bigl( {\rm e}^{-(v\sigma+vt)} - \Ei(-(v\sigma+vt)) \bigr) \bigr) = - \ln \bigl( {\rm e}^{-vt} - {\rm e}^{v\sigma} \underbrace{\Ei(-v\sigma-vt)}_{\to 0} \bigr) \\
 \hphantom{\xi(t)}{} = vt + o(1) ,\qquad \text{as $t \to \infty$} .
 \end{gather*}
The characteristic curves to the left of this special curve, i.e., those with $\xi(0) < \Xi$, asymptotically slow down to a halt,
 \begin{gather*}
 \lim_{t \to \infty} \xi(t) = -\ln \bigl( {\rm e}^{-\xi(0)} - {\rm e}^{-\Xi} \bigr) ,
 \end{gather*}
while the curves with $\Xi < \xi(0) < 0$ have a finite lifespan; they cease to exist when they cross the shock curve $x=vt$, which they must do, since the expression for $\xi(t)$ tends to $\infty$ in finite time.

Similarly, the characteristics starting to the right of the shock ($\xi(0)>0$) satisfy
 \begin{gather*}
 \dot\xi = \left( v - \frac{1}{\sigma + t} \right) {\rm e}^{vt-\xi} ,
 \end{gather*}
so
 \begin{gather*}
 {\rm e}^{\xi(t)} - {\rm e}^{\xi(0)} = \int_0^t \left( v - \frac{1}{\sigma + \tau} \right) {\rm e}^{v \tau} {\rm d}\tau = \int_0^{vt} \left( 1 - \frac{1}{v\sigma + s} \right) {\rm e}^{s} {\rm d}s \\
 \hphantom{{\rm e}^{\xi(t)} - {\rm e}^{\xi(0)}}{}
 = {\rm e}^{-v\sigma} \int_{v\sigma}^{v\sigma+vt} \left( 1 - \frac{1}{r} \right) {\rm e}^{r} {\rm d}r = {\rm e}^{-v\sigma} \bigl( G(v\sigma+vt) - G(v\sigma) \bigr) ,
 \end{gather*}
 where $G(z)$ is uniquely determined up to an additive constant by
 \begin{gather*}
 G'(r) = \left( 1 - \frac{1}{r} \right) {\rm e}^{r} ,\qquad r > 0 .
 \end{gather*}
 For definiteness, let us take
 \begin{gather*}
 G(r) = {\rm e}^r - \Ei(r) \qquad
 \bigl( = F(-r) \bigr) .
 \end{gather*}
 Thus, the characteristic curves to the right of the shock are
 \begin{gather*}
 x = \xi(t) = \ln \bigl( {\rm e}^{\xi(0)} + {\rm e}^{-v\sigma} \bigl( G(v\sigma+vt) - G(v\sigma) \bigr) \bigr) .
 \end{gather*}
This is an increasing function of $\xi(0)$, so different characteristic curves never intersect. Note that~$G(r)$ is decreasing for $0 < r \le 1$ and increasing for $r \ge 1$. This means that $\xi(t)$ is increasing for all $t \ge 0$ if $\sigma v \ge 1$, and initially decreasing and then increasing if $\sigma v < 1$. This is as expected, since $v < 1/\sigma$ means that the right part of the shockpeakon formed at the collision dips down to $u<0$ (recall that $u(0^\pm,0) = v \mp \frac{1}{\sigma}$), making the characteristics to the right of the shock go left until the shock strength has decayed enough for $u$ to be positive everywhere. Thus, all these characteristics turn from left to right at the same instant, namely when $\inf\limits_{x \in \R} u(x,t) = v - \frac{1}{\sigma+t}$ becomes zero so that $u$ is identically zero to the right of the shock, i.e., when $t = \frac{1}{v} - \sigma$; this can be seen in Fig.~\ref{fig:DP-p-ap-asymm}.

 These curves all have a finite lifespan, since the shock catches up with each one of them sooner or later. Indeed,
 \begin{gather*}
vt < \xi(t) \iff {\rm e}^{vt} < {\rm e}^{\xi(0)} + {\rm e}^{-v\sigma} \bigl( G(v\sigma+vt) - G(v\sigma) \bigr) \\
\hphantom{vt < \xi(t)}{} \iff {\rm e}^{v\sigma + vt} - G(v\sigma+vt) < {\rm e}^{v\sigma + \xi(0)} - G(v\sigma) \\
\hphantom{vt < \xi(t)}{} \iff \Ei(v\sigma + vt) < {\rm e}^{v\sigma + \xi(0)} - G(v\sigma) .
 \end{gather*}
Here, the right-hand side is a constant greater than ${\rm e}^{v\sigma+0} - G(v\sigma) = \Ei(v\sigma)$, and the left-hand side increases from $\Ei(v\sigma)$ to~$\infty$ as $t$ goes from $0$ to~$\infty$, so there is a unique value of~$t$ (depending on $\xi(0)$, and of course also on $v\sigma$) for which the shock reaches the characteristic curve.
\end{Example}

\section{Novikov ghostpeakons}\label{sec:Novikov-ghost}

The derivation of the ghostpeakon formulas for Novikov's equation is rather similar to what we have seen for the Camassa--Holm and Degasperis--Procesi equations in Theorems~\ref{thm:CH-ghost} and~\ref{thm:DP-ghost} -- especially the latter, as far as notation is concerned. As in the DP case, we formulate the theorems for pure peakon solutions first, and comment on their validity for mixed peakon--antipeakon solutions in Remark~\ref{rem:Novikov-ghost-antipeakons}.

\begin{Theorem} \label{thm:Novikov-ghost} Fix some $p$ with $0 \le p \le N$. The solution of the Novikov $(N+1)$-peakon ODEs~\eqref{eq:Novikov-peakon-ODEs-shorthand} with $x_1 < \dotsb < x_{N+1}$ and all amplitudes $m_k(t)$ positive except for $m_{N+1-p}(t) = 0$ is as follows: the position of the ghostpeakon is given by
 \begin{gather} \label{eq:Novikov-ghost-general}
 x_{N+1-p}(t) = \frac12 \ln\frac{Z_{p+1} + \theta Z_{p}}{W_{p} + \theta W_{p-1}} ,\qquad 0 < \theta < \infty ,
 \end{gather}
 while the other peakons are given by the general solution formulas \eqref{eq:Novikov-generalsolution} up to renumbering:
 \begin{gather}
 x_{N+1-k}(t) = \begin{cases}
 \dfrac12 \ln\dfrac{Z_{k+1}}{W_{k}}, & 0 \le k < p, \vspace{1mm}\\
 \dfrac12 \ln\dfrac{Z_{k}}{W_{k-1}}, & p < k \le N,
 \end{cases}\nonumber\\
 m_{N+1-k}(t) =
 \begin{cases}
 \dfrac{\sqrt{Z_{k+1} W_{k}}}{U_{k+1} U_{k}}, & 0 \le k < p, \vspace{1mm}\\
 \dfrac{\sqrt{Z_{k} W_{k-1}}}{U_{k} U_{k-1}}, & p < k \le N.
 \end{cases}\label{eq:Novikov-N-peakons-renumbered}
 \end{gather}
\end{Theorem}

\begin{proof} Consider the reparametrization
 \begin{gather} \label{eq:Novikov-reparametrization}
 \varepsilon = \frac{1}{\sqrt{\lambda_{N+1}}} ,\qquad \theta = \lambda_{N+1}^{2p-1} b_{N+1}(0)^2 ,
 \end{gather}
 which (for $\varepsilon>0$) is equivalent to
 \begin{gather} \label{eq:Novikov-alt-reparametrization}
 \lambda_{N+1} = \frac{1}{\varepsilon^2} ,\qquad b_{N+1}(0) = \varepsilon^{2p-1} \sqrt{\theta} .
 \end{gather}
 The same computation as in the proof of Theorem~\ref{thm:DP-ghost}, but with $\varepsilon^2$, $p-\frac12$ and $\sqrt{\theta}$ instead of~$\varepsilon$,~$p$ and~$\theta$, gives
 \begin{gather*}
 \tilde U_{k}^a = U_{k}^a + U_{k-1}^a \sqrt{\Theta} \varepsilon^{2(p-k-a)+1} ( 1 + \order{\varepsilon} ) \qquad (\text{as $\varepsilon \to 0$}) ,
 \end{gather*}
 where $\Theta = \Theta(t) = \lambda_{N+1}^{2p-1} b_{N+1}(t)^2 = \theta {\rm e}^{2 \varepsilon^2 t}$, and thus
 \begin{gather*}
 \tilde W_{k} =
 \begin{vmatrix} U_{k}^0 + U_{k-1}^0 \sqrt{\Theta} \varepsilon^{2(p-k)+1} ( 1 + \order{\varepsilon} ) & U_{k-1}^1 + U_{k-2}^1 \sqrt{\Theta} \varepsilon^{2(p-k)+1} ( 1 + \order{\varepsilon}) \vspace{1mm}\\ U_{k+1}^0 + U_{k}^0 \sqrt{\Theta} \varepsilon^{2(p-k)-1} ( 1 + \order{\varepsilon} ) & U_{k}^1 + U_{k-1}^1 \sqrt{\Theta} \varepsilon^{2(p-k)-1} ( 1 + \order{\varepsilon} ) \end{vmatrix}
 \\
 \hphantom{\tilde W_{k}}{} =
 W_k +
 \begin{vmatrix} U_{k-1}^0 & U_{k-2}^1 \\[1ex] U_{k+1}^0 & U_{k}^1 \end{vmatrix}
 \sqrt{\Theta} \varepsilon^{2(p-k)+1} ( 1 + \order{\varepsilon} )
 + W_{k-1} \Theta \varepsilon^{4(p-k)} ( 1 + \order{\varepsilon} )
 \\
\hphantom{\tilde W_{k}}{} =
 \begin{cases}
 W_{k} + \order{\varepsilon}, & k < p, \vspace{1mm}\\
 W_{p} + W_{p-1} \Theta + \order{\varepsilon}, & k = p, \vspace{1mm}\\
 \dfrac{W_{k-1} \Theta + \order{\varepsilon}}{\varepsilon^{4(k-p)}}, & k > p.
 \end{cases}
 \end{gather*}
 Similarly,
 \begin{gather*}
 \tilde Z_{k} =
 \begin{vmatrix} U_{k}^{-1} + U_{k-1}^{-1} \sqrt{\Theta} \varepsilon^{2(p-k)+3} ( 1 + \order{\varepsilon} ) & U_{k-1}^0 + U_{k-2}^0 \sqrt{\Theta} \varepsilon^{2(p-k)+3} ( 1 + \order{\varepsilon}) \vspace{1mm}\\ U_{k+1}^{-1} + U_{k}^{-1} \sqrt{\Theta} \varepsilon^{2(p-k)+1} ( 1 + \order{\varepsilon} ) & U_{k}^0 + U_{k-1}^0 \sqrt{\Theta} \varepsilon^{2(p-k)+1} ( 1 + \order{\varepsilon} ) \end{vmatrix} \\
 \hphantom{\tilde Z_{k}}{} =
 Z_k +
 \begin{vmatrix} U_{k-1}^{-1} & U_{k-2}^0 \vspace{1mm}\\ U_{k+1}^{-1} & U_{k}^0 \end{vmatrix}
 \sqrt{\Theta} \varepsilon^{2(p-k)+3} ( 1 + \order{\varepsilon} )
 + Z_{k-1} \Theta \varepsilon^{4(p-k)+4} ( 1 + \order{\varepsilon} )
 \\
\hphantom{\tilde Z_{k}}{} =
 \begin{cases}
 Z_{k} + \order{\varepsilon}, & k < p+1, \vspace{1mm}\\
 Z_{p+1} + Z_{p} \Theta + \order{\varepsilon}, & k = p+1, \vspace{1mm}\\
 \dfrac{Z_{k-1} \Theta + \order{\varepsilon}}{\varepsilon^{4(k-p-1)}}, & k > p+1.
 \end{cases}
 \end{gather*}
 Thus, the peakon solution formulas reduce to
 \begin{gather*}
 x_{N+1-k}(t) =
 \dfrac12 \ln\dfrac{\tilde Z_{k+1}}{\tilde W_{k}} =
 \begin{cases}
 \dfrac12 \ln\dfrac{Z_{k+1} + \order{\varepsilon}}{W_{k} + \order{\varepsilon}}, & k < p, \vspace{1mm}\\
 \dfrac12 \ln\dfrac{Z_{p+1}^0 + Z_{p}^0 \Theta + \order{\varepsilon}}{W_{p}^1 + W_{p-1}^1 \Theta + \order{\varepsilon}}, & k = p, \vspace{1mm}\\
 \dfrac12 \ln\dfrac{Z_{k}^0 \Theta + \order{\varepsilon}}{W_{k-1}^1 \Theta + \order{\varepsilon}}, & k > p,
 \end{cases}
 \end{gather*}
 and
 \begin{gather*}
 m_{N+1-k}(t) = \dfrac{\sqrt{\tilde Z_{k+1} \tilde W_{k}}}{\tilde U_{k+1} \tilde U_{k}} \\
\hphantom{m_{N+1-k}(t)}{} =
 \begin{cases}
 \dfrac{\sqrt{( Z_{k+1} + \order{\varepsilon}) ( W_{k} + \order{\varepsilon})}}{( U_{k+1} + \order{\varepsilon}) ( U_{k} + \order{\varepsilon} )}, & k < p, \vspace{1mm}\\
 \dfrac{\varepsilon \sqrt{\bigl( Z_{p+1} + Z_{p}^0 \sqrt{\Theta} + \order{\varepsilon} \bigr) \bigl( W_{p} + W_{p-1} \sqrt{\Theta} + \order{\varepsilon} \bigr)}}{\bigl( U_{p} \sqrt{\Theta} + \order{\varepsilon} \bigr) \bigl( U_{p} + \order{\varepsilon} \bigr)}, & k = p, \vspace{1mm}\\
 \dfrac{\sqrt{( Z_{k} \Theta + \order{\varepsilon}) ( W_{k-1} \Theta + \order{\varepsilon} )}}{\bigl( U_{k} \sqrt{\Theta} + \order{\varepsilon} \bigr) \bigl( U_{k-1} \sqrt{\Theta} + \order{\varepsilon} \bigr)}, & k > p.
 \end{cases}
 \end{gather*}
 It follows that $m_{N+1-p} \to 0$ as $\varepsilon \to 0$, while the other expressions tend to those given in \eqref{eq:Novikov-ghost-general} and~\eqref{eq:Novikov-N-peakons-renumbered}.
\end{proof}

\begin{Corollary} \label{cor:Novikov-add-a-ghostpeakon} For the Novikov pure $N$-peakon solution given by~\eqref{eq:Novikov-generalsolution},
 \begin{gather*}
 x_{N+1-k}(t) = \frac12 \ln\dfrac{Z_{k}}{W_{k-1}} , \qquad m_{N+1-k}(t) = \frac{\sqrt{Z_{k} W_{k-1}}}{U_{k} U_{k-1}} ,\qquad 1 \le k \le N,
 \end{gather*}
 the characteristic curves $x = \xi(t)$ in the $k$th interval from the right,
 \begin{gather*}
 x_{N-k}(t) < \xi(t) < x_{N+1-k}(t) , \qquad 0 \le k \le N ,
 \end{gather*}
 are given by
 \begin{gather}
 \label{eq:Novikov-add-a-ghostpeakon}
 \xi(t)
 = \frac12 \ln\frac{Z_{k+1} + \theta Z_{k}}{W_{k} + \theta W_{k-1}}
 ,\qquad
 \theta > 0.
 \end{gather}
\end{Corollary}

\begin{Remark} \label{rem:Novikov-ghost-antipeakons} Pure antipeakon solutions are obtained from pure peakon solutions by changing~$u$ to~$-u$, which in terms of spectral data is effected by keeping the eigenvalues~$\lambda_k$ and changing the sign of every~$b_k$. In the proof of Theorem~\ref{thm:Novikov-ghost}, we see from~\eqref{eq:Novikov-alt-reparametrization} that we can change the sign of~$b_{N+1}$ by letting $\varepsilon \to 0^-$ instead of~$0^+$, i.e., we take $\varepsilon = -1/\sqrt{\lambda_{N+1}}$ in~\eqref{eq:Novikov-reparametrization}. Everything else is unchanged, including the positivity of~$\theta$, so the ghostpeakon formulas are valid without change for pure antipeakon solutions. Similarly it is seen that they are valid for (conservative) peakon--antipeakon solutions, in the \emph{generic} case where the (possibly complex) eigenvalues are \emph{simple} and have \emph{positive real part}; such solutions are still given by the usual formulas~\eqref{eq:Novikov-generalsolution}. For every such solution~$u$ there is a corresponding solution~$-u$ obtained by changing the signs of all~$b_k$, and the two are handled by taking $\varepsilon = \pm 1/\sqrt{\lambda_{N+1}}$. Unlike the DP case, there is never any need to use negative~$\theta$.

Novikov's equation also admits a great variety of non-generic peakon--antipeakon solutions, where some eigenvalues may have multiplicity greater than one and/or lie on the imaginary axis. Those solutions are described by separate formulas (see Kardell and Lundmark~\cite{kardell-lundmark:2016p:novikov-peakon-antipeakon}), where the quantities $U_k^a$ (and consequently $W_k$ etc.) are modified via limiting procedures that commute with the limit $\varepsilon \to 0$ considered here. It follows that the ghostpeakon formulas are valid also for solutions of this kind, provided only that all $W_k$ and~$Z_k$ are replaced with their suitably modified counterparts.
\end{Remark}

Next, we will prove formula~\eqref{eq:Novikov-ughost} for the value of~$u$ along a characteristic curve, which will allow us make parametric plots of multipeakon solutions,
\begin{gather*}
 (\theta,t) \mapsto (x,t,u) = \bigl( \xi(t,\theta), t, u(\xi(t,\theta),t) \bigr) .
\end{gather*}
We can use $u(\xi)=\dot \xi$ for this purpose in the CH and DP cases, but we need another formula for~$u(\xi)$ in the Novikov case, since $\dot \xi = u(\xi)^2$ gives us no information about the sign of~$u(\xi)$.

\begin{Theorem} \label{thm:Novikov-u-and-ughost} The Novikov $N$-peakon solution~\eqref{eq:Novikov-generalsolution} satisfies
 \begin{gather} \label{eq:Novikov-u}
 u(x_{N+1-k}) = \frac{Y_k}{\sqrt{Z_k W_{k-1}}} ,
 \end{gather}
 where
 \begin{gather*}
 Y_k = \begin{vmatrix}
 T_k & V_{k-2} \\
 T_{k+1} & V_{k-1}
 \end{vmatrix}
 =
 \begin{vmatrix}
 U_k^{-1} & U_{k-2}^1 \\
 U_{k+1}^{-1} & U_{k-1}^1
 \end{vmatrix} .
 \end{gather*}
 Moreover, if $x=\xi(t)$ is the characteristic curve~\eqref{eq:Novikov-add-a-ghostpeakon}, then
 \begin{gather} \label{eq:Novikov-ughost}
 u(\xi) = \frac{Y_{k+1} + \theta Y_k}{\sqrt{(Z_{k+1} + \theta Z_k)(W_k + \theta W_{k-1})}} .
 \end{gather}
\end{Theorem}

\begin{proof} Equation~\eqref{eq:Novikov-u} is proved by direct computation (with empty sums interpreted as zero):
 \begin{gather*}
 u(x_{N+1-k}) = \sum_{i=1}^N m_{N+1-i} {\rm e}^{-\abs{x_{N+1-k} - x_{N+1-i}}} \\
 \hphantom{u(x_{N+1-k})}{} = \sum_{1 \le i < k} m_{N+1-i} \frac{{\rm e}^{x_{N+1-k}}}{{\rm e}^{x_{N+1-i}}}
 + m_{N+1-k} + \sum_{k < i \le N} m_{N+1-i} \frac{{\rm e}^{x_{N+1-i}}}{{\rm e}^{x_{N+1-k}}} \\
 \hphantom{u(x_{N+1-k})}{}
 = \sum_{1 \le i < k} \frac{\sqrt{Z_{i} W_{i-1}}}{U_{i} U_{i-1}} \frac{\sqrt{Z_k/W_{k-1}}}{\sqrt{Z_i/W_{i-1}}}
 + \frac{\sqrt{Z_{k} W_{k-1}}}{U_{k} U_{k-1}}
 + \sum_{k < i \le N} \frac{\sqrt{Z_{i} W_{i-1}}}{U_{i} U_{i-1}} \frac{\sqrt{Z_i/W_{i-1}}}{\sqrt{Z_k/W_{k-1}}}
 \\
 \hphantom{u(x_{N+1-k})}{} = \frac{\sqrt{Z_k}}{\sqrt{W_{k-1}}} \sum_{1 \le i < k} \frac{W_{i-1}}{U_{i} U_{i-1}}
 + \frac{\sqrt{Z_{k} W_{k-1}}}{U_{k} U_{k-1}}
 + \frac{\sqrt{W_{k-1}}}{\sqrt{Z_k}} \sum_{k < i \le N} \frac{Z_{i}}{U_{i} U_{i-1}}
 \\
\hphantom{u(x_{N+1-k})}{} = \frac{\sqrt{Z_k}}{\sqrt{W_{k-1}}} \sum_{1 \le i < k} \left( \frac{V_{i-1}}{U_{i}} - \frac{V_{i-2}}{U_{i-1}}\right)
 + \frac{\sqrt{Z_{k} W_{k-1}}}{U_{k} U_{k-1}}\\
\hphantom{u(x_{N+1-k})=}{}
 + \frac{\sqrt{W_{k-1}}}{\sqrt{Z_k}} \sum_{k < i \le N} \left( \frac{T_{i}}{U_{i-1}} - \frac{T_{i+1}}{U_{i}}\right) \\
 \hphantom{u(x_{N+1-k})}{} = \frac{\sqrt{Z_k}}{\sqrt{W_{k-1}}} \frac{V_{k-2}}{U_{k-1}} + \frac{\sqrt{Z_{k} W_{k-1}}}{U_{k} U_{k-1}} + \frac{\sqrt{W_{k-1}}}{\sqrt{Z_k}} \frac{T_{k+1}}{U_{k}}
 \qquad \text{(due to telescoping sums)} \\
 \hphantom{u(x_{N+1-k})}{} = \frac{Z_{k} U_{k} V_{k-2} + Z_{k} W_{k-1} + W_{k-1} U_{k-1} T_{k+1}}{\sqrt{Z_{k} W_{k-1}} U_{k} U_{k-1}} ,
 \end{gather*}
 where the numerator is
 \begin{gather*}
 Z_{k} (U_{k} V_{k-2} + W_{k-1}) + W_{k-1} U_{k-1} T_{k+1} = Z_{k} U_{k-1} V_{k-1} + W_{k-1} U_{k-1} T_{k+1} \\
 \qquad{} = U_{k-1} (Z_{k} V_{k-1} + W_{k-1} T_{k+1}) \\
\qquad{} = U_{k-1} \bigl( (T_{k} U_{k} - T_{k+1} U_{k-1} ) V_{k-1} + (U_{k-1} V_{k-1} - U_{k} V_{k-2}) T_{k+1} \bigr) \\
\qquad {} = U_{k-1} U_{k} (T_{k} V_{k-1} - T_{k+1} V_{k-2}) = U_{k-1} U_{k} Y_{k} ,
 \end{gather*}
 as desired.

 \begin{figure}[t] \centering

 \begin{tikzpicture}
 \node[anchor=south west,inner sep=0] at (0,0)
 {\includegraphics[width=127mm]{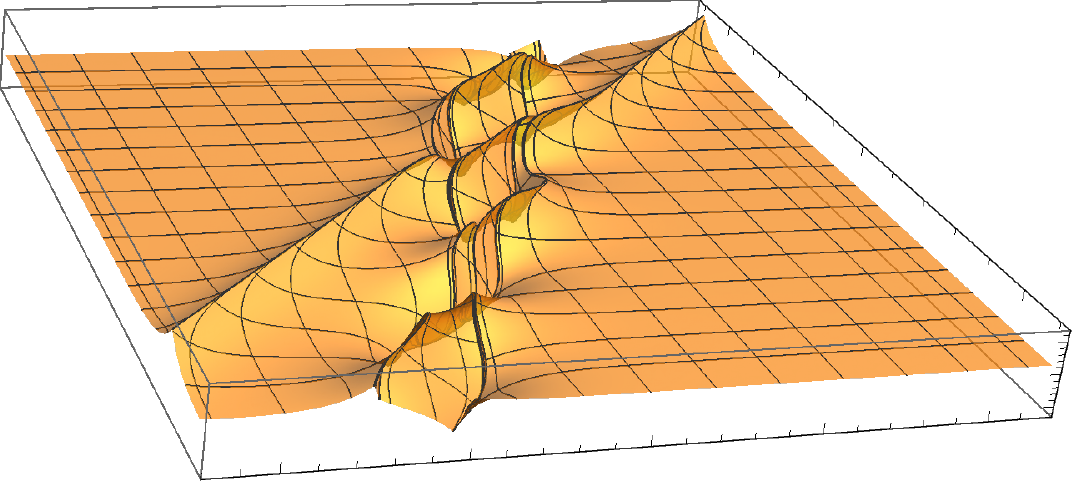}};

 \draw (7.65,0.3) node {$x$};
 \draw (5.5,0.2) node {\scriptsize $-5$};
 \draw (9.7,0.5) node {\scriptsize $5$};

 \draw (10.3,4.2) node {$t$};
 \draw (12.3,2.5) node {\scriptsize $-10$};
 \draw (8.8,5.55) node {\scriptsize $10$};

 \draw (12.75,1.4) node {$u$};
 \end{tikzpicture}

 \caption{A conservative peakon--antipeakon solution of Novikov's equation,
 with $N=5$ and spectral data \eqref{eq:Novikov-1p-4cluster-data},
 as described in Example~\ref{ex:Novikov-1p-4cluster}.
 A peakon with velocity~$1$ interacts with a breather-like cluster of four peakons/antipeakons
 which velocity~$1/2$.
 The peakon trajectories $x=x_k(t)$ are shown in Fig.~\ref{fig:Novikov-1p-4cluster}.
 The dimensions of the box are $\abs{x} \le 12$, $\abs{t} \le 12$, $\abs{u} \le 1.2$. } \label{fig:Novikov-1p-4cluster-wave}
\end{figure}
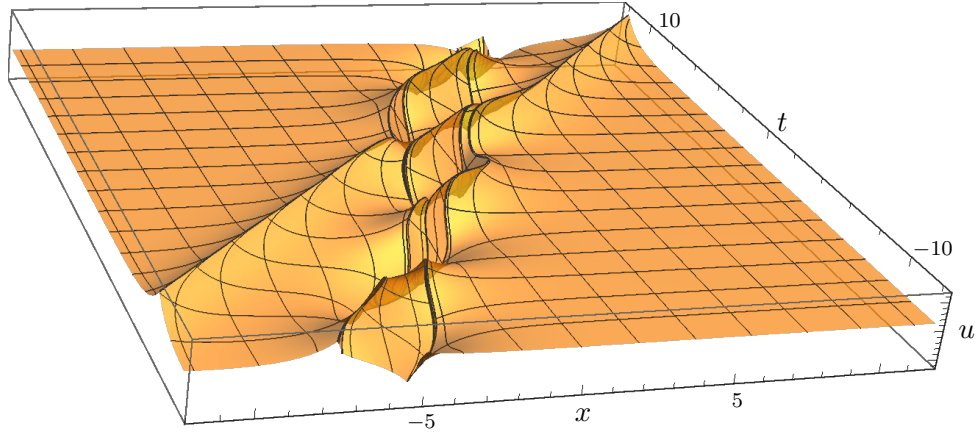

Regarding \eqref{eq:Novikov-ughost} we proceed as in the proof of Theorem~\ref{thm:Novikov-ghost}, i.e., we consider the $(N+1)$-peakon solution and make again the same substitutions as in that proof, so that the peakon at $x_{N+1-p}$ will turn into a ghostpeakon as $\varepsilon \to 0$. For the $(N+1)$-peakon solution, we have just showed that
 \begin{gather*}
 u(x_{N+1-p}) = u(x_{(N+1)+1-(p+1)}) = \frac{\tilde Y_{p+1}}{\sqrt{\tilde Z_{p+1} \tilde W_{p}}} ,
 \end{gather*}
where we know from the previous proof that
 \begin{gather*}
 \tilde Z_{p+1} = Z_{p+1} + Z_{p} \Theta + \order{\varepsilon} ,\qquad \tilde W_{p} = W_{p} + W_{p-1} \Theta + \order{\varepsilon} ,
 \end{gather*}
 and in the same way we derive
 \begin{gather*}
 \tilde Y_{p+1} = \begin{vmatrix}
 \tilde T_{p+1} & \tilde V_{p-1} \\
 \tilde T_{p+2} & \tilde V_{p}
 \end{vmatrix}
 =
 \begin{vmatrix}
 \tilde U_{p+1}^{-1} & \tilde U_{p-1}^1 \\
 \tilde U_{p+2}^{-1} & \tilde U_{p}^1
 \end{vmatrix} \\
 \hphantom{\tilde Y_{p+1}}{} =
 \begin{vmatrix}
 U_{p+1}^{-1} + U_{p}^{-1} \sqrt{\Theta} \varepsilon ( 1 + \order{\varepsilon} )
 &
 U_{p-1}^1 + U_{p-2}^1 \sqrt{\Theta} \varepsilon ( 1 + \order{\varepsilon} )
 \\[1ex]
 U_{p+2}^{-1} + U_{p+1}^{-1} \sqrt{\Theta} \varepsilon^{-1} ( 1 + \order{\varepsilon} )
 &
 U_{p}^1 + U_{p-1}^1 \sqrt{\Theta} \varepsilon^{-1} ( 1 + \order{\varepsilon} )
 \end{vmatrix}
 \\
 \hphantom{\tilde Y_{p+1}}{} = Y_{p+1} + \Theta Y_{p} + \order{\varepsilon} .
\end{gather*}
Letting $\varepsilon \to 0$ and relabelling as in Corollary~\ref{cor:Novikov-add-a-ghostpeakon}, we thus obtain~\eqref{eq:Novikov-ughost}.
\end{proof}

\begin{figure}[t] \centering

\includegraphics{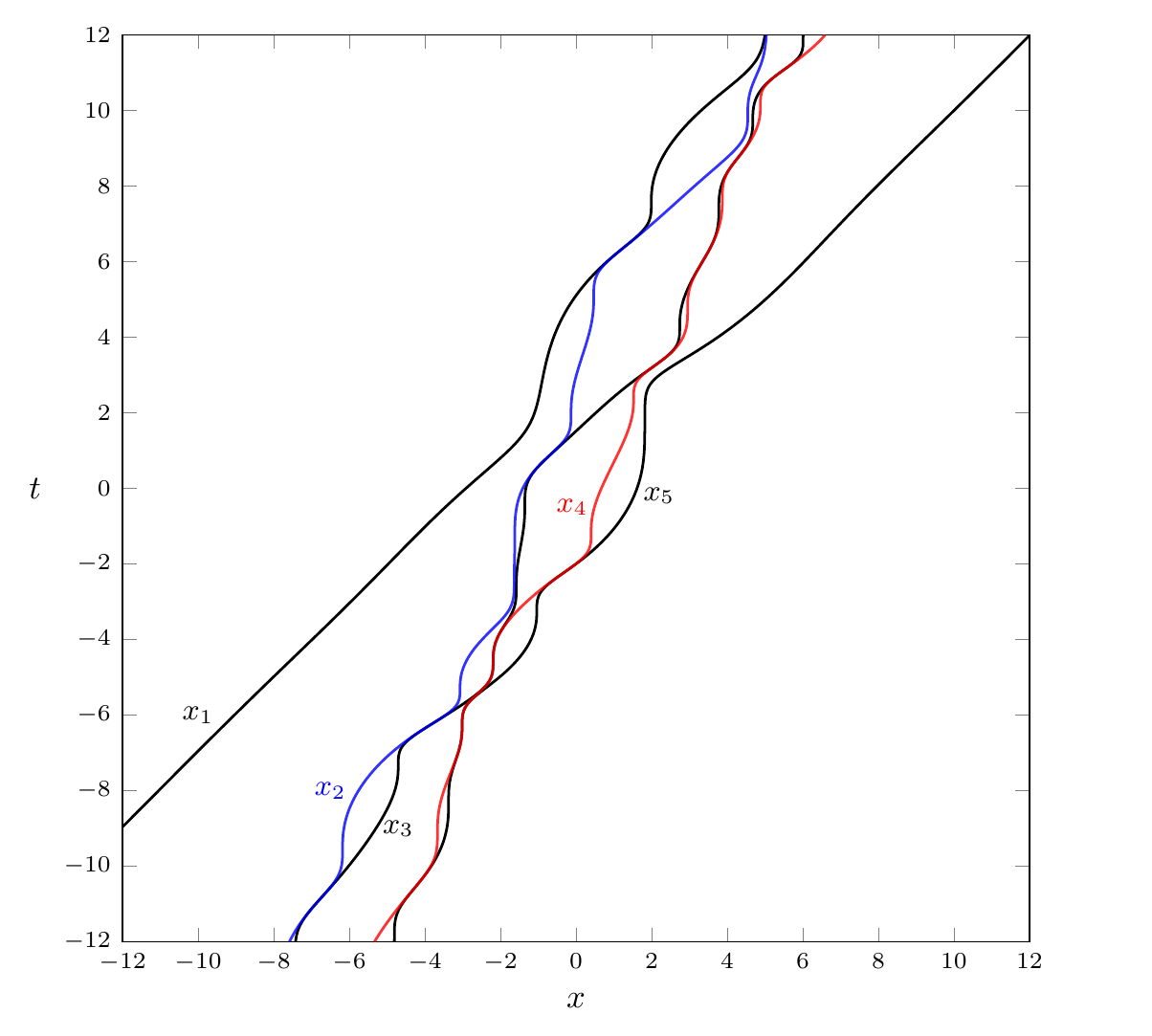}

\caption{Peakon trajectories $x=x_k(t)$ for the $N=5$ conservative peakon--antipeakon Novikov solution shown in Fig.~\ref{fig:Novikov-1p-4cluster-wave}. The curve $x=x_2(t)$ is shown in blue and $x=x_4(t)$ in red. The strict ordering $x_1 < x_2 < x_3 < x_4 < x_5$ holds for almost all~$t$, the exceptions being the isolated (but infinitely many) instants $t_c$ when a peakon--antipeakon collision occurs: $x_k(t_c) = x_{k+1}(t_c)$. Near a collision, the curves stay very close together, and therefore they appear to overlap in the figure; in fact $x_{k+1}(t)-x_k(t)$ is approximately a positive constant times $(t-t_c)^4$. See Example~\ref{ex:Novikov-1p-4cluster}.} \label{fig:Novikov-1p-4cluster}
\end{figure}

\begin{figure}[t] \centering

\includegraphics{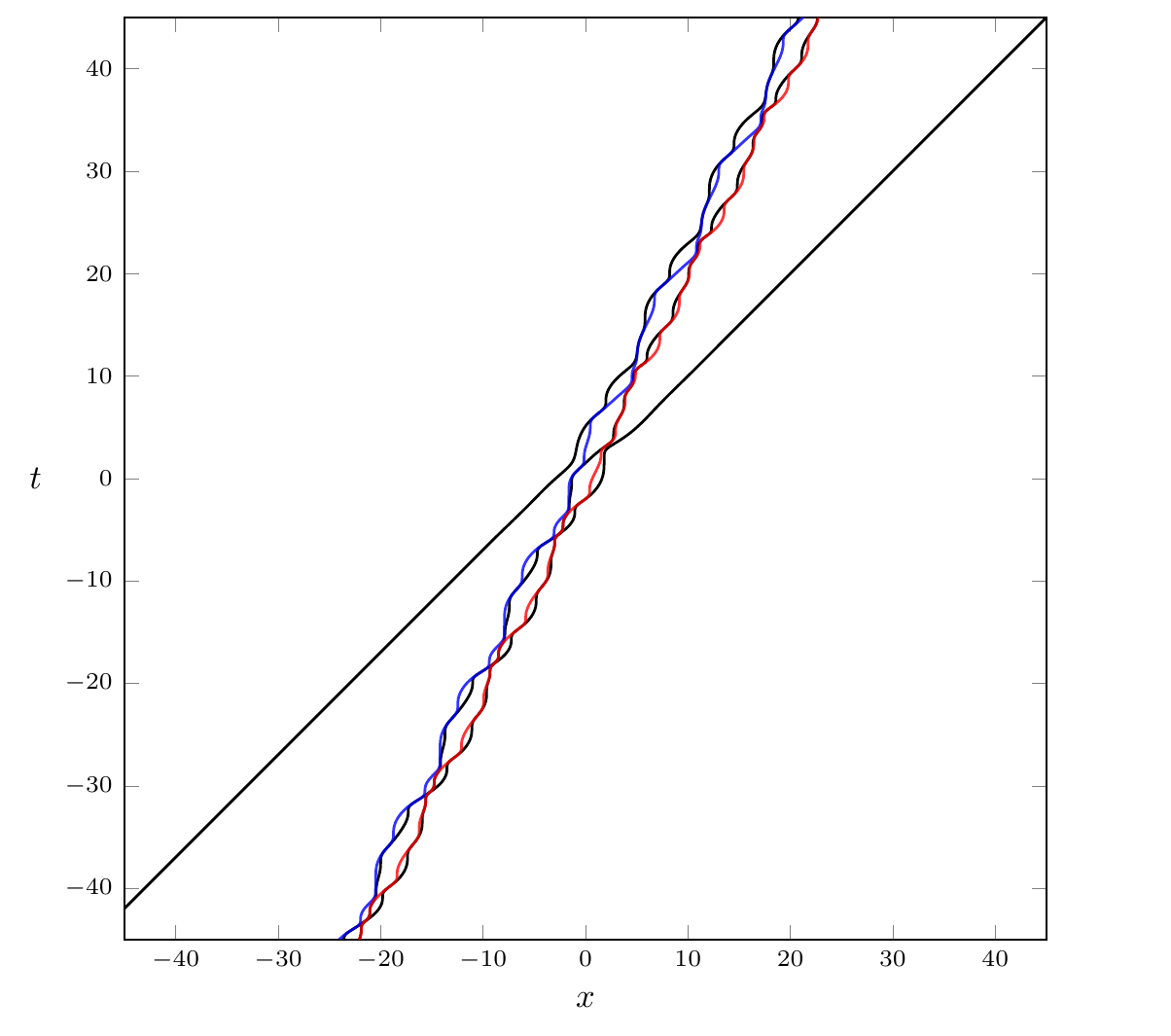}

\caption{The same as Fig.~\ref{fig:Novikov-1p-4cluster}, but on a larger domain. Note that the four-peakon cluster which emerges after the interaction with the single peakon is not merely a shifted version of the cluster seen before the interaction -- the pattern of oscillation within the cluster has also changed.} \label{fig:Novikov-1p-4cluster-zoomout}
\end{figure}

\begin{Example} \label{ex:Novikov-1p-4cluster} As an illustration, consider the solution of the $N=5$ Novikov peakon ODEs which is given by the explicit formulas \eqref{eq:Novikov-generalsolution} with the complex spectral data
\begin{gather} \label{eq:Novikov-1p-4cluster-data}
 \lambda_1 = \overline{\lambda_2} = \frac{1}{\tfrac12 + \tfrac12 {\rm i}} ,\qquad \lambda_3 = \overline{\lambda_4} = \frac{1}{\tfrac12 + {\rm i}} ,\qquad \lambda_5 = 1 ,\qquad b_k(0) = 1, \qquad 1 \le k \le 5 .
 \end{gather}
This is a peakon--antipeakon solution, where $m_k(t)$ and $m_{k+1}(t)$ blow up at those instants $t=t_c$ when a collision $x_k(t_c) = x_{k+1}(t_c)$ occurs. The function $u(x,t) = \sum\limits_{i=1}^N m_i(t) {\rm e}^{-|x-x_i(t)|}$ extends by continuity to a globally defined conservative weak solution, which is given by the formulas in Theorem~\ref{thm:Novikov-u-and-ughost} without any singularities, since the troublesome factor $U_{k-1} U_{k}$ was cancelled in the proof. The graph $u=u(x,t)$ is plotted from those formulas in Fig.~\ref{fig:Novikov-1p-4cluster-wave}, while the peakon trajectories $x=x_k(t)$ are shown in Fig.~\ref{fig:Novikov-1p-4cluster} and on a larger domain in Fig.~\ref{fig:Novikov-1p-4cluster-zoomout}.

For $1 \le k \le 4$, the reciprocal eigenvalues $1/\lambda_k$ all have real part~$1/2$, which means that as $t \to \pm \infty$ one will see a cluster of four peakons/antipeakons travelling together, with the positions $x_k(t)$ displaying oscillations superimposed on an overall drift with velocity~$1/2$. The peakons and antipeakons in the cluster interact in a breather-like manner, with infinitely many collisions taking place, in this case (asymptotically) periodically since the frequencies $\imag(1/\lambda_1)=\tfrac12$ and $\imag(1/\lambda_3)=1$ are commensurable, but quasi-periodic oscillations are also possible.

At collisions, the peakon and the antipeakon approach each other very closely; for the Novikov equation, the leading term in $x_{k+1}(t)-x_k(t)$ as $t \to t_c$ is a positive constant times $(t-t_c)^{4j}$, where $j=1$ (fourth power) is the generic case, although ``higher-order collisions'' with $j > 1$ are also possible. In contrast, for Camassa--Holm collisions the leading term is always a positive constant times the second power $(t-t_c)^2$, as shown by Beals, Sattinger and Szmigielski~\cite{beals-sattinger-szmigielski:2000:moment}.

Since $1/\lambda_5=1$, there will also be a lone peakon with asymptotic velocity~$1$. As seen in Fig.~\ref{fig:Novikov-1p-4cluster-zoomout}, the periodic pattern of the four-peakon cluster as $t \to +\infty$ (after interaction with the single peakon) does not look the same as when $t \to -\infty$ (before the interaction), so the peakon scattering process involves more than simply a phase shift in this case. Indeed, merely writing down the precise asymptotic behaviour of this five-peakon solution requires knowing the exact formulas for the four-peakon solution. We refer to Kardell and Lundmark~\cite{kardell-lundmark:2016p:novikov-peakon-antipeakon} for details.
\end{Example}

\subsection*{Acknowledgements}

This work has been in the making for a long period, during parts of which Hans Lundmark was supported by the Swedish Research Council (Veten\-skaps\-r{\aa}det, grant 2010-5822) and Budor Shuaib by the Libyan Higher Education Ministry. We are also grateful to the Department of Mathematics at Link\"oping University for financial support. Krzysztof Marciniak has been most helpful, providing many questions and comments which have improved the readability of the article tremendously. And finally we thank the referees for valuable input, especially regarding Remark~\ref{rem:Matsuno}.

\newpage
\pdfbookmark[1]{References}{ref}
\LastPageEnding

\end{document}